\documentclass[%
 aip,
 amsmath,amssymb,
 reprint,%
]{revtex4-1}

\usepackage{graphicx}
\usepackage{dcolumn}
\usepackage{bm}
\usepackage{physics}
\usepackage{amsfonts}
\usepackage{amssymb}
\usepackage{amsthm}
\usepackage{amsmath}
\usepackage{mathtools}
\usepackage{hyperref}
\usepackage{xcolor}
\usepackage[utf8]{inputenc}
\usepackage[T1]{fontenc}
\usepackage{mathptmx}

\newtheorem{definition}{Definition}
\newtheorem{theorem}{Theorem}

\begin{document}

\preprint{AIP/123-QED}

\title[Nonclassical Light and Metrological Power: An Introductory Review ]{Nonclassical Light and Metrological Power: An Introductory Review}

\author{Kok Chuan, Tan}
 \email{bbtankc@gmail.com}

\author{Hyunseok, Jeong}
 \email{h.jeong37@gmail.com}
 
\affiliation{%
Quantum Information Science Group \& Institute of Applied Physics, Department of Physics and Astronomy, Seoul National University, Seoul, 08826, Korea
}%

\date{\today}

\begin{abstract}
In this review, we introduce the notion of quantum nonclassicality of light, and the role of nonclassicality in optical quantum metrology. The first part of the paper focuses on defining and characterizing the notion of nonclassicality and how it may be quantified in radiation fields. Several prominent examples of nonclassical light is also discussed. The second part of the paper looks at quantum metrology through the lens of nonclassicality. We introduce key concepts such as the Quantum Fisher information, the Cram{\'e}r-Rao bound, the standard quantum limit and the Heisenberg limit, and discuss how nonclassical light may be exploited to beat classical limitations in high precision measurements. The discussion here will be largely theoretical, with some references to specific experimental implementations.
\end{abstract}

\maketitle

\tableofcontents

\section{Introduction}

Being an empirical science, our ability to understand nature through physics is deeply tied to our ability to measure things. Needless to say, the study of measurements, or metrology, is a foundational aspect of physics and indeed, all of the other natural sciences. Classical physics imposes certain natural limits, not related to the skill or the ingenuity of the observer, on our ability to perform precise measurements on physical systems. One of the crowning achievements of modern day quantum mechanics is the realization, and ultimate verification, that such limitations that apply to classical systems do not in fact extend over to quantum ones. The quantum regime therefore supplies us with a new bag of tricks, thus allowing us to look ever deeper into the inner workings of nature. The ultimate hope is that by doing so, the next step  forward in our understanding will be revealed.

This review is intended to introduce the topic of nonclassicality in light fields, with an eye on their applications in ultra high precision measurements. In classical mechanics, the primary limitation imposed on our ability to make precision measurements comes from energy. From another point of view, energy may be considered as being converted to measurement precision. As we go through the arguments, we will see that quantum mechanics provides another avenue. With the same amount of energy, it is possible to achieve levels of precision in the quantum regime that is orders of magnitude higher than what is possible in the classical regime. This level of precision is contingent on our ability to produce highly nonclassical states. In other words, quantum nonclassicality itself can be converted to measurement precision, thus presenting us with an alternate path towards achieving higher precision. Producing a nonclassical state and extracting metrological usefulness from it is by no means a trivial task, but at least this is only limited by our current techniques and ingenuity, rather than by any natural constraint. The field of quantum metrology, in the broadest terms, essentially concerns itself with coming up with ever more inventive ideas to (i) produce useful nonclassical states, and (ii) extract useful metrological content from them. Many of these ideas have seen applications in areas such as quantum information\cite{Braunstein2005}, biology\cite{Taylor2016} and imaging\cite{Berchera2019}.

Research into nonclassicality and metrology spans nearly six decades of continuous scientific progress, and covering all aspects of these two topics will go far beyond the scope and ambitions of this paper. Instead, the contents of this paper is intended to be a curated view of the subject focussing on what the authors feel are key developments.

This paper is mainly split into two parts.  In Section~\ref{sec::part1}, we will mainly discuss the notion of nonclassicality, how it may be defined and how it may be characterized, and provide examples of such nonclassical states of light. A survey of various approaches of quantifying nonclassicality in light is performed. In Section~\ref{sec::part2} the concept of metrological power will be discussed, where we loosely interpret metrological power as any metrological advantage that can be attributed solely to the nonclassicality of the state. We introduce concepts such as the quantum Fisher information and the Cram{\'e}r-Rao bound, and discuss scenarios where nonclassicality may be leveraged to surpass classical limits. We also briefly touch upon methods of generating nonclassical light. 

We hope that through the course of the ensuing discussions,
the interested reader will be able to develop an overall feel
for the subject and be sufficiently equipped to initiate a research direction of their own. Let us start by discussing what classicality means within the context of quantum optics.

\section{Classical and nonclassical light} \label{sec::part1}

\subsection{Defining classicality in quantum mechanics} \label{sec::classicality}

A more traditional treatment of classical light will begin with a description of electromagnetic fields using classical electrodynamics, which is then compared to the quantum regime when the field is subsequently quantized\cite{Grynberg2010}. This approach, while chronologically respecting the way quantum mechanics was developed,  slightly misrepresents the relationship between the classical and quantum regimes by suggesting that quantum electrodynamics somehow emerges from classical electrodynamics. The actual relationship is in fact much closer to the opposite. There is in fact no such thing as a classical system, much less a classical system that is "quantized". As far as we can tell, the whole of nature is quantum mechanical, so it is far more appropriate to say that classical physics emerges from quantum mechanics rather than the other way round.

We will therefore begin with the quantum description of light, which is more in line with the modern approach. 

As with all quantum systems, the dynamics of light is governed by the Hamiltonian. We consider the simplest possible representation of the Hamiltonian for a single mode of light with frequency $\omega$. In this case, the Hamiltonian takes on the form $$H=\hbar \omega (a^\dag a+\frac{1}{2}).$$ The operators $a$ and $a^\dag$ are called annihilation and creation operators, and they satisfy the fundamental commutation relation $\comm{a}{a^\dag}=1$. For convenience, we assume that $\hbar = \omega = 1$, such that $H=a^\dag a+\frac{1}{2}$. 

One may recall that this Hamiltonian is identical to the one describing a quantum harmonic oscillator. Indeed, one may define analogous position position and momentum operators, also called the $x$ and $p$ quadratures, as $x \coloneqq \frac{1}{\sqrt{2}}(a+a^\dag)$ and $p \coloneqq \frac{1}{\sqrt{2} i}(a-a^\dag)$ such that $\comm{x}{p}=i$ and equivalently write $H=\frac{1}{2}(x^2+p^2)$. This makes the connection between the quantum description of light and the quantum Harmonic oscillator explicit. Note that despite the notation $x$ and $p$, they  do not correspond to the actual physical position and momentum of light fields. They do however, provide us with a definition of phase space coordinates $(x,p)$ which we can then use to study quantum light. Physically, they can be interpreted as coordinates on a phasor diagram\cite{Grynberg2010} (see Fig.~\ref{fig::phasor}). In this picture $p$ and $x$ coordinates respectively corresponds to the amplitude of  wave components that are in phase and $\pi/2$ out of phase with respect to a given reference. These coordinates can be sampled in the laboratory via homodyne measurements\cite{Yuen1980, Yuen1983,Schumaker1984, Yurke1987, Vogel2006}.

\begin{figure}
\includegraphics[width = 0.8\linewidth ]{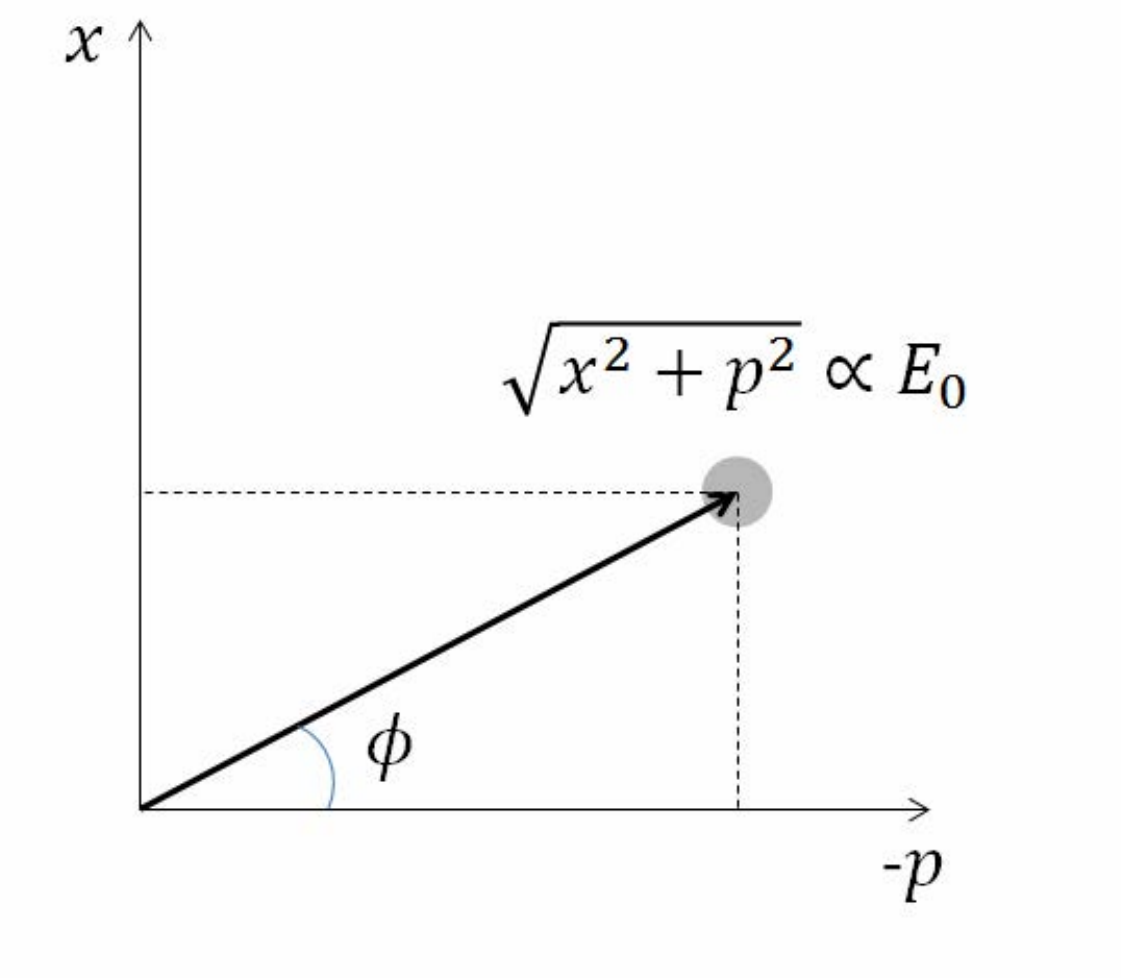}
\caption{\label{fig::phasor} Interpreting $(x,p)$ phase space coordinates as phasor diagrams. In classical mechanics, a point vector $(x,p)$ represents a plane wave with relative phase $\phi$ and electric field amplitude $E_0$ which is proportional to the magnitude $\sqrt{x^2+p^2}$. In quantum mechanics, it is not possible to represent a state with a single point, due to the uncertainty principle.}
\end{figure}


Given the Hamiltonian $H$, one may further show that its eigenstates are the well known Fock or number states $\ket{n}$. As far back as a century ago, Planck\cite{Planck1901} and Einstein\cite{Einstein1905} already demonstrated the existence of individual photons. The Fock states $\ket{n}$ are basically quantum descriptions of a single mode of light containing $n$ photons. One may  show that the Fock states, annihilation and creation operators satisfy the following elementary properties: \begin{align*}
a^\dag a \ket{n} &= n \ket{n} \\
a^\dag\ket{n} &= \sqrt{n+1}\ket{n+1} \\
a\ket{n} &= \sqrt{n} \ket{n-1} \\
\ket{n} &= \frac{a^\dag}{\sqrt{n!}}\ket{0}\\
\end{align*}

We have not yet arrived at our desired definition of classical light. Indeed, we see that the Fock states, despite emerging naturally as eigenstates of the Hamiltonian, clearly does not possess the requisite qualities of being classical, since individual photons were not suspected up until the advent of quantum mechanics. 

At the beginning of the section, it was mentioned that there is no such thing as a truly classical system. As far as we can tell, every physical system obeys quantum rules, so ``classical" light does not actually exist in the strictest sense of the word.

One can, however, make a compelling case that certain quantum states possess classical properties, at least more so that other quantum states. This can be done without ever leaving the quantum mechanical framework. The name "classical light" is therefore somewhat of a misnomer. They are actually quantum states of light that fully obeys quantum mechanical laws, but can be argued to possess properties that are closest in nature to what we traditionally see in a classical system.

One distinctive feature of classical physics is that classical systems can be described by a point in phase space. Within the context of light, this means that at any given point in time, one may specify both the electric field amplitude as well as the phase with perfect precision (See Fig.~\ref{fig::phasor}). We know that this is impossible in quantum physics, due to the well known Heisenberg uncertainty relation\cite{Heisenberg1927, Robertson1929} $$\Delta x \, \Delta p \geq 1/2,$$ where $\Delta O \coloneqq \sqrt{\langle O^2\rangle-\langle O \rangle^2}$ for any observable $O$. We know that every physical state must obey the uncertainty principle, so the question we should be asking is: among the states obeying the uncertainty principle, what kinds of states permits a description closest to a point particle in phase space? If we can find such a class of states, we will call these states ``classical" in the quantum mechanical sense of the word.

\subsection{Coherent states as classical states of light} \label{sec::cohStates}

If our starting point for the definition of classicality is how closely the quantum mechanical description resembles a point in phase space, then the answer to the previous question is clear. A classical state should be a minimal uncertainty state such that $$\Delta x \, \Delta p = 1/2.$$ Furthermore, classical dynamics treats both position and momentum variables on equal footing, so we should also have $$\Delta x = \Delta p=1/\sqrt{2}.$$

It turns out that only one class of quantum states satisfy the above constraints under quantum mechanics, and it is the class of coherent states\cite{Glauber1963}. These states were considered as far back as 1926 by Schr{\" o}dinger as special solutions to the harmonic oscillator problem\cite{Schrodinger1926}, but their relationship to quantum light was greatly expanded much later by Glauber\cite{Glauber1963} and Sudarshan\cite{Sudarshan1963}. They now play a foundational role in the field of quantum optics.

The set of coherent states may be defined as the set of eigenstates of the annihilation operator $a$, such that $$a \ket{\alpha} = \alpha \ket{\alpha}, $$ where $\alpha$ is in general a complex number. 

Another way to define coherent states is by first defining the displacement operator $$D(\alpha) \coloneqq e^{\alpha a^\dag - \alpha^* a}.$$ You then generate the set of coherent states by performing a displacement operation on vacuum: $$D(\alpha)\ket{0} = \ket{\alpha}.$$

They are called displacement operators because for any state $\ket{\psi}$, the map $\ket{\psi} \rightarrow D(\alpha)\ket{\psi}$ is equivalent to the maps $x \rightarrow x + \sqrt{2}\Re(\alpha)$ and $p \rightarrow p + \sqrt{2}\Im(\alpha)$. This is essentially a linear displacement on phase space coordinates $(x,p)$.

One may show that coherent states $\ket{\alpha}$ and the displacement operators $D(\alpha)$ obey the following set elementary properties:

\begin{gather*}
\ket{\alpha} = e^{-\abs{\alpha}^2/2} \sum^\infty_{n=0} \frac{\alpha^n}{\sqrt{n!}} \ket{n} \\ 
D^\dag(\alpha)D(\alpha) = \openone \\ 
D(\alpha+\beta) = D(\alpha)D(\beta)e^{-i\Im(\alpha \beta^*)} \\
D^\dag(\alpha) \, a \, D(\alpha) = a+\alpha \\
\langle n \rangle = \bra{\alpha}a^\dag a \ket{\alpha} = \abs{\alpha}^2 \\
\braket{\beta}{\alpha} = e^{-\abs{\beta}^2/2 -\abs{\alpha}^2/2 + \beta^*\alpha} \\
\bra{\alpha} x \ket{\alpha} = \sqrt{2}\Re(\alpha) \\ 
\bra{\alpha} p \ket{\alpha} = \sqrt{2}\Im(\alpha) \\
\frac{1}{\pi} \int \dd[2]\alpha \ket{\alpha}\bra{\alpha} = \openone \\
\end{gather*} 

Based on the above properties, one may then directly calculate that for coherent states, $\Delta x = \Delta p = 1/\sqrt{2}$ and that $\Delta x \, \Delta p = 1/2$ as required. We therefore established that the set of coherent states satisfies the requirements that were laid out at the beginning of this section. They can therefore be considered classical in this sense. 

Furthermore, one may also show that they are the only set of pure quantum states that can make this claim. In order to see this, we note that using the property $D^\dag(\alpha) \, a \, D(\alpha) = a+\alpha$, we can verify that $\Delta x$ and $\Delta p$ are invariant under displacement operations, regardless of the initial state $\ket{\psi}$. We can therefore displace any state such that it satisfies $\langle x \rangle = \langle p \rangle = 0$, which we will assume is satisfied without loss in generality. For such states, the variance is just given by $\Delta^2 x = \langle x^2 \rangle$ and $\Delta^2 p =\langle p^2 \rangle$. We recall that $H= a^\dag a+\frac{1}{2} = \frac{1}{2}(x^2+p^2)$, which leads to the following series of equations: 

\begin{align*}
\langle a^\dag a \rangle + \frac{1}{2} &= \frac{1}{2}(\langle x^2 \rangle+ \langle p^2 \rangle) \\
&= \frac{1}{2}(\Delta^2 x + \Delta^2 p) \\
&= \frac{1}{2},
\end{align*} where in the last line, we substituted in the classicality requirement $\Delta x= \Delta p = \frac{1}{\sqrt{2}}$. Clearly, this requires $\langle a^\dag a \rangle = 0$. Since $n \coloneqq a^\dag a$ is just the photon number operator, the only state with zero photons is the vacuum $\ket{0}$. As such, the vacuum state $\ket{0}$, up to a displacement operator, is the unique minimum uncertainty state satisfying $\Delta x= \Delta p = \frac{1}{\sqrt{2}}$. Since a displaced vacuum defines the set of coherent states, coherent states are the unique set of pure states that can be considered classical under our current definition.

That coherent states may be considered the most classical quantum states is further supported when we consider the dynamics of the system. The evolution of the state is completely described by the unitary evolution $U(t)=e^{-iHt}=e^{-i(a^\dag a +1/2)t}$. Under such dynamics, the coherent states evolves according to 
\begin{align*}
U(t)\ket{\alpha} &= e^{-\abs{\alpha}^2/2} \sum^\infty_{n=0} \frac{\alpha^n}{\sqrt{n!}} e^{-i(a^\dag a +1/2)t}\ket{n} \\
&= e^{-\abs{\alpha}^2/2} \sum^\infty_{n=0} \frac{\alpha^n}{\sqrt{n!}} e^{-i(n +1/2)t}\ket{n} \\
&= e^{-it/2}e^{-\abs{\alpha}^2/2} \sum^\infty_{n=0} \frac{(e^{-it}\alpha)^n}{\sqrt{n!}} \ket{n} \\
&= e^{-it/2} \ket{e^{-it}\alpha}.
\end{align*}
We see that coherent states are rotated in complex parameter space by the phase factor $e^{-it}$, but otherwise remain as coherent states under free time evolution. Since $\langle x \rangle = \Re(e^{-it}\alpha)$ and $\langle p \rangle = \Im(e^{-it}\alpha)$, we see that the time evolution in phase space is described by to an clockwise rotation along a circle with radius $\abs{\alpha}$. If we were to compute the wavefunction $\psi(x,t)$ by projecting the state onto the eigenstates $\ket{x}$ of the $x$ quadrature, we can verify that the probability density at each time $t$ is just a Gaussian wavepacket $$\abs{\psi(x,t)}^2 = \sqrt{\frac{1}{\pi}}e^{-(x - \sqrt{2}\Re[\exp(-it)\alpha] )^2},$$ where we assumed $m=\hbar = \omega =1$ for the parameters of the harmonic oscillators. We see that the probability density is just an oscillating Gaussian wavepacket at every time $t$. This dynamical behaviour is similar to what we would expect from a classical harmonic oscillator, except with a point particle replaced by a wavepacket, so the dynamics of coherent states are also similar to a classical system. 

Another strong argument that suggests that coherent states are classical comes from the physical systems that they represent. A coherent state $\ket{\alpha}$ is the quantum mechanical representation of coherent, monochromatic light source whose electric field amplitude is proportional to $\abs{\alpha}$, and relative phase is specified by $\arg(\alpha)$. Notwithstanding the fact that its working mechanism relies on quantum mechanics, the output of a laser source is typically considered to be close to an ideal classical light source: i.e. it is a source of strongly coherent, monochromatic light. The coherent state describes the output of a laser operating high above its threshold very well, although there had been some controversy on the theoretical side as to whether the output of laser can be safely assumed to be a coherent state\cite{Molmer1997, Rudolph2001, Enk2001, Wiseman2003}.

Thus far, we have only considered pure states. More generally, mixed quantum states can be represented via density operators which are statistical mixtures of pure states of the form $\rho = \int \dd x \, p(x) \ket{\psi(x)}\bra{\psi(x)}$ where $\int \dd x \, p(x) = 1$. Since we have already ascertained what states are the most classical among the pure quantum states, the generalization to mixed states is relatively straightforward. We consider any statistical mixture of pure classical states to also be classical. That is, if a density operator can be expressed in the form $$\rho = \int \dd[2]\alpha P_{\mathrm{cl}}(\alpha) \ket{\alpha}\bra{\alpha},$$ where $\int \dd[2]\alpha P_{\mathrm{cl}}(\alpha)=1$ is some positive probability density function, then we say that the quantum state is classical.

\subsection{Defining nonclassicality via the Glauber-Sudarshan $P$-function} \label{sec::nonclassicality}

So far, we have considered which states among the set of quantum states are considered the most classical. Based on this, nonclassical states may be defined almost immediately. By definition, any quantum state that is not classical, must be nonclassical. In terms of density operators, this means that nonclassical states are states which cannot be expressed in the form $\rho = \int \dd[2]\alpha P_{\mathrm{cl}}(\alpha) \ket{\alpha}\bra{\alpha}$ using some positive probability density function $ P_{\mathrm{cl}}(\alpha)$.

This definition of nonclassicality is however not necessarily the most natural one to adopt, as it does not suggest a method, analytical or otherwise, of determining whether a positive probability density function $P_{\mathrm{cl}}(\alpha)$ exists for an arbitrary mixed state $\rho$. 

A more natural definition of nonclassicality is possible if one moves away from the density operator representation of a quantum state. An alternative representation of a quantum state comes from the seminal work of Glauber\cite{Glauber1963} and Sudarshan\cite{Sudarshan1963}, who observed that any quantum state of light can be written in the form $$\rho = \int \dd[2]\alpha \, P(\alpha) \ket{\alpha}\bra{\alpha},$$ where $P(\alpha)$ is called the Glauber-Sudarshan $P$-function. Note the formal similarity to the definition of a classical state $\rho = \int \dd[2]\alpha P_{\mathrm{cl}}(\alpha) \ket{\alpha}\bra{\alpha}$. The key difference is that $P(\alpha)$ is a quasiprobability instead of a positive probability density function. This means that $P(\alpha)$ always is always normalized such that $\int \dd[2]\alpha \; P(\alpha) = 1$, but may permit negative values. When  $P(\alpha)$ does correspond to a positive probability density function however, we immediately see that the state must be classical. This leads to the following definition of nonclassicality.

\begin{definition}[Nonclassical states of light] \label{def::nonclassicality}
A quantum state of light is nonclassical iff its Glauber-Sudarshan $P$-function is not a positive probability density function. 
\end{definition}

In literature, it is sometimes stated that a state is nonclassical when the $P$-function is {\it negative or more singular than a delta function.} This does not contradict our definition as classical probability density functions do not contain singularities more exotic than delta functions. At the same time, the distinction between negativity and highly singular points is largely a point of technicality, as the existence of highly singular points always implies some notion of negativity\cite{Kiesel2010, Kuhn2018, Tan2019}. For the rest of this paper, we will treat "negative $P$-functions" and "nonclassicality" as basically interchangeable terms. 

The primary benefit of defining nonclassicality with respect to the $P$-function is that it points to a clear method, at least analytically, of determining whether a given state is nonclassical or not. Given some density operator $\rho$, the $P$-function may be computed in the following way. 

\begin{align}
P(\alpha) & = \int \dd[2]\beta P(\beta) \delta(\beta - \alpha) \\
&= \int \dd[2]\beta P(\beta) \frac{1}{\pi^2}\int \dd[2]\gamma e^{- (\gamma \alpha^* -\gamma^*\alpha )}e^{ (\gamma \beta^* - \gamma^* \beta)} \\ 
&= \int \dd[2]\beta P(\beta) \frac{1}{\pi^2}\int \dd[2]\gamma \; e^{\gamma^*\alpha - \gamma \alpha^*  } \bra{\beta} e^{\gamma a^\dag} e^{-\gamma^* a} \ket{\beta} \\ 
&=  \frac{1}{\pi^2}\int \dd[2]\gamma \; e^{\gamma^*\alpha - \gamma \alpha^*  } \Tr[e^{\gamma a^\dag} e^{-\gamma^* a} \int \dd[2]\beta P(\beta) \ket{\beta}\bra{\beta}] \\ 
&=  \frac{1}{\pi^2}\int \dd[2]\gamma \; e^{\gamma^*\alpha - \gamma \alpha^*  } \Tr[e^{\gamma a^\dag} e^{-\gamma^* a} \rho] \label{eq::computeP},
\end{align} where we used the identity $\delta(\beta - \alpha) = \frac{1}{\pi^2}\int \dd[2]\gamma e^{- (\gamma \alpha^* -\gamma^*\alpha )}e^{ (\gamma \beta^* - \gamma^* \beta)}$, which comes from property that the Fourier transform of a constant is proportional to the delta function. Since we can obtain the the $P$-function from the density operator, and the density operator can be retrieved via the identity $\rho = \int \dd[2]\alpha \, P(\alpha) \ket{\alpha}\bra{\alpha}$, they are equivalent representations of the quantum state. To decide whether or not a state is nonclassical however, one just has to determine whether $P(\alpha)$ displays any negativities.

At this juncture, it is also worth mentioning that $P$-functions are not the only quasiprobability distributions considered in quantum optics\cite{Cahill1969a, Cahill1969b}. Let us consider the previously derived expression $P(\alpha) =\frac{1}{\pi^2}\int \dd[2]\beta \; e^{\beta^*\alpha - \beta \alpha^*  } \Tr[e^{\beta a^\dag} e^{-\beta^* a} \rho]$. From the Baker-Campbell-Hausdorff formula, we have $e^{\beta a^\dag} e^{-\beta^* a} = e^{\beta a^\dag-\beta^* a+\abs{\beta}^2/2} = D(\beta)e^{\abs{\beta}^2/2}$. This leads to the simplified expression $$P(\alpha) =\frac{1}{\pi^2}\int \dd[2]\beta \; e^{\beta^*\alpha - \beta \alpha^*  } \Tr[D(\beta) \rho] e^{\abs{\beta}^2/2}.$$ From this, we observe that that the above expression is actually just the Fourier transform of the characteristic function $\Tr[D(\beta) \rho] e^{\abs{\beta}^2/2}$. One may generalize the characteristic function by adding a real parameter $s$ such that \begin{align} \label{eq::sCharFun}
\chi_s(\beta) \coloneqq \Tr[D(\beta) \rho] e^{s\abs{\beta}^2/2}.
\end{align}

The above is called the $s$-parametrized characteristic function. From the $s$-parametrized characteristic function, one may obtain the $s$-parametrized quasiprobability distribution function by considering the Fourier transform 
\begin{align} \label{eq::sQuasiProb} P_s(\alpha) \coloneqq \frac{1}{\pi^2}\int \dd[2]\beta \; e^{\beta^*\alpha - \beta \alpha^*  } \chi_s(\beta).
\end{align}
For every real value of $s$, we see that $\int \dd[2]\alpha P_s(\alpha) =  \chi_s(0) = 1$, so they are indeed quasiprobabilities. Typically, the range of values $s\in [-1,1]$ is considered. At $s=1$, we retrieve the $P$-function\cite{Glauber1963,Sudarshan1963}, at $s=0$, we obtain the Wigner function\cite{Wigner1932}, and at $s=-1$, we have the Husimi $Q$-function\cite{Husimi1940}. Negativities of the $s$-quasiprobabilties other than the $P$-function have also been previously considered within the context of nonclassicality\cite{Kenfack2004,Tan2019} (see Section~\ref{sec::negVolume}).

 Finally, to conclude this part of the discussion, we would like to mention that it is possible to consider different notions of nonclassicality apart from the one in Definition~\ref{def::nonclassicality}, so long as one justifies it with physical arguments. For instance, one may adopt anti-bunched light, or non-Gaussian light as their notion of nonclassicality. However, as further discussed in Section~\ref{sec::survey}, such definitions can often be viewed as special cases of Definition~\ref{def::nonclassicality}.

\subsection{Examples of $P$-functions} \label{sec::examplesPfun}

In this section, we will mainly discuss several important examples of states with known $P$-functions. These will include several classical states, but the main focus is on $P$-functions that are nonclassical.

\subsubsection{Coherent states}

The simplest $P$-functions are given by the coherent states, which are classical by definition. The $P$-function of a coherent state $\ket{\beta}$ can be written as $$P(\alpha) = \delta(\alpha - \beta),$$ so it is just the delta function. Classical states in general do not contain singularities more exotic than delta functions.

\subsubsection{Thermal states} \label{sec::thermalstates}

The thermal state describes a radiation field in thermal equilibrium with a heat bath at inverse temperature $\beta$. Its density operator has the form $$\rho = \frac{1}{Z}e^{-\beta a^\dag a}=(1-e^{-\beta})\sum_{n\geq 0} e^{-\beta n}\ket{n}\bra{n},$$ where $Z = \Tr(e^{-\beta a^\dag a})$ is the partition function and we assumed that $\hbar = \omega = 1$. The mean photon number of the thermal state is given by $\langle n \rangle = \langle a^\dag a \rangle = (1-e^{-\beta})^{-1}$. 

Using Eq.~\ref{eq::computeP}, one may directly compute the $P$-function of via the density operator, which gives the expression\cite{Schleich2001a} $$P(\alpha) = \frac{1}{\langle n \rangle \pi} e^{-\abs{\alpha}^2/\langle n \rangle}.$$ This is just an isotropic Gaussian distribution with variance $\sigma^2 = \langle n \rangle /2$. Every isotropic Gaussian distribution therefore corresponds to the $P$-function of a thermal state, up to some displacement operation. See Fig.~\ref{fig::SPAT} for a plot of the $P$-function of the thermal state.

\subsubsection{Fock states}

In Section~\ref{sec::classicality}, the Fock states $\ket{n}$ were introduced as the eigenstates of the Hamiltonian $H = a^\dag a +1/2 $, or alternatively, the number operator $n = a^\dag a$. They describe the quantum state of light containing a definite number $n$ of photons\cite{Hofheinz2008, Cirac1993, Varcoe2000,Bertet2002}.

One may show via direct calculation that the $P$-function of Fock states takes the form $$P(\alpha) = L_n\left [ -\frac{1}{4} \left (\pdv[2]{\Re(\alpha)}+ \pdv[2]{\Im(\alpha)}\right ) \right ] \delta(\alpha),$$ where $L_n$ is the $n$th Laguerre polynomial. In general, the $n$th Laguerre polynomial contains powers up to $n$, which suggests that the $P$-function contains derivatives of the delta function up to the $2n$th order. These are more singular than regular delta functions, so the state is nonclassical.

\subsubsection{Squeezed states} \label{sec::squeezedStates}

Together with Fock states, squeezed states\cite{Andrews2014, Lvovsky2016, Loudon1987, Slusher1985, Slusher1987, Kim1994} are perhaps the archetypal examples of nonclassical light. Just like the coherent states, it is a minimal uncertainty state so it satisfies $\Delta x \, \Delta p = 1/2$. Unlike coherent states however, squeezed states do not treat each quadrature equally, such that in general $\Delta x \neq \Delta p. $ This necessarily means that one of the quadratures is ``squeezed", such that, up to a rotation in phase space coordinates, $\Delta x$ or $\Delta p$ is less than $1/\sqrt{2}$. Squeezed states can be defined via the squeeze operator $$S(\epsilon) \coloneqq e^{ (\epsilon^* a^2 - \epsilon a^{\dag 2})/2}.$$ In general, $\epsilon = \abs{\epsilon} e^{i\theta}$ is a complex parameter. One may define the set of squeezed coherent states as $$\ket{\alpha,\epsilon} \coloneqq D(\alpha) S(\epsilon) \ket{0}.$$ Alternatively, one may also define the set of squeezed coherent states as the eigenstates of an operator such that \begin{align*}
&\cosh(\abs{\epsilon})a+e^{i\theta}\sinh(\abs{\epsilon})a^\dag] \ket{\alpha,\epsilon} \\ & \qquad = \left [ \alpha \cosh(\abs{\epsilon})+\alpha^* e^{i\theta}\sinh(\abs{\epsilon})\right ] \ket{\alpha, \epsilon}.
\end{align*}

The squeeze operator and squeezed states has the following elementary properties.

\begin{gather*}
\ket{0, \abs{\epsilon}} = \frac{1}{\sqrt{\cosh(\abs{\epsilon})}}\sum_{n \geq 0} [-\tanh (\abs{\epsilon})]^n \frac{\sqrt{(2n)!}}{2^n n!} \ket{2n} \\ S(\epsilon)^\dag S(\epsilon) = \openone \\
D(\alpha)S(\epsilon) = S(\epsilon)D\left [ \alpha \cosh(\abs{\epsilon})+\alpha^* e^{i\theta}\sinh(\abs{\epsilon})\right ] \\
S^\dag  (\epsilon) \, a \, S(\epsilon) = a \cosh(\abs{\epsilon}) -a^\dag e^{i\theta}\sinh(\abs{\epsilon}) \\
\langle n \rangle = \bra{\alpha, \epsilon} a^\dag a \ket{\alpha, \epsilon} = \abs{\alpha}^2 + \sinh^2(\abs{\epsilon}^2)\\
\bra{\alpha, \epsilon} x \ket{\alpha, \epsilon} = \sqrt{2}\Re(\alpha, \epsilon) \\ 
\bra{\alpha, \epsilon} p \ket{\alpha, \epsilon} = \sqrt{2}\Im(\alpha, \epsilon) 
\end{gather*}

Let us focus our attention on the squeezed vacuum state $\ket{0, \epsilon} = S(\epsilon) \ket{0}$. By performing a map $a \rightarrow e^{i\theta/2}a$, which corresponds to performing a rotation $R(\theta/2)$ in phase space, i.e. a $\theta/2$ rotation about the origin, we see that $S(\epsilon) \rightarrow S(\abs{\epsilon})$. This suggests that we can further write $\ket{\alpha, \epsilon} = D(\alpha)R(-\theta/2)S(\abs{\epsilon})\ket{0}$. In summary, this means that every squeezed coherent state is equivalent to a squeezed vacuum state $\ket{0, \abs{\epsilon}}$, up to a rotation followed by a displacement in phase space. Neither rotation nor linear displacements in phase space affects the nonclassicality properties of the state, so for the purpose studying nonclassicality, considering squeezed vacuum will suffice.

Consider the quadrature variances for the squeezed vacuum state $\ket{0, \abs{\epsilon}}$. It can be verified that they are $\Delta x = e^{-\abs{\epsilon}}/\sqrt{2}$ and $\Delta p = e^{\abs{\epsilon}}/\sqrt{2}$, so we see that the $x$ quadrature is indeed ``squeezed", while the $p$ quadrature is ``stretched" to compensate. Furthermore, $\Delta x \, \Delta p = 1/2$ so it is a minimum uncertainty state.

Finally, one may also verify that the $P$-function of the squeezed state\cite{Schleich2001a} $\ket{0, \abs{\epsilon}}$ reads $$P(\alpha) = \exp\left[\frac{1-s}{8s} \pdv[2]{\Re(\alpha)}  - \frac{1-s}{8} \pdv[2]{\Im(\alpha)}  \right] \delta(\alpha),$$ where $s \coloneqq 2e^{2\abs{\epsilon}}$. We therefore see that the $P$-function  contains infinitely high order derivatives of delta functions, which is a signature of nonclassicality.

\subsubsection{Cat states}

In quantum optics, cat states\cite{Yurke1986, Milburn1986,Milburn1986a,Schleich1991, Brune1992} often refer to equal superpositions of two coherent states of the form $$\ket{\mathrm{\psi_{\pm}}} \coloneqq \frac{1}{\sqrt{\mathcal{N}}} (\ket{\beta} \pm \ket{-\beta}), $$ where $\mathcal{N} \coloneqq 2(1\pm e^{-2\abs{\beta}^2})$ is the normalization constant. 
They are named after Schr\"odinger's  cat paradox \cite{Schrodinger1935} that illustrates a quantum superposition on a macroscopic scale.
Their properties as macroscopic superpositions are manifest when $\beta$ is sufficiently large \cite{Lee2011}.
Depending on the sign, we can write the states in the number basis as $$\ket{\psi_+} = \frac{2 e^{-2\abs{\beta}^2}}{\sqrt{
\mathcal{N}}} \left (\sum_{n\geq 0}\frac{\beta^{2n}}{\sqrt{(2n)!}}\ket{2n} \right) $$ or $$\ket{\psi_-} = \frac{2 e^{-2\abs{\beta}^2}}{\sqrt{\mathcal{N}}} \left (\sum_{n\geq 0}\frac{\beta^{2n+1}}{\sqrt{(2n+1)!}}\ket{2n+1} \right ). $$ We see that the former is a superposition of the even number Fock states, while the latter is a superposition of the odd number Fock states. For this reason, they are referred to as the even and odd cat states respectively.

They have a $P$-function of the form 
\begin{widetext}
\begin{align*}&P(\alpha) = \frac{1}{\mathcal{N}}\delta(\alpha - \beta) + \frac{1}{\mathcal{N}}\delta(\alpha + \beta) \\
&\pm \frac{2e^{\abs{\alpha}^2-\abs{\beta}^2}}{\mathcal{N}}\exp(-2\Re(\beta)\pdv{(2\alpha^* - 2i\Im(\beta))} )  \exp(2\Re(\beta)\pdv{(2\alpha-2i\Im(\beta))} ) \delta(2\alpha - 2i\Im(\beta)) \\ 
&\pm \frac{2e^{\abs{\alpha}^2-\abs{\beta}^2}}{\mathcal{N}}\exp(2\Re(\beta)\pdv{(2\alpha^*  - 2i\Im(\beta))} )  \exp(-2\Re(\beta)\pdv{(2\alpha-2i\Im(\beta))} ) \delta(2\alpha - 2i\Im(\beta)).
\end{align*}
\end{widetext} Again, we see that due to the exponential terms, the $P$-function contains infinitely high order derivatives of the delta function, thus indicating that the state is nonclassical.

\subsubsection{Nonclassical states with regular $P$-functions} \label{sec::smoothP}

The examples of nonclassical states discussed thus far has highly singular $P$-functions. While it is true that many of the states that are actively being studied displays such exotic singularities, there are in fact many states, especially when one considers mixed quantum states, where the $P$-function is a regular function but has negative values.

One example of this state is the single photon added thermal states. Single photon added thermal states have density operators of the form $$\rho = \coloneqq a^\dag e^{-\beta a^\dag a} a / \Tr(a^\dag e^{-\beta a^\dag a} a),$$ where we again assume that $\hbar = \omega = 1$. Their $P$-functions looks like\cite{Agarwal1992, Kiesel2008} $$P(\alpha) = \frac{1+n_{T}}{\pi n_{T}^3} \left (\abs{\alpha}^2-\frac{n_{T}}{1+ n_{T}} \right)e^{-\abs{\alpha}^2/n_{T}},$$ where $n_T = \Tr(a^\dag a  e^{-\beta a^\dag a} / Z)$ is the mean photon number of a thermal state at inverse temperature $\beta$. We see that when $\abs{\alpha}^2 < \frac{n_{T}}{1+ n_{T}}$ is sufficiently small, the $P$-function is negative, so the state is nonclassical. This is illustrated in Fig.~\ref{fig::SPAT}.

Another class of nonclassical states with regular $P$-functions are the set of so called punctured states\cite{Damanet2018}, which are $P$-functions that has the form $$P(\alpha) = \mathcal{N} \left [P_{\mathrm{cl}}(\alpha) - \sum_{i=1}^N w_i \pi_i(\alpha - \alpha_i) \right ],$$ where $\mathcal{N}$ is a normalization factor, $P_{cl}(\alpha)$ is some positive $P$-function, $w_i$ are positive real numbers, and $\pi(\alpha - \alpha_i)$ are positive distributions centred at $\alpha_i$. A punctured $P$ function is therefore a positive $P$-function which are ``punctured" with negative values at several points. For certain combinations of $P_{\mathrm{cl}}(\alpha)$, $w_i$ and $\pi(\alpha)$, one may show that they correspond to physical quantum states.

Finally, it was also observed that any of the $s$-parametrized quasiprobabilities introduced in Section~\ref{sec::nonclassicality} are themselves also valid $P$-functions. $s$-parametrized quasiprobabilities may be interpreted as the $P$-functions of states that are subject to varying degrees of interaction with thermal noise\cite{Lee1991, Kuhn2018a, Tan2019} (see also Section~\ref{sec::nonclassDepth}). Therefore, any $s$-parametrized quasiprobability which is a regular functions and has negative values is an example of a nonclassical state with a regular $P$-function. For instance, the $s$-parametrized quasiprobability of the Fock state, $\ket{n}$, which is given by $$P_s(\alpha) = \frac{2}{\pi(1+s)}\left (-\frac{1-s}{1+s} \right)^n \exp(-\frac{2\abs{\alpha^2}}{1+s})L_n\left ( \frac{4\abs{\alpha}^2}{1-s^2} \right ).$$ $L_n$ is the $n$th Laguerre polynomial. We see that for $s < 1$, there are no singularities, so it is just a regular function with negativities.

The nonclassicality that comes in the form of highly singular $P$-functions or negativities of regular $P$-functions are actually identical since the existence of negativities in highly singular $P$-functions is implied\cite{Kiesel2010, Kuhn2018, Tan2019}. This is also further discussed in Section~\ref{sec::negVolume}.

\begin{figure}
\includegraphics[width = 1\linewidth ]{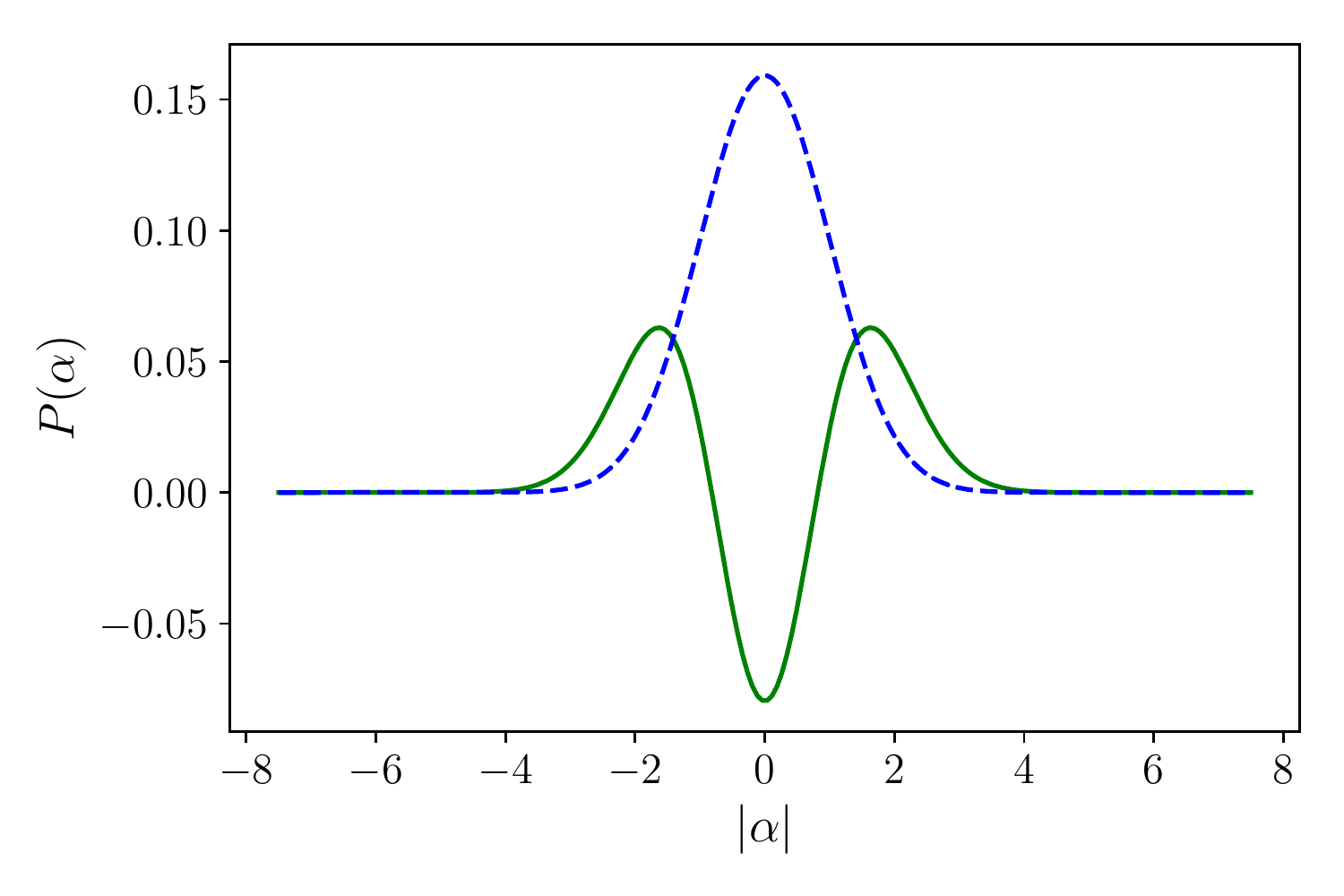}
\caption{\label{fig::SPAT} The $P$-function of a single photon added thermal state (solid line), overlaid with the $P$-function of a thermal state for $n_{T} = 2$. The $P$-functions are rotationally symmetric about the origin. The single photon added thermal state is an example of a nonclassical state whose $P$-function doesn't contain singularities.}
\end{figure}

\subsection{Survey of nonclassicality quantifiers} \label{sec::survey}

The discussion of nonclassicality presented thus far has been fairly binary in nature: either a state is classical or it is not. It is desirable however to be able to develop a more nuanced view of the subject and develop rigorous methods of discussing the extent of nonclassicality in a system. Indeed, recall that classical states are not truly classical systems (see Section~\ref{sec::classicality}), but rather the closest quantum description of one. We therefore see that right from the outset, the discussion about nonclassicality is necessarily a matter of degree. 

In the previous section, we discussed several examples of nonclassical states, many of which are highly singular, in the sense that they contain singularities more exotic that delta functions. This actually presents a real obstacle, both theoretically and experimentally, in the analysis of nonclassicality in quantum light. Theoretically, it is a problem because it is mathematically cumbersome to deal with such highly singular functions. In addition, states that \emph{require} a highly singular representation is nonclassical, but states that \emph{permit} a highly singular representation is not necessarily always nonclassical\cite{Sperling2016}. Experimentally, such highly singular points are not very well defined, which suggests that it is not feasible to sample the $P$-function directly in experiments. It will therefore be useful to be able to be able to capture the essential aspects of nonclassicality in a manner that is quantitatively informative, but also be able to avoid the technical difficulties of directly manipulating the $P$-function itself.

For the reasons above, an emerging topic in the field concerns the study of nonclassicality quantifiers. They form a set of proposals that allows us to consider nonclassicality in a more quantitative manner, and allows us to study nonclassical effects in a wider variety of systems. In this section, we will survey the variety of prominent nonclassicality quantifiers. While these nonclassicality quantifiers possess a variety of different attributes and interpretations, fundamentally speaking, they are capturing different aspects of the same notion of nonclassicality as Definition~\ref{def::nonclassicality}. 

\subsubsection{Mandel Q parameter} \label{sec::mandelQ}

One of the earliest attempts to quantify nonclassicality of light is the Mandel Q parameter\cite{Mandel1979}. We recall the representation of the coherent states in the Fock basis: $$\ket{\alpha} = e^{-\abs{\alpha}^2/2} \sum^\infty_{n=0} \frac{\alpha^n}{\sqrt{n!}} \ket{n}.$$ This gives rise to the number distribution $$P_n = e^{-\abs{\alpha}^2} \frac{\abs{\alpha}^2}{n!},$$ which has the form of a Poisson distribution\cite{Grynberg2010} $P_n = e^{-\lambda}\frac{\lambda^n}{n!}$ where $\lambda = \abs{\alpha}^2$. One distinctive feature of a Poisson distribution is that its mean and variance is equal. For the coherent states, this means that $\Delta^2n/ \expval{n}=1 $. For pure states, since coherent states are the only classical pure states\cite{Hillery1985}, this suggests that if $\Delta^2n/ \expval{n} - 1 \neq 0$, then the state must be nonclassical.

Let us consider a general classical state with density operator $\rho = \int \dd[2]\alpha P_{\mathrm{cl}}(\alpha)\ketbra{\alpha}$. The variance of the state is $\Delta^2 n = \Tr(n^2 \rho) - \Tr(n \rho)^2$. Since $(\cdot)^2$ is a convex function, we can use Jensen's inequality\cite{Rudin1987} to show that $\Tr(n \rho)^2 \leq \int \dd[2]\alpha P_{\mathrm{cl}}(\alpha)\Tr(n \ketbra{\alpha})^2$. This implies that 
\begin{align*}
\Delta^2 n &= \Tr(n^2 \rho) - \Tr(n \rho)^2 \\
&\geq \Tr(n^2 \rho) - \int \dd[2]\alpha P_{\mathrm{cl}}(\alpha)\Tr(n \ketbra{\alpha})^2 \\
&= \int \dd[2]\alpha P_{\mathrm{cl}}(\alpha) \left [\Tr(n^2 \ketbra{\alpha}) - \Tr(n \ketbra{\alpha})^2 \right ] \\ 
&= \int \dd[2]\alpha P_{\mathrm{cl}}(\alpha) \abs{\alpha}^2 \\
&= \expval{n}
\end{align*}

The above argument suggests that every classical state satisfies the inequality $\Delta^2n/\expval{n}-1 \geq 0$. This leads to the definition of the Q parameter $$Q \coloneqq \frac{\Delta^2n}{\expval{n}}-1.$$ 

The Q parameter quantifies the deviation from Poissianity. When $Q > 0$, we say that the state is super-Poissonian, when $Q < 0$ we say that it is sub-Poissonian, and when $Q=0$, the state is Poissonian. In terms of nonclassicality, only the sub-Poissonian regime matters, as sub-Poissonian states must be nonclassical. In the super-Poissonian regime, the $Q$ parameter alone is insufficient determine whether a state is classical or nonclassical. 

For instance, consider the squeezed vacuum $\ket{0, \abs{\epsilon}}$, which is a nonclassical state. One may verify that the squeezed vacuum has mean photon number $\expval{n} = \sinh^2(\abs{\epsilon})$ and variance $\Delta^2 n  = 2 \cosh^2(\abs{\epsilon})\sinh^2(\abs{\epsilon})$\cite{Andrews2014, Lvovsky2016}, hence $Q = 2\cosh^2 (\abs{\epsilon})-1 \geq 0 $ so it is in the super-Poissonian regime. More generally, $\ket{\alpha, \epsilon}$ may be either super-Poissonian or sub-Poissonian depending on the paramters. In contrast, a Fock state has $\Delta^2 n = 0$, so immediately we get $Q = -1$ and it is sub-Poissonian.

An important aspect of the $Q$ parameter is that it is related to the second order correlation function, also called the $g^{(2)}(0)$ correlation function\cite{Vogel2006}. The $g^{(2)}(0)$ correlation function is defined to be $$g^{(2)}(0)\coloneqq \frac{\expval{a^\dag a^\dag a a}}{\expval{a^\dag a}^2}.$$ It quantifies the observation of bunching/antibunching effects over an infinitesimally small detection window and can be measured in the laboratory in a Hanbury Brown-Twiss type experiment\cite{HanburyBrown1956} by measuring intensity correlations. Based on the definition of $g^{(2)}(0)$, it is not difficult to show that the $Q$ parameter is directly related to $g^{(2)}(0)$ via the relation $$Q = \expval{n}(g^{(2)}(0)-1).$$ When $g^{(2)}(0) < 1$, we are less likely to detect two photons over the detection window so we are in the anti-bunching regime.  We also see that $Q$ is negative when anti-bunching is observed. Therefore, negative $Q$, and hence nonclassicality, may be directly associated to anti-bunching. Observable anti-bunching effects is  a clear signature of a nonclassical light source\cite{Short1983, Hong1986, Kimble1977}. 

\subsubsection{Nonclassical distance} \label{sec::nonclassDist}

The nonclassical distance is a geometric based measure that was first proposed by Hillery\cite{Hillery1987}, who considered how one may distinguish between between a classical or a nonclassical state. He started with the trace norm, which is defined as $$\norm{A}_1 \coloneqq \Tr(\sqrt{A^\dag A}).$$ The trace distance between 2 operators $A$ and $B$ is then defined as the trace norm of the difference between the 2 operators, i.e. $$d_{\Tr}(A,B) \coloneqq \norm{A-B}_1.$$ Suppose we have two density operators $\rho_1$ and $\rho_2$. Then the trace distance acquires a particularly neat interpretation as the maximum probability of successfully distinguishing between $\rho_1$ and $\rho_2$ via a quantum measurement\cite{NielsenChuang}. 

Based on the trace distance, one may define the nonclassical distance of a state $\rho$ as \begin{align} \label{eq::nonclassicalityDist}
\delta_{\Tr}(\rho) \coloneqq  \inf_{\sigma_{\mathrm{cl}}}\norm{\rho-\sigma_\mathrm{cl}}_1 = \inf_{\sigma_{\mathrm{cl}}} \; d_{\Tr}(\rho, \sigma_{\mathrm{cl}}),
\end{align} where the minimization is over all classical states $\sigma_{\mathrm{cl}}$. The nonclassicality distance $\delta_{\Tr}(\rho)$ is therefore the distance between the state $\rho$ to the closest classical state. Of course, the geometric picture above does not depend on the particular choice of the distance measure. One may for instance also consider the Hilbert-Schmidt norm\cite{Dodonov2000} $$\norm{A}_{\mathrm{HS}} \coloneqq \sqrt{\Tr(A^\dag A)}$$ and the corresponding Hilbert-Schmidt distance $$d_{\mathrm{HS}}(A,B) \coloneqq \norm{A-B}_{\mathrm{HS}},$$ or the Bures fidelity\cite{Marian2002} between states $$F(\rho,\sigma) \coloneqq \Tr{\sqrt{\sqrt{\rho}\sigma \sqrt{\rho}}}$$ and the associated Bures distance $$d_{\mathrm{BU}}(\rho,\sigma) \coloneqq \sqrt{2-2F(\rho,\sigma)}.$$ We then obtain different nonclassical distances by appropriately substituting the distance measure in Eq.~\ref{eq::nonclassicalityDist}. 

In general, the definition in Eq.~\ref{eq::nonclassicalityDist} is difficult to compute, even if the state $\rho$ is completely known. This is because there is no known simple characterization of the geometry of classical and nonclassical states. However, the problem becomes much more tractable if one considers only the set of pure quantum states.

Suppose we limit ourselves to consider only the distances between two pure states $\ket{\psi}$ and $\ket{\phi}$. In this case, the respective distances between pure states are given by
\begin{subequations} \label{eq::distMeas}
\begin{align}
d_{\Tr}(\ket{\psi},\ket{\phi}) &= 2\sqrt{1 - \abs{\bra{\psi}\ket{\phi}}^2} \\
d_{\mathrm{HS}}(\ket{\psi},\ket{\phi}) &= \sqrt{2(1 - \abs{\bra{\psi}\ket{\phi}}^2)} \\
d_{\mathrm{BU}}(\ket{\psi},\ket{\phi}) &= \sqrt{2(1 - \abs{\bra{\psi}\ket{\phi}})}.
\end{align}
\end{subequations} We see that they share fairly similar expressions and  depend on the square overlap $\abs{\bra{\psi}\ket{\phi}}^2.$ Over the set of pure states, the only classical states are the coherent states, so we can consider the following alternative definition of nonclassical distance: 

\begin{align} \label{eq::nonclassDistPure}
\delta(\ket{\psi}) \coloneqq \inf_{\ket{\alpha}} \; d(\ket{\psi},\ket{\alpha}),
\end{align} 
where the minimization is over the set of coherent states $\ket{\alpha}$ and $d$ can be substituted with any of the above distance measures. For these measures, the minimization in Eq.~\ref{eq::nonclassDistPure} is achieved when $\abs{\bra{\psi}\ket{\alpha}}^2$  in Eq.~\ref{eq::distMeas} is maximal. We note that the overlap with a coherent state is related to the the Husimi $Q$-function\cite{Husimi1940} via $\abs{\bra{\psi}\ket{\alpha}}^2 = \pi Q(\alpha)$ (see Section~\ref{sec::nonclassicality}). In summary, when we consider only pure states, the nonclassical distance can be determined from the maximums of the Husimi $Q$-function\cite{Malbouisson2003}.

The Husimi $Q$-function can be probed directly in the laboratory using balanced homodyne measurements\cite{Vogel1989}, so in principle $\delta(\ket{\psi})$ can be directly measured, so long as one is reasonably confident that the output state is pure.

\subsubsection{Nonclassicality depth} \label{sec::nonclassDepth}

The nonclassicality depth was first introduced by Lee\cite{Lee1991} and then subsequently considered by L{\"u}tkenhaus and Barnett\cite{Lutkenhaus1995}. It originates from consideration of the $s$-parametrized quasiprobabilities introduced by Cahill and Glauber\cite{Cahill1969a,Cahill1969b} (see Section~\ref{sec::nonclassicality}, Eqs.~\ref{eq::sCharFun} and  \ref{eq::sQuasiProb}). We recall the $s$-parametrized characteristic function and perform a simple re-parametrization such that $\tau \rightarrow (1-s)/2$, $\chi_s(\alpha) \rightarrow \chi_\tau$ and $P_s(\alpha) \rightarrow P_\tau(\alpha)$. This leads to the characteristic function $$\chi_\tau(\beta) \coloneqq \Tr[D(\beta) \rho] e^{(1-2\tau)\abs{\beta}^2/2},$$ and the associated quasiprobability 
\begin{align*} P_\tau(\alpha) &\coloneqq \frac{1}{\pi^2}\int \dd[2]\beta \; e^{\beta^*\alpha - \beta \alpha^*  } \chi_\tau(\beta) \\
&= \frac{1}{\tau \pi} \int \dd[2]{\beta} e^{-\abs{\alpha - \beta}^2/\tau}P(\alpha)
\end{align*} where $P(\alpha)$ is the $P$-function of some given state $\rho.$ The last line indicates that the expression is just the convolution of $P(\alpha)$ with a Gaussian distribution with variance $\tau/2$. The Gaussian convolution applies a smoothing function to the $P$-function by averaging over the points around the point $\alpha$. In general, the larger the value of $\tau$, the stronger and more aggressive the smoothing.

For every $\tau \geq 0$, $P_\tau(\alpha)$ is a valid $P$-function of some quantum state\cite{Tan2019} (see Section~\ref{sec::smoothP}). One may then ask whether $P_\tau(\alpha)$ corresponds to the $P$-function of a classical state or not.  For a given state $\rho$ and corresponding set of quasiprobabilities $\{ P_\tau(\alpha) \}$, let $\mathcal{C}$ be the set of values $\tau$ such that $P_{\tau}(\alpha)$ corresponds to a positive probability distribution.

The nonclassical depth may then defined to be the quantity $$\tau_m(\rho) \coloneqq \inf_{\tau \in \mathcal{C}} \; \tau. $$ One may interpret this quantity as the amount of Gaussian smoothing required before $P(\alpha)$ becomes a classical positive distribution. We know that this is always possible because as $\tau = 1$, $P_{\tau}(\alpha)$ corresponds to the Husimi $Q$-function\cite{Husimi1940}, which is always a classical positive distribution for any state $\rho$. As result, we have that $ 0 \leq \tau_m(\rho) \leq 1$.

There is also a physical interpretation of the nonclassical depth $\tau_m(\rho)$ in terms of the amount of optical mixing with thermal noise~\cite{Lee1991,Kuhn2018a} that is necessary to make a state $\rho$ classical. This is illustrated in Fig.~\ref{fig::thermalMix}. If the input state $\rho$ interacts with a thermal state $\rho_{\mathrm{th}}$ via a highly transmissive beam splitter, then the $P$-function of the output state $\rho_{\mathrm{out}}$ is $P_\tau(\alpha)$ for some value of $\tau$. If the transmissivity $T\approx 1$ and the mean photon number of the thermal state is $\expval{n_{\mathrm{th}}}$, then it may be shown that $$\tau_m(\rho_{\mathrm{out}}) = \tau_m(\rho) - \frac{R^2}{T^2}\expval{n_{\mathrm{th}}},$$ where $R = \sqrt{1-T^2}$ is the reflectivity and $T^2+R^2 = 1$. When the temperature and hence $\expval{n_{\mathrm{th}}}$ is sufficiently large, $\tau_m(\rho_{\mathrm{out}})=0$ and the output $\rho_{\mathrm{out}}$ is classical, so we see that the nonclassicality depth may be understood as a measure of the ability of a nonclassical state to withstand thermal noise.

The nonclassicality depth may be worked out analytically for some states. For for the squeezed coherent states $\ket{\alpha, \epsilon}$, we have that $\tau_m(\ket{\alpha, \epsilon}) = (e^{2\abs{\epsilon}}-1)/(2e^{2\abs{\epsilon}})$ which suggests that $\tau_m(\ket{\alpha, \epsilon}) \in [0,1/2]$. This is monotonically increasing with $\abs{\epsilon}$ so the measure captures the nonclassicality from squeezing. 

However, it can also be shown that for the Fock states $\tau_m(\ket{n}) =  1$, regardless of the photon number $n$. We can also show that for the even and odd cat states $\tau_m(\ket{\psi_\pm}) =  1$. In fact, any non-Gaussian pure state has maximal nonclassicality depth\cite{Lutkenhaus1995}. For mixed states, one can also show that any state $\rho$ that has zero overlap with the vacuum, such that $\bra{0}\rho \ket{0} = 0$ is guaranteed to have maximal nonclassicality depth\cite{Lee1995} $\tau_m(\rho)=1$. We therefore see that the measure is unable to distinguish between many classes of nonclassical quantum states.

\begin{figure}  
\includegraphics[width = 0.8\linewidth ]{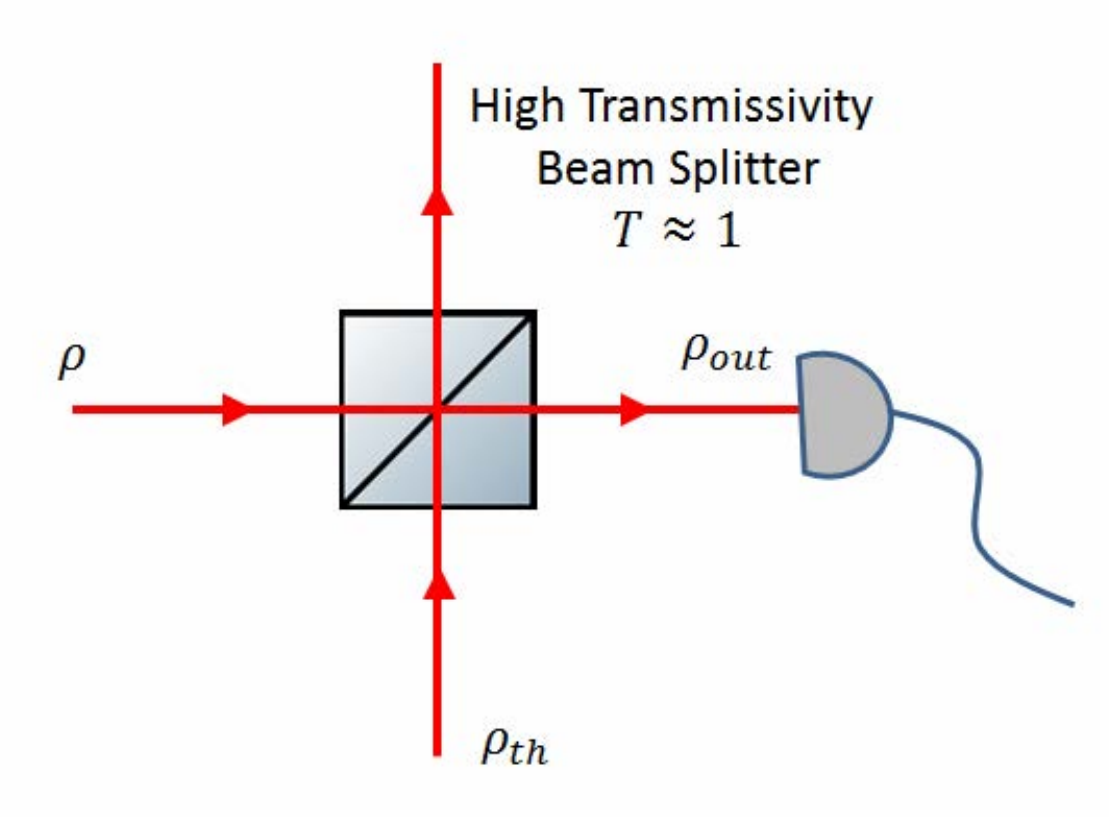}

\caption{ \label{fig::thermalMix} Optical mixing of $\rho$ with a thermal state $\rho_{\mathrm{th}}$. The $P$-function of the output state $\rho_{\mathrm{out}}$ going into the detector is a Gaussian convolution of the $P$-function of $\rho$ .}
\end{figure}

\subsubsection{Negative volume of quasiprobabilities} \label{sec::negVolume}

Recall the $s$-parametrized quasiprobabilities (Eq.~\ref{eq::sQuasiProb}, Section~\ref{sec::nonclassicality}) \begin{align} P_s(\alpha) \coloneqq \frac{1}{\pi^2}\int \dd[2]\beta \; e^{\beta^*\alpha - \beta \alpha^*  } \chi_s(\beta).
\end{align}

If any of the $s$-parametrized quasiprobabilities displays negativities, then the associated $P$-function must also be nonclassical. The main technical difficulty is that in general $P_s(\alpha)$ may display points more singular than a delta function, such as those we see in the examples discussed in Section~\ref{sec::examplesPfun}. The existence of such highly singular points leads to many technical difficulties.

It is known however, that for $s \leq 0$, the quasiprobabilities are always uniformly continuous functions\cite{Cahill1969a}, this allows us to sidestep the problem of highly singular points. As $s=0$ corresponds to the Wigner function, this prompted Kenfack and {\.Z}yczkowski\cite{Kenfack2004} to consider the negative volume of the Wigner function as a nonclassicality measure. The negative volume may be defined as $$N_w \coloneqq \frac{1}{2}[\int \dd[2]\alpha \abs{P_0(\alpha)}-1].$$

The primary benefit of this approach is that it is easy to compute for both pure and mixed quantum states, and that the Wigner function can be directly sampled in the laboratory\cite{Lvovsky2009}. However, it is also clear that there are nonclassical quantum states with positive Wigner functions, so the measure will fail to detect some nonclassical states. For instance, from Hudson's Theorem\cite{Hudson1974}, we know that the only pure quantum states with positive Wigner functions are coherent and squeezed states, but squeezed states are highly nonclassical. The negativity of the Wigner function also turns out to be a measure of non-Gaussianity, which is further discussed in Section~\ref{sec::nonGaussianity}.

More recently, inspired by the nonclassical filter approach of Ref.~\onlinecite{Kiesel2010}, Tan, Choi and Jeong\cite{Tan2019} considered expanding the notion of negative volume to every $P_s(\alpha)$. In order to sidestep the problem of singularities, they considered a filtered characteristic function of the form $$\chi_{s,w}(\alpha) \coloneqq \chi_s(\alpha)\Omega_w(\alpha),$$ and the corresponding filtered quasiprobability $$P_{s,w}(\alpha) \coloneqq \frac{1}{\pi^2}\int \dd[2]\beta \; e^{\beta^*\alpha - \beta \alpha^*  } \chi_{s,w}(\beta).$$ They showed that by applying an appropriate filtering function $\Omega_w(\alpha)$, they can ensure that $P_{s,w}(\alpha)$ contains no singularities for every $s$ and $w$, and that $P_{s,w}(\alpha) \rightarrow P_{s}(\alpha)$ as $w \rightarrow \infty$. This means that the positive and negative regions of $P_{s,w}(\alpha)$ are always well defined. This allows one to define the negative volume of every $s$-parametrized quasiprobability in the form of the limit $$N_s(\rho) \coloneqq \lim_{w \rightarrow \infty} \frac{1}{2}[\int \dd[2]\alpha \abs{P_{s,w}(\alpha)}-1].$$ We see that under this definition, if $P_s(\alpha)$ is a regular function with no singularities, we  retrieve the regular definition of negative volume $$N_s(\rho) \coloneqq \frac{1}{2}\left [\int \dd[2]\alpha \abs{P_s(\alpha)}-1 \right]. $$ $N_s(\rho)$ therefore forms a continuous hierarchy of well defined nonclassicality measures. In general, as $s$ decreases, $N_s(\rho)$ becomes a weaker measure in the sense that the number of nonclassical state it is able to identify decreases. At $s=1$, which is the negativity of the $P$-function itself, the measure identifies every nonclassical state, and has an operational interpretation as the robustness to classical noise. They also show that $N_s(\rho)$ belongs to a resource theory of nonclassicality, which is further discussed in Section~\ref{sec::resTheoryNonclass}.

\subsubsection{Entanglement potential}

Entanglement is another notion of nonclassicality that has gained considerable interest in the physics community\cite{Horodecki2009}. This interest is in part thanks to the advent of quantum information technologies, where many  quantum protocols are enabled by the esoteric properties of quantum entanglement\cite{NielsenChuang}. A state $\rho_{\mathrm{sep}}$ is said to be separable if it can be written as some convex combination of pure, product states such that $\rho_{\mathrm{sep}} = \sum_i p_i \ketbra{\psi_i} \otimes \ketbra{\phi_i},$ where $p_i$ is some probability distribution. If $\rho$ cannot be written in this form, then we say that the state is entangled.

Nonclassicality in light and quantum entanglement are in fact closely knit notions. It was first noted by Aharanov \textit{et al}\cite{Aharonov1966} that the only pure state of light that will produce a separable state after passing through a beam splitter is the coherent state. This observation subsequently extended to general mixed states\cite{kim2002, Wang2002}, and we now know that the nonclassicality of a light source is both necessary and sufficient to generate entanglement via a beam splitter\cite{Asboth2005}.

In order to see this, let us consider a 50:50 beam splitter. This can be described via a unitary $U_{50:50}$ that performs the transformation 
\begin{align*}
a_1 &\rightarrow \frac{1}{\sqrt{2}}(a_1-a_2) \\
a_2 &\rightarrow \frac{1}{\sqrt{2}}(a_1+a_2).
\end{align*} In particular, we see see that if the input states of the beam splitter are coherent states $\ket{\alpha}\ket{\beta}$, then the output state is just the product state $\ket{\frac{\alpha-\beta}{\sqrt{2}}}\ket{\frac{\alpha+\beta}{\sqrt{2}}}.$

Now, consider the case where the input state has the form $\rho \otimes \ketbra{0}$. Suppose the input state is classical, so $\rho = \int \dd[2]\alpha P_{\mathrm{cl}}(\alpha) \ketbra{\alpha}$ where $P_{\mathrm{cl}}(\alpha)$ is a positive probability distribution. It is clear that under the 50:50 beam splitter, the resulting output state is given by the transformation $$\rho \otimes \ketbra{0} \rightarrow \int \dd[2]\alpha P_{\mathrm{cl}}(\alpha) \ketbra{\frac{\alpha}{\sqrt{2}}} \otimes \ketbra{\frac{\alpha}{\sqrt{2}}}.$$ We see that it is a convex combination of product states, so it is always separable. This shows that in this setup, nonclassicality is a necessary condition to produce entanglement.

Now, suppose the output state is separable. This means that the output state can be written in the form $$\sum_i p_i \ketbra{\psi_i} \otimes \ketbra{\phi_i}.$$ This cannot contradict the fact that the initial state has the form $\rho \otimes \ketbra{0}$, which means that for every $i$, $U^{-1}_{50:50}\ket{\psi_i}\ket{\phi_i} = \ket{\psi'_i}\ket{0}$. As previously noted, the only pure input states that permits a separable state after passing through a beam splitter is the coherent state. As such, we must have that $\ket{\psi'_i}\ket{0} = \ket{\alpha_i} \ket{0}$ where $\ket{\alpha_i}$ is some coherent state. This in turn suggests that we can write $\rho = \sum_{i} p_i \ketbra{\alpha_i}$ where $p_i$ is some probability distribution, so $\rho$ has a classical $P$-function since the expression is a classical statistical mixture of coherent states. This means that if the output state is entangled, we can safely conclude that the input state is nonclassical. This proves the converse statement, so nonclassicality at the input is both a necessary and sufficient condition for the output fields of a beam splitter to be entangled. Note that the previous argument does not specifically rely on the fact that the beam splitter is $50:50$, so the same statement applies for all beam splitters. The primary motivation for choosing a 50:50 beam splitter relies on the fact that numerical evidence suggests that it generates greater amounts of entanglement as compared to unbalanced beam splitters\cite{Asboth2005}. 

Since nonclassicality always gives rise to entanglement at the output ports of the beam splitter, this motivates the following definition of the entanglement potential: $$\mathrm{EP}(\rho) = E \left [U_{50:50}\left (\rho\otimes \ketbra{0} \right )U^\dag_{50:50} \right ],$$ where $E$ can be any measure of entanglement. We have therefore have shifted the problem of quantifying nonclassicality to a problem of quantifying entanglement. Examples of entanglement measures include the logarithmic negativity\cite{Horodecki1998} and the relative entropy of entanglement\cite{Horodecki2000}.

However, even though entanglement is a very well understood phenomena, quantifying entanglement is not necessarily a simple task. For instance, while the logarithmic negativity is computable, it is not able to detect every entangled state, which in turn suggests it cannot quantify the nonclassicality of every state. In comparison, the relative entropy of entanglement can detect all entangled states, but is not computable in general. To date, there is no known measure of entanglement that is simultaneously easy to compute, and able to detect all entangled states.

\subsubsection{Negativity of normal ordered observables}

It is well known that it is easier to calculate the mean values of observables via the $P$-function when the observable is normal ordered than when it is not\cite{Schleich2001a}. A somewhat more surprising fact is that nonclassical $P$-functions treat normal ordered observables differently from non-normal ordered ones\cite{Shchukin2005}, and that this may be used as the basis of a nonclassicality quantifier\cite{Gehrke2012}. 

Recall that normal ordering operation $:(\cdot):$ simply means all the creation operators are moved to the left, while all the annihilation operators to the right. For instance, $:aa^\dag: = a^\dag a$.

Suppose we have some function of the creation and annihilation operators $f(a^\dag, a)$. By itself, this is not necessarily a Hermitian operator, so it may not be an observable. In contrast, the function $f^\dag(a^\dag, a)f(a^\dag, a)$ is always Hermitian, and in principle, always corresponds to some measurable observable. 

Consider the normal ordered observable ${:f^\dag(a^\dag, a)f(a^\dag, a):}$, which is also Hermitian. For a pure coherent state $\ket{\alpha}$, the expectation value of the observable is $\bra{\alpha}:f^\dag(a^\dag, a)f(a^\dag, a): \ket{\alpha} = f^*(\alpha^*, \alpha)f(\alpha^*, \alpha)= \abs{f(\alpha^*, \alpha)}^2 \geq 0.$ Note that the expectation value is  always positive. 

We now extend this observation to a general classical state $\rho = \int \dd[2]{\alpha} P_{\mathrm{cl}}(\alpha) \ketbra{\alpha}$, where $P_{\mathrm{cl}}(\alpha)$ is a positive classical distribution. Computing the expectation value again, we have the expression \begin{align*}
\expval{:f^\dag(a^\dag, a)f(a^\dag, a):} &= \Tr(:f^\dag(a^\dag, a)f(a^\dag, a): \rho )\\
&= \int \dd[2]{\alpha} P_{\mathrm{cl}}(\alpha) \abs{f(\alpha^*, \alpha)}^2.
\end{align*} Since $P_{\mathrm{cl}}(\alpha)$ and $\abs{f(\alpha^*, \alpha)}^2$ are always positive, we must have that $\expval{:f^\dag(a^\dag, a)f(a^\dag, a):} \geq 0$ for every classical state $\rho$. As a consequence, any observed negativity in the observed expectation value $\expval{:f^\dag(a^\dag, a)f(a^\dag, a):}$ must be a consequence of the negativity of $P(\alpha)$, which indicates nonclassicality. 

Motivated by the observation above, Gehrke, Sperling and Vogel\cite{Gehrke2012} defined the following quantity, which they call the operational relative nonclassicality: 
\begin{align*} 
R(\rho) \coloneqq 
\begin{cases}
\expval{:f^\dag f:}/\Delta, & \mathrm{if}\; \expval{:f^\dag f:}< 0\\
0, & \mathrm{otherwise}
\end{cases}
\end{align*}
where $\Delta \coloneqq \expval{:f^\dag f:} - \expval{f^\dag f} $ and $f \coloneqq f(a^\dag,a).$ The factor $\Delta$ is essentially a normalization factor, since $\expval{f^\dag f} \geq  0$ so $\expval{:f^\dag f:}/\Delta \leq 1$. Furthermore, one may also show that for any nonclassical state $\rho$, it is always possible to find some $f(a^\dag, a)$ such that $R(\rho) > 0$. 

Interestingly, one may also view this approach as a generalization of the Mandel $Q$ parameter\cite{Mandel1979} (see Section~\ref{sec::mandelQ}). For any given state $\rho$, we can choose $f \coloneqq  a^\dag a - \Tr(a^\dag a \rho) $. We then obtain $R(\rho) = - Q$, where $Q$ is exactly the Mandel $Q$ parameter.

Viewed as a generalization of the Mandel $Q$ parameter, the main technical complication is also similar. Just as the $Q$ parameter is unable to detect all nonclassical states, there is no guarantee that $R(\rho)$ will be able to detect every nonclassical state for any given choice of $f(a^\dag, a)$.

\subsubsection{Degree of nonclassicality} \label{sec::nonclassDegree}

In the theory of entanglement, it is a well known property of pure entangled states\cite{NielsenChuang} $\ket{\psi}_{ab}$ that, up to local unitary operations, they can always be written in the form $$\ket{\psi}_{ab} = \sum_{i=1}^r \lambda_i \ket{i,i},$$ where $\ket{i,i}$ is a product of local orthonormal basis states. The above is called the Schmidt decomposition\cite{Schmidt1906}, and the total number of superpositions involved is given by the Schmidt number $r$. The Schmidt number\cite{Terhal1999} is a well studied measure of entanglement.

We may also attempt to construct a similar construction for nonclassicality based on the number of superpositions. For pure states, the only nonclassical states are the coherent states. Furthermore, since the set of coherent states forms an overcomplete basis(see Section~\ref{sec::classicality}), any state can be written as a superposition of coherent states. Therefore, we can consider the minimum number of superpositions $r$ such that $$\ket{\psi}= \sum_{i=1}^r\lambda_i \ket{\alpha_i},$$ where $\{ \ket{\alpha_i} \}$ is some set of coherent states. One may then define the nonclassicality degree\cite{Gehrke2012} as $$\kappa(\ket{\psi}) \coloneqq r-1,$$ where it is clear that for any nonclassical pure state, we must have $\kappa \geq 0$.

In order to generalize this to mixed states we consider all possible pure state decompositions $\{ p_i, \ket{\psi_i} \}$ such that $\rho = \sum_i p_i \ketbra{\psi_i}.$ The nonclassicality degree of $\rho$ is defined as $$\kappa(\rho)\coloneqq \min_{\{ p_i, \ket{\psi_i} \}} \max_i \kappa(\ket{\psi_i}).$$ The above quantity is the largest nonclassicality degree $\kappa(\ket{\psi_i})$ found in every pure state decomposition $\{ p_i, \ket{\psi_i} \}$, minimized over all such decompositions. Such minimax constructions are typically called convex roof constructions\cite{Bennett1996,Uhlmann1998}.

There are several complications involved in applying this measure to nonclassical states. First, this is a discrete measure, which immediately means that some of the details and nuance of continuous nonclassicality measures are lost. For instance, the even cat states $\ket{\psi_+} \coloneqq \frac{1}{\sqrt{\mathcal{N}}}(\ket{\beta}+ \ket{-\beta})$ can be made arbitrarily close to the vacuum state $\ket{0}$ as $\abs{\beta} \rightarrow 0$, but $\kappa(\ket{\psi_+}) = 2$ even for very small $\ket{\beta}.$

Second, $\kappa(\rho)$ is generally not computable for an arbitrary mixed state $\rho$, due to the convex roof construction, which requires a minimization over all possible pure state decompositions. In general, this is a difficult proposition. 

Finally, there is no simple computational method to determine the minimum number of superpositions even for pure states. Determining the degree of nonclassicality will require mathematical analysis on a case by case basis, which may be quite complicated in general. For instance, the Fock states $\ket{n}$ requires an infinite number of superpositions of coherent states in the exact case\cite{Gehrke2012}, but can also be written as the limit of only $n$ superpositions\cite{Kuhn2018}, so it is not always apparent what the degree of nonclassicality is. 

\subsubsection{Operator ordering sensitivity}

We recall the $s$-parametrizes characteristic function (see also Eq.~\ref{eq::sCharFun}, Section~\ref{sec::nonclassicality}) \begin{align*} 
\chi_s(\beta) \coloneqq \Tr[D(\beta) \rho] e^{s\abs{\beta}^2/2}.
\end{align*} We can observe that the characteristic function is closely related to the displacement operator, which we recall has the form $$D(\beta) \coloneqq e^{\beta a^\dag - \beta^* a}.$$ The parameter $s$ can in fact be related to operator ordering, and is sometimes also called the order parameter\cite{Cahill1969a, Cahill1969b}. We can see this by directly applying the Baker-Campbell-Hausdorff formula\cite{Vogel1989} $e^Ae^B = e^C$ where $c = A+B+\comm{A}{B}/2+\comm{A}{\comm{A}{B}}/12-\comm{B}{\comm{A}{B}}/12\ldots.$ From the commutation relation $\comm{a}{a^\dag}=1$, we can show the following: 
\begin{align*} 
D(\beta) e^{\abs{\beta}^2/2} &= e^{\beta a^\dag}e^{- \beta^* a}  =\mathopen{:} D(\beta) \mathclose{:}  \\
D(\beta) e^{-\abs{\beta}^2/2} &= e^{- \beta^* a}e^{\beta a^\dag}  = \mathopen{\vdots} D(\beta) \mathclose{\vdots} , 
\end{align*} where $\mathopen{:} (\cdot) \mathclose{:}$ and $\mathopen{\vdots} (\cdot) \mathclose{\vdots}$ denotes the normal and antinormal ordering operations respectively. Therefore, when $s=1$, we have normal ordering as we can write $\chi_s(\beta) \coloneqq \Tr[\mathopen{:} D(\beta) \mathclose{:} \rho]$. When $s=-1$, we have antinormal ordering as we can write $\chi_s(\beta) \coloneqq \Tr[\mathopen{\vdots} D(\beta) \mathclose{\vdots} \rho]$. Furthermore, at $s=0$, we see from the Taylor expansion that $D(\beta) = e^{\beta a^\dag - \beta^* a} = \sum_{k\geq 0 } (\beta a^\dag - \beta^* a)^k/k! $ contains every possible permutation of the operators $a^\dag$ and $a$, which corresponds to symmetric ordering. Therefore, as we increase the value of $s$ from $s=-1$ to $s=1$, we are transitioning from antinormal ordering to symmetric ordering to normal ordering. Lee's\cite{Lee1991} nonclassicality depth (see Section~\ref{sec::nonclassDepth} is essentially based on a reparametrization of the order parameter $s$, so it also has an interpretation in terms of operator ordering.

In Ref~\onlinecite{Bievre2019}, Bi{\`e}vre \textit{et al.} introduced the $s$-ordered entropy of a state\[  \label{eq::ordSensClass} H(s,\rho) \coloneqq -\ln(\pi \norm{P_s}^2_2),\] where $P_s$ is the $s$ parametrized quasiprobability defined in Eq.~\ref{eq::sQuasiProb} and $\norm{P_s}^2_2 \coloneqq \int \dd[2]{\alpha} P_{s}^2(\alpha)$ is just the square integral. When $P_S(\alpha)$ is a classical probability distribution function, then $H(s,\rho)$ is an entropy measure belonging to the family of R{\'e}nyi\cite{Renyi1960} entropies. They were able to demonstrate that for any classical state $\rho_{\mathrm{cl}}$, the derivative $H'(s,\rho)\coloneqq\pdv{s}H(s,\rho)$  will always satisfy the following inequality: $$0 \leq -H'(0,\rho_{\mathrm{cl}}) \leq 1. $$ This implies the nonclassicality condition: \begin{align} \label{eq::ordSensCriterion} S_o(\rho) \coloneqq -H'(0,\rho_{\mathrm{cl}}) >  1 \Rightarrow \rho \text{ is nonclassical.} \end{align}Therefore, the sensitivity of the entropy at $s=0$ to a small increase in $s$ may be used to as a nonclassicality criterion. For this reason, $S_o$ is called the operator ordering sensitivity. 

Recall that at $s=0$, $P_{s}(\alpha)$ is the Wigner function. It is unclear from physical grounds why the sensitivity at $s=0$ proves particularly important. However, the above criterion can be given a geometric interpretation. Suppose instead of $\rho$, we consider the space of $\tilde{\rho} \coloneqq \rho/\sqrt{\Tr(\rho^2)}$, which is just the density operator space scaled by the purity of the state.  It can be shown that $S_o(\rho)$ satisfies all the properties of a norm on the space of $\tilde{\rho}$, so we can further write $S_o(\rho) \coloneqq \norm{\tilde{\rho}}_{o}$. Finally, following the same procedure as the nonclassicality distance\cite{Hillery1987, Dodonov2000,Marian2002}(see Section~\ref{sec::nonclassDist}), we can define the distance measure $d_o(\tilde{\rho}, \tilde{\sigma}) \coloneqq \norm{\tilde{\rho}-\tilde{\sigma}}_{o}$ and consider the geometric measure $$\delta_0(\rho) \coloneqq \inf_{\tilde{\sigma}_\text{cl}} d_o(\tilde{\rho}, \tilde{\sigma}_\text{cl}),$$ where the minimization is over $\tilde{\sigma}_\text{cl} = \sigma_{\text{cl}}/\Tr(\sigma^2))$ where $\sigma_\text{cl}$ is classical. This quantifies the distance to the closest classical state on the space of $\tilde{\rho}$. In this picture, Eq.~\ref{eq::ordSensCriterion} says that all classical states lie inside the unit ball on this scaled space. If a state $\tilde{\rho}$ is found outside of the unit ball, then it must be nonclassical. The main difference between this approach and the nonclassical distance is the rescaling of the geometry by purity.

The geometric picture also suggests the inequality $$ S_o(\rho) -1 \leq \delta_0(\rho) \leq S_o(\rho),$$ so the geometric nonclassicality measure is always bounded by the operator ordering sensitivity $S_o(\rho)$.

The authors\cite{Bievre2019} were careful to point out that the geometric measure $\delta_0(\rho)$ is the nonclassicality measure, while the operator ordering sensitivity $S_o(\rho)$ is just a bound. In general, $\delta_0(\rho)$ is not readily computable. Nonetheless, the bound becomes sufficiently tight when $\delta_0(\rho) \gg 1$ and $S_o(\rho)$ is computable given the eigendecomposition of $\rho$. For pure states $\ket{\psi}$, we have $S_o(\ket{\psi}) = \Delta^2x +\Delta^2p$ which is just the sum of quadrature variances. This coincides with some other previously considered measures for pure states\cite{Lee2011, Kwon2019}. It is clear that only coherent states satisfy $S_o(\ket{\psi}) = 1$, so the criterion in Eq.~\ref{eq::ordSensCriterion} is sufficient to detect every nonclassical pure state, but is not sufficient in general to detect every nonclassical mixed state.

\subsubsection{Measures of non-Gaussianity} \label{sec::nonGaussianity}

A Gaussian state is a special class of optical quantum states whose Wigner function ($P_0(\alpha)$ in Section~\ref{sec::nonclassicality}) is a Gaussian function\cite{Adesso2014}. For the (single mode) description of Gaussian states, it is more convenient to use  the cartesian coordinates  $\mathbf{r} \coloneqq (\sqrt{2}\Re(\alpha),\sqrt{2}\Im(\alpha) )$ in phase space over the complex variable $\alpha$. Let us denote the Wigner function as $W(\mathbf{r}) \coloneqq P_0(\alpha)$.

By definition, every Gaussian state $\rho_{\text{G}}$ has a Wigner function of the form 
\begin{align} \label{eq::gaussianState}
W(\mathbf{r}) = \frac{1}{2\pi\sqrt{\det V}} \exp[-\frac{1}{2}(\mathbf{r}- \bar{\mathbf{r}})^TV^{-1}(\mathbf{r}- \bar{\mathbf{r}})], 
\end{align}
where $\bar{\mathbf{r}} = (\expval{x}, \expval{p})$ and $V$ is the covariance matrix, which is given by 
$$\begin{bmatrix}
\Delta^2x &\expval{ \frac{1}{2}\anticommutator{x-\expval{x}}{p-\expval{p}}} \\  
\frac{1}{2}\expval{\anticommutator{x-\expval{x}}{p-\expval{p}}}  & \Delta^2p
\end{bmatrix}.$$ Examples of Gaussian states include coherent states, thermal states, and squeezed states.

Also relevant are the set of Gaussian operations. A Gaussian unitary $U_\text{G}$ is any combination of displacement operations, phase shifters, beam splitters and squeezing operations\cite{Ma1990,Cariolaro2016}. A general Gaussian operation $\Phi_\text{G}$ is any operation that can be written in the form $$\Phi_\text{G}(\rho_1) = \Tr_2[ U_\text{G} \rho_1 \otimes \ketbra{0}_2 U_\text{G}^\dag].$$ Such maps always maps a Gaussian state to another Gaussian state.   

From the above, we see that Gaussian states permit a particularly simple description only in terms of the first and second moments $\bar{\mathbf{r}}$  and $V$. By considering only an initial Gaussian state, and then performing Gaussian operations, we can stay completely within the Guassian regime and thereby work out every required property by considering only $ \bar{\mathbf{r}}$ and $V$. Many Gaussian states can also be produced under laboratory settings\cite{Braunstein2005}. As a result, the properties of Gaussian states are particularly well understood and confirmed by experiments, resulting in a whole subfield called Gaussian quantum information\cite{Weedbrook2012}. 

However, it should be clear that Gaussian states comprise only a small subset of the possible quantum states. It is therefore not a surprise that many quantum protocols are not possible if one stays strictly within the Gaussian regime\cite{Lloyd1999, Eisert2002, Giedke2002, Fiurasek2002, Bartlett2002, Cerf2005, Menicucci2006, Niset2009, Zhang2010, Ohliger2010}. This has prompted the study of non-Gaussian states as a possible supplement to Gaussian resources in order to fill this gap and hence led to the development of a family of measures of non-Gaussianity.

In the strict definition of the non-Gaussianity, every quantum state whose Wigner function is not a Gaussian distribution is considered non-Gaussian. Several non-Gaussianity measures have been proposed according to this strict definition, most of which are geometric based measures similar to the nonclassicality distance\cite{Genoni2007, Genoni2008, Genoni2010, Ivan2012, Marian2013, Ghiu2013, Park2017} (see Section~\ref{sec::nonclassDist}) which tries to measure the distance of a given state $\rho$ to the closest Gaussian state $\rho_\text{G}$. In this strict definition, there is no clear relationship between non-Gaussianity and nonclassicality, as many non-Gaussian states are classical. A simple example of this is the equal mixture of two coherent states $\rho = (\ketbra{\alpha_1}+\ketbra{\alpha_2})/2$, which is clearly classical. Such states have two peaks and clearly cannot be written in the form of Eq.~\ref{eq::gaussianState}, so they must be non-Gaussian. This points to yet another issue, which is that the set of Gaussian states is not a convex set. This means that it is possible to mix two Gaussian state to form a non-Gaussian state, thereby producing non-Gaussianity. It is not clear why the non-Gaussianity of such states would lead to any interesting quantum effects.

More recently, there have been proposals to formulate a quantum resource theory of non-Gaussianity\cite{Albarelli2018, Lami2018, Takagi2018, Zhuang2018, Park2019} where the definition of non-Gaussianity  is modified to include any quantum state that is not inside the convex hull of Gaussian states. (See Section~\ref{sec::resTheoryNonclass} for a more in depth description of quantum resource theories.) According to this definition, only states that cannot be written in the form $$\rho = \sum_i p_i \rho^i_G,$$ where $p_i$ is a probability distribution and $\rho^i_G$ is some Gaussian state, is a genuine non-Gaussian resource. Since coherent states are Gaussian, this means that every state with a classical $P$-function lies within the convex hull of Gaussian states. As a consequence, this newly redefined, genuine non-Gaussian resource states must also have nonclassical $P$-functions. Non-Gaussianity of this type are therefore genuinely quantum in nature. Indeed, the negativity of the Wigner function was one of the proposed measures of non-Gaussianity\cite{Albarelli2018, Takagi2018}. Given the Wigner function $W(\mathbf{r})$ of some given state $\rho$, the logarithmic negativity of $\rho$ is defined as $$L_w(\rho) \coloneqq \log \int \dd[2]\mathbf{r} \abs{W(\mathbf{r})},$$ which quantifies the (logged) negative volume of the Wigner function. We already know that the negativity of the Wigner function implies a nonclassical $P$-function (see Section~\ref{sec::nonclassicality} as well as Section~\ref{sec::negVolume}).

However, even with such a redefinition, it remains debatable whether measures of non-Gaussianity can be considered a measure of nonclassicality. The convex hull of Gaussian states necessarily contain many nonclassical states, with the most prominent being the squeezed coherent states (see Section~\ref{sec::squeezedStates}). Any non-Gaussianity measure will therefore  exclude such states. For instance, the Wigner function of a squeezed state is always positive, so the corresponding Wigner negativity will always be zero.

Furthermore, under the resource theoretical approach, there is a strict requirement that measures of non-Gaussianity do not increase under Gaussian operations, which includes squeezing operations. Such non-Gaussianity measures therefore cannot capture any increase in nonclassicality due to squeezing. There is no apparent way to resolve the aforementioned issues because they are a feature of the definition of non-Gaussianity itself. As such, since the starting point of non-Gaussianity is  qualitatively different from nonclassicality, it is perhaps more appropriate for it to be considered a concept with significant overlap with the notion of nonclassicality, rather than a measure of nonclassicality itself.

\subsubsection{Resource theory of nonclassicality} \label{sec::resTheoryNonclass}

In the previous section, the resource theoretical approach towards quantifying non-Gaussianity was briefly discussed. While the notion of non-Guassianity has significant overlap with nonclassicality, it does not completely address the nonclassicality of light per se(see discussion in Section~\ref{sec::nonGaussianity}). As such, there have been recent proposals to adopt the resource theoretical approach to directly quantify the nonclassicality of light. This section will discuss the recent developments in this space.

A quantum resource theory\cite{Chitambar2019} is a framework for quantifying various notions of quantumness. In general, there are many different kinds of quantum resource theories. Examples include the resource theories of entanglement\cite{Horodecki2009}, coherence\cite{Streltsov2017}, and the aforementioned resource theory of non-Gaussianity (see Section~\ref{sec::nonGaussianity}). While many different resource theories are currently being studied, the underlying approach remains broadly the same across all such theories. The essential idea is to cast different notions of quantumness as resources that are not freely available.

Let us define this concept more precisely. Suppose we have a well defined set of classical states $\mathcal{C}$, which is a strict subset of the Hilbert space. Any state that does not belong to $\mathcal{C}$ is  nonclassical by definition. Associated with the set of classical states $\mathcal{C}$, let us also define some set of operations $\mathcal{O}$, which is a strict subset of the set of all possible quantum operations, with the only requirement being that if $\Phi \in \mathcal{O}$ and $\rho \in \mathcal{C}$, then $\Phi(\rho) \in \mathcal{C}$. In other words, we require that any quantum operation belonging to $\mathcal{O}$ be unable to produce nonclassical states from classical ones.

For a given resource theory, we then require that any measure of nonclassicality $N(\rho)$ to be a nonnegative quantity that satisfies the following properties:

\begin{enumerate}
\item \label{item::classState} $N(\rho) = 0$ if  $\rho \in \mathcal{C}$.

\item \label{item::monotonicity}  (Monotonicity)  $N(\rho) \geq N(\Phi(\rho))$ if $\Phi \in \mathcal{O}$.

\item \label{item::convexity} (Convexity), i.e. $N(\sum_i p_i \rho_i) \leq \sum_i p_i N( \rho_i)$ .

\end{enumerate}

Property~\ref{item::classState} simply requires that the measure $N(\rho)$ returns positive values only when $\rho$ is nonclassical. Property~\ref{item::convexity} requires that $N(\rho)$ be a convex function of state. This is to ensure that you cannot increase nonclassicality by creating a simple statistical mixture of states $p\rho+(1-p)\sigma$. Such statistical mixing processes clearly does not involve quantum processes, and so cannot be expected to increase quantum nonclassicality in any reasonable measure $N$.

Property~\ref{item::monotonicity} requires that $N(\rho)$ always monotonically decreases if an operation $\Phi$ is an operation of $\mathcal{O}$. The monotonicity property is perhaps the defining property of all resource theoretical measures. It encapsulates the idea that one can neither freely produce nor increase quantum nonclassicality by performing any operation in $\mathcal{O}$. In this sense, nonclassical states $\rho$, and the nonclassicality of the state $N(\rho)$ are both resource that are not freely available. Under the resource theoretical framework, nonclassical quantum states acquire an interpretation as resources that overcomes the limitations of classical states $\mathcal{C}$ and operations $\mathcal{O}$.

For the quantification of nonclassicality in light, the set of classical states is unambiguous: $\mathcal{C}$ must be the set of states with classical $P$-functions (Section~\ref{sec::nonclassicality}). The set of operations $\mathcal{O}$ therefore needs to be defined in order to formulate a resource theory of nonclassicality. The earliest known proposal to formulate a resource theory of nonclassicality for light is by Gehrke \emph{et al.}\cite{Gehrke2012,Sperling2015}. There, it was proposed that $\mathcal{O}$ be the maximal set of quantum operations $\Phi$ that always maps every classical state $\rho_{\text{cl}} \in \mathcal{C}$ into another classical state $\Phi(\rho_{\text{cl}}) \in \mathcal{C}$. It was subsequently shown that under this proposal, the nonclassicality degree (see Section~\ref{sec::nonclassDegree}) satisfies Properties~\ref{item::classState},\ref{item::monotonicity} and\ref{item::convexity}. Other examples of measures belonging to this resource theory are the nonclassicality distance (see Section~\ref{sec::nonclassDist}, in particular for the trace\cite{Hillery1987} and Bures\cite{Marian2002} distance based measures. This is because both trace and Bures distances are known to monotonically decrease under general quantum maps\cite{Gilchrist2005}, which guarantees that the nonclassical distance also monotonically decreases under $\mathcal{O}$. 

The primary difficulty with Gehrke \emph{et al.}'s approach is that while $\mathcal{O}$ is simple to define, there is no known characterization of the set of operations $O$ and what kind of operations they represent physically. We recall that in a resource theory, one of the motivation is to cast nonclassicality as a resource that overcomes the limitations of the set of operations $\mathcal{O}$. In this case, there is no clear argument from physical grounds why one should be interested to overcome the limitations inherent to this definition of $\mathcal{O}$.

Subsequently, Tan \textit{et al.}\cite{Tan2017}, noting that nonlinear operations are required in order to produce nonclassical states, proposed a resource theory of nonclassicality based on the set of linear optical operations. A unitary linear optical operation $U_L$ is defined to be the set of passive linear optical elements (i.e. any combination of beam splitters, mirrors, and phase shifters) supplemented by displacement operations. Let $a^\dag_k$ be the creation operator of the $k$th mode, then $U_L$ represents any transformation of the type $$a_i^\dagger \rightarrow \sum_{k=1}^K \mu_k a_k^\dagger +  \bigoplus_{k=1}^K \alpha_k \openone_k$$ where $\mu_k$ are any complex values satisfying $\sum_{k=1}^K \abs{\mu_k}^2=1$ and $\alpha_k$ are arbitrary complex numbers. More generally, a linear optical map is defined to be any map $\Phi_L$ that can be expressed in the form  $$\Phi_L(\rho_A) = \mathrm{Tr}_E(U_L \rho_A \otimes \sigma_{\text{cl}} U_L^\dag),$$ 
where $\sigma_{\text{cl}}$ is some classical state. By defining $\mathcal{O}$ to be the set of linear optical maps, we can see that the set $\mathcal{O}$ is not only simple to define, it is also well characterized with a clear physical interpretation. Under this approach, nonclassicality may be interpreted as resources that overcome the limitations of classical states and linear optical operations. 

One example of a measure under this resource theory is the amount of coherent superposition between the coherent states\cite{Tan2017}. The amount of coherent superposition can be quantified via a family of coherence measures from the resource theory of coherence\cite{Streltsov2017}. These measures essentially capture quantum effects contributed by the off-diagonal elements of the density matrix. For example, in the basis $\{ \ket{0}, \ket{1} \}$, the qubit state $\rho = (\ketbra{0}+\ketbra{1})/2$ does not contain any quantum coherences because its off diagonal elements are zero while the state
$ \sigma = (\ketbra{0}+\ket{0}\bra{1}+\ket{1}\bra{0}+\ketbra{1})/2$ is said to be maximally coherent because its off-diagonal element is maximally large. By decomposing a state $\ket{\psi}$ as a superposition of a carefully chosen set of coherent states $\ket{\alpha_i}$, such that $\ket{\psi} = \sum_i c_i \ket{\alpha_i}$, one can take any continuous coherence measure $C$ from the resource theory of coherence to form a nonclassicality measure $N_C(\ket{\psi})$ by quantifying the amount of coherent superposition specified by the coefficients $c_i$. One may then show that $N_C(\ket{\psi})$ satisfies the required Properties~\ref{item::classState},\ref{item::monotonicity} and\ref{item::convexity} under the resource theory of Tan \textit{et al.}. The measure $N_C$ may be interpreted as a continuous extension of the discrete nonclassical degree Ref.~\onlinecite{Gehrke2012}, which quantifies the number of superposition rather than the amount of superposition. This also provides a bridge between the resource theory of coherence\cite{Streltsov2017} and the resource theory of nonclassicality. In fact, it was noted\cite{Tan2017} that nonclassicality in light shares many interesting characteristics with coherence, such as the close relationship and interconvertibility with entanglement\cite{Streltsov2015, Tan2018a, Tan2016, Tan2018}. The resource theories of entanglement, coherence, and nonclassicality of light therefore appear to be deeply connected, which is worth further exploring.

More recently, Ref.~\onlinecite{Tan2019} considered the extension of negativity to cover the set of all $s$-parametrized quasiprobabilities $P_s(\alpha)$ (see also Section~\ref{sec::negVolume}). They were able to show that the negativity of all such distributions $$N_s(\rho) \coloneqq \frac{1}{2}\left [\int \dd[2]\alpha \abs{P_s(\alpha)}-1 \right] $$ also belong to the resource theory of Tan \textit{et al.}\cite{Tan2017}. As $s$ decreases, $N_s(\rho)$ becomes increasingly weaker as a nonclassicality measure in the sense that the negativity decreases and fewer nonclassical states are identified by the measure. Recall that at $s=0$, we recover the Wigner negativity, which was also considered as a measure of non-Gaussianity (Section~\ref{sec::nonGaussianity}).

In Ref.~\onlinecite{Yadin2018}, Yadin {\textit et al.} also considered a resource theory where $\mathcal{O}$ is expanded to include the set of linear optical operations, plus operations allowing for the feed forward of measurement outcomes. We note that, by definition, linear optical operations belong to this expanded set of operations. As such any measure of nonclassicality under the resource theory of Yadin {\textit et al.}\cite{Yadin2018} will monotonically decrease under linear optical operations and also falls under the resource theory of Tan \textit{et al.}\cite{Tan2017}. Similar arguments can also be made for the resource theory of Gehrke \emph{et al.}\cite{Gehrke2012,Sperling2015}, as well as the recently proposed convex resource theories of non-Gaussianity\cite{Albarelli2018, Takagi2018}. We see that measures from all such resource theories necessarily falls under the resource theory of Tan \textit{et al.}\cite{Tan2017}, so this resource theory encompasses the widest range of nonclassicality measures among the resource theories discussed.

We also mention that Refs.~\onlinecite{Yadin2018,Kwon2019} recently proposed nonclassicality measures using metrological quantities. As this has to do with the concept of extracting metrological power from nonclassical states, we will discuss them later in Section~\ref{sec::estDisplacements}.

\section{Metrological power from nonclassicality} \label{sec::part2}
The second half of this paper will mainly review some elements of metrology, and how nonclassical light sources may be exploited in order to improve metrological performance beyond classical limits. In this paper, we will refer to any quantum enhancement that can be attributed solely to nonclassicality of the probe $\rho$ as metrological power. 

We begin by first introducing several key aspects of parameter estimation. 

\subsection{Elements of parameter estimation} \label{sec::paraEst}

In metrology, the most elementary problem is to perform some estimate of some unknown physical parameter. This can be treated in a very general way. Let us begin with a classical parameter estimation problem\cite{Kay1993,Lehmann1998}. Suppose we have a single unknown parameter $\theta$, which we are trying to estimate. In order to do this, we perform a measurement and obtain a set of measurement outcomes. For a given value of $\theta$, let us suppose the measurement outcomes follow a probability distribution function $f(x \mid \theta)$ that depends on $\theta$, such that $\int \dd{x} f(x \mid \theta)  = 1.$ Suppose we perform a single experiment, and the outcome is $x_1$. Based on this measurement outcome, we need to guess the value of $\theta$, which is represented by a function $t(x_1)$. The function $t(x)$ is called the estimator. Since we are trying to estimate the value of $\theta$, if $\theta$ is fixed, our guess should be correct on average if we repeat the experiment enough times. This means that we should have $\expval{t}_\theta \coloneqq \int \dd{x} f(x \mid \theta) t(x) = \theta $. An estimator $t(x)$ which satisfies  $\expval{t}_\theta = \theta$ is called an unbiased estimator.

Let us define the following quantity: 
\begin{definition}[Fisher Information] \label{def::classFishInfo}
The Fisher information is defined to be 
\begin{align*}
I(\theta) &\coloneqq \expval{\left[\pdv{\theta}\log f(x\mid \theta)\right]^2}_\theta \\
&= \int \dd{x} f(x \mid \theta) \left(\pdv{\theta}\log f(x\mid \theta)\right)^2
\end{align*}

\end{definition}

Note that the Fisher information depends on the parameter $\theta.$ We shall see that the our ultimate ability to determine what the value of the parameter $\theta$ is is largely determined by the Fisher information $I(\theta)$. This is a consequence of the famous Cram{\'e}r-Rao bound.

\begin{theorem} [Cram{\'e}r-Rao bound] \label{thm::CramerRao}
Let $t(x)$ be any unbiased estimator satisfying $\expval{t}_\theta=\theta$. Then the variance of your estimate $\Delta^2t$ satisfies the Cram{\'e}r-Rao bound $$\Delta^2t \geq \frac{1}{nI(\theta)},$$ where $n$ is the number of independent samples/experiments performed.
\end{theorem}

\begin{proof}
For compactness, let us define denote the function $L(x \mid \theta) \coloneqq \log f(x \mid \theta).$ Readers who are already somewhat familiar with parameter estimation will identify $L(x \mid \theta)$ as nothing more than the log likelihood. We will discuss more about the significance of the log likelihood later, but for now, it is just for convenience. Using this notation, we can write $I(\theta) = \ \expval{\left [ \pdv{\theta}L(x \mid \theta) \right ]^2}_\theta.$

We begin with the case where $n=1$, and we are interested to find out the minimum uncertainty of our estimate $t(x)$ based on a single experiment. Let us consider the covariance between the estimator $t(x)$ and $\pdv{\theta}L(x \mid \theta)$. Recall that the covariance between $g(x)$ and $h(x)$ is defined as $\text{Cov}[g, h] \coloneqq \expval{(g-\expval{g})(h-\expval{h})}$. Evaluating $\text{Cov}[t(x), \pdv{\theta}L(x \mid \theta)]$, we get 
\begin{align}
&\text{Cov}[t(x), \pdv{\theta}L(x \mid \theta)]  \\
&= \int \dd{x} f(x \mid \theta) \left [t(x)-\theta \right ]\left[\pdv{\theta}L(x \mid \theta) \right ]  \label{eq::crb1}\\
&= \int \dd{x} f(x \mid \theta) t(x)\pdv{\theta}L(x \mid \theta)  \label{eq::crb2}\\
&= \int \dd{x} t(x)\pdv{\theta}f(x \mid \theta) \label{eq::crb3}\\
&=  \pdv{\theta} \int \dd{x} t(x)f(x \mid \theta) \\
&= \pdv{\theta} (\theta) \label{eq::crb4} \\
&= 1.
\end{align} In Eq.~\ref{eq::crb1}, we used the assumption that $t(x)$ is an unbiased estimator $\expval{t}_\theta=\theta$, together with the fact that $\pdv{\theta} L(x \mid \theta) = \pdv{\theta}f(x\mid\theta)/f(x\mid\theta)$ and $\expval{\pdv{\theta}L(x \mid \theta)}_\theta = \int \dd{x} \pdv{\theta}f(x\mid\theta) = \pdv{\theta} (1) = 0. $. In Eq.~\ref{eq::crb2}, we again used the property $\expval{\pdv{\theta}L(x \mid \theta)}_\theta = 0$. In Eq.~\ref{eq::crb3}, we substituted in $\pdv{\theta} L(x \mid \theta) = \pdv{\theta}f(x\mid\theta)/f(x\mid\theta)$ again. In Eq.~\ref{eq::crb4}, we again used the assumption $\expval{t}_\theta=\theta$. 

The next step of the proof is the direct application of the Cauchy-Schwarz inequality. We recall that for any  probability density function $f(x)$, $\int \dd{x} f(x)g(x)h(x) \coloneqq \langle g,h \rangle$  defines an inner product. We then see that the covariance is actually an inner product of the form $\text{Cov}[g, h] = \langle g-\expval{g},h -\expval{h}\rangle$. The Cauchy-Schwarz inequality implies that  $\text{Cov}[g, h]^2 = \abs{\langle g-\expval{g},h -\expval{h}\rangle}^2 \leq \expval{\langle g-\expval{g}}^2\expval{\langle h-\expval{h}}^2 = \Delta^2g \;\Delta^2h$. Directly applying this inequality gives us 
\begin{align}
\Delta^2t\; \Delta^2[\pdv{\theta}L(x \mid \theta)] \geq \text{Cov}[t(x), \pdv{\theta}L(x \mid \theta)] = 1.
\end{align} Finally, observing that since $\expval{\pdv{\theta}L(x \mid \theta)}_\theta = 0$, we have $\Delta^2[\pdv{\theta}L(x \mid \theta)] = I(\theta)$, which leads to the required inequality for a single experiment \begin{align} \label{eq::crSingleSample}
\Delta^2t \geq \frac{1}{I(\theta)}.
\end{align} 

Finally, for the general case $n\geq 1$, consider a vector $\mathbf{x} = (x_1,\ldots, x_n)$ where the outcomes follow the probability distribution $f(\mathbf{x} \mid \theta) = f(x_1 \mid \theta)\ldots f(x_n \mid \theta),$ which is the distribution for $n$ independent samples each following the distribution $f(x_i \mid \theta)$. We can then treat the vector $\mathbf{x}$ as a single sample of the distribution $f(\mathbf{x} \mid \theta)$. Calculating the Fisher information of $f(\mathbf{x} \mid \theta)$, we get 
\begin{align}
&\expval{\left[\pdv{\theta}\log f(\mathbf{x}\mid \theta)\right]^2}_\theta \\
&=\int \dd{x_1}\ldots\dd{x_n} f(x_1 \mid \theta)\ldots f(x_n \mid \theta) \left[ \sum_i \pdv{\theta}\log f(x_i \mid \theta)\right]^2  \label{eq::crb5}\\
\begin{split} \label{eq::crb6}
&= \int \dd{x_1}\ldots\dd{x_n} f(x_1 \mid \theta)\ldots f(x_n \mid \theta)\\
& \hspace{0.17\textwidth} \left[ \sum_{i,j}  \pdv{\theta}\log f(x_i \mid \theta) \pdv{\theta}\log f(x_j \mid \theta) \right] 
\end{split} \\
&= \int \dd{x_i} f(x_i \mid \theta)\left [ \pdv{\theta}\log f(x_i \mid \theta) \right ]^2  \label{eq::crb7}\\
&= \expval{\left [ \pdv{\theta}\log f(x_i \mid \theta) \right ]^2}_\theta\\
&= nI(\theta),
\end{align} In Eq.~\ref{eq::crb5}, we used the expansion $\log f(\mathbf{x} \mid \theta) = \sum_{i=1}^n\log f(x_i \mid \theta)$. In line Eq.~\ref{eq::crb7}, we again used the property $\expval{\pdv{\theta}\log f(x_i \mid \theta)}_\theta= \int \dd{x_i} f(x_i \mid \theta) \pdv{\theta}\log f(x_i \mid \theta)) = 0$ to eliminate every term in the summation except where $i=j$. In summary, the Fisher information of $n$ independent samples is just $n I(\theta)$ where $I(\theta)$ is the Fisher information for the single sample case $n=1$. Substituting this back into Eq.~\ref{eq::crSingleSample}, we get for the general $n$ sample case $$\Delta^2t \geq \frac{1}{nI(\theta)}.$$
\end{proof}

The Cram{\'e}r-Rao bound therefore sets fundamental limits our ability to extract information about the unknown parameter $\theta$, for every possible unbiased estimator $t(x)$. The next natural question to ask is if this lower bound can be saturated.

Recall that in the proof of Theorem~\ref{thm::CramerRao}, we introduced the quantity $L(x \mid \theta) \coloneqq \log f(x \mid \theta)$. This quantity is called the logged likelihood and holds the key to a method of saturating the Cram{\'e}r-Rao bound. Suppose we perform an experiment, and we get only one sampled outcome $x_1$, what value of $t(x)$ should we choose so that we are as close to $\theta$ as possible? Intuitively, one should expect, based on what we know about  a single sample, that $x_1$ is unlikely to be a rare event. Based on this intuition, one reasonable strategy is to choose $t(x_1)$ to be the value of $\theta$ that maximizes the probability of obtaining $x_1$, i.e. we find the maximum of $f(x_1 \mid \theta)$, or equivalently, $L(x_1 \mid \theta) \coloneqq \log f(x_1 \mid \theta)$. An estimator $t(x)$ which satisfies $L(x \mid t(x)) = \sup_{\theta} L(x \mid \theta) $ for every $x$ is called the maximum likelihood estimator. Of course, intuition alone does not make this a good strategy. We can show that this estimator is in fact optimal in the asymptotic regime.

\begin{theorem} [Asymptotic reachability of Cram{\'e}r-Rao bound] \label{thm::crbAsymp}

Let $t(\mathbf{x})$ be a maximum likelihood estimator satisfying $L(\mathbf{x} \mid t(\mathbf{x})) = \sup_{\theta} L(\mathbf{x} \mid \theta) $ where $\mathbf{x} = (x_1, \ldots, x_n)$ is a vector of independent samples of size $n$. Then in the limit of sample size $n \rightarrow \infty$, the asymptotic distribution of $t$ follows a normal distribution $$ t \sim \mathbf{N}\left [\theta, \frac{1}{nI(\theta)} \right ]$$ where $\theta$ and $1/[nI(\theta)]$ are the mean and variance of the normal distribution respectively.

\end{theorem}

\begin{proof}
Suppose we have $n$ samples, which are collected in a vector $\mathbf{x} = (x_1,\ldots,x_n)$. Since all the samples $x_i$ are assumed to be independent, this means that the vector $\mathbf{x}$ follows a probability distribution of the form $f(\mathbf{x} \mid \theta) = f(x_1\mid \theta)\ldots f(x_n \mid \theta)$. We can then write the log likelihood as the sum $L(\mathbf{x} \mid \theta) = \sum_{i=1}^n L(x_i \mid \theta).$

We first perform a Taylor expansion at a point $\theta_0$ close to $\theta$, which gives us 
\begin{align} 
\begin{split} \label{eq::cbAsymp1}
&L(\mathbf{x} \mid \theta) \approx L(\mathbf{x} \mid \theta_0) + \left.\sum_{i=1}^n \pdv{L(x_i \mid \theta)}{\theta}\right \rvert_{\theta = \theta_0}(\theta - \theta_0) \\
&\hspace{0.12\textwidth} +\frac{1}{2}\sum_{i=1}^n \left. \pdv[2]{L(x_i \mid \theta)}{\theta} \right \rvert_{\theta = \theta_0} (\theta - \theta_0)^2
\end{split}
\end{align}

Maximizing the log likelihood, we seek solutions to $\pdv{L(\mathbf{x} \mid \theta)}{\theta} = 0$. Differentiating Eq.~\ref{eq::cbAsymp1}, we get 
\begin{align}
\begin{split} \label{eq::cbAsymp2}
&0 = \pdv{L(\mathbf{x} \mid \theta)}{\theta} \approx   \left.\sum_{i=1}^n \pdv{L(x_i \mid \theta)}{\theta}\right \rvert_{\theta = \theta_0}  \\
&\hspace{0.12\textwidth} +\sum_{i=1}^n \left. \pdv[2]{L(x_i \mid \theta)}{\theta} \right \rvert_{\theta = \theta_0} (\theta - \theta_0)
\end{split}
\end{align}

Since $t(\mathbf{x})$ is the maximum likelihood estimator, it should satisfy $\left. \pdv{L(\mathbf{x} \mid \theta)}{\theta} \right \rvert_{\theta = t(x)} = 0 $. In other words $t(x)$ is a solution to Eq.~\ref{eq::cbAsymp2}, so 
\begin{align} \label{eq::cbAsymp3}
0 \approx   \left.\sum_{i=1}^n \pdv{L(x_i \mid \theta)}{\theta}\right \rvert_{\theta = \theta_0} +\sum_{i=1}^n \left. \pdv[2]{L(x_i \mid \theta)}{\theta} \right \rvert_{\theta = \theta_0} (t(\mathbf{x}) - \theta_0)
\end{align}

Rearranging, we get 
\begin{align} \label{eq::cbAsymp4}
 \left [ \sum_{i=1}^n \left. \pdv[2]{L(x_i \mid \theta)}{\theta} \right \rvert_{\theta = \theta_0} \right ] (t(\mathbf{x}) - \theta_0) \approx   -\left.\sum_{i=1}^n \pdv{L(x_i \mid \theta)}{\theta}\right \rvert_{\theta = \theta_0} 
\end{align}

Let us consider the term $\sum_{i=1}^n \left. \pdv[2]{L(x_i \mid \theta)}{\theta} \right \rvert_{\theta = \theta_0}$ on the left hand side of the equation. Assuming $n \gg 1$, then we can expect, using the law of large numbers, that $n f(x \mid \theta) \dd{x}$ of the elements in the list $(x_1,\dots, x_n)$ to lie within the region $x+\dd{x}$ for every $x$. This means that 
\begin{align*}\sum_{i=1}^n \left. \pdv[2]{L(x_i \mid \theta)}{\theta} \right \rvert_{\theta = \theta_0} & \approx n \int \dd{x} \left. f(x \mid \theta) \pdv[2]{L(x \mid \theta)}{\theta} \right \rvert_{\theta = \theta_0} \\
&= n\left. \expval{\pdv[2]{L(x \mid \theta)}{\theta}}_\theta \right \rvert_{\theta = \theta_0} \\
&=-nI(\theta_0),
\end{align*} where the final equality above can be directly computed using the identities $L(x \mid \theta) = \log f(x \mid \theta)$ and $\expval{\pdv{\theta}L(x \mid \theta)}_\theta = 0.$

For the term $ \left.\sum_{i=1}^n \pdv{L(x_i \mid \theta)}{\theta}\right \rvert_{\theta = \theta_0} $ on the right hand side, we will use the central limit theorem, which says that for sufficiently large $n$,  $\sum_{i=1}^n \pdv{L(x_i \mid \theta)}{\theta}$ will approximately follow a normal distribution with mean $ \left. \expval{\sum_{i=1}^n \pdv{L(x_i \mid \theta)}{\theta}}_\theta \right \rvert_{\theta = \theta_0} $ and variance $\Delta^2 \left [ \left. \sum_{i=1}^n \pdv{L(x_i \mid \theta)}{\theta}\right ] \right \rvert_{\theta = \theta_0}$.  Direct calculation will verify that the mean is zero, while the variance is $nI(\theta_0)$. Putting this back into Eq.~\ref{eq::cbAsymp4}, we get $$nI(\theta_0) (t - \theta_0) \sim \mathbf{N}(0, nI(\theta_0))$$ which we can further simplify to get $$t \sim \mathbf{N}(\theta_0, \frac{1}{nI(\theta_0)}).$$ So we see that for large enough $n$, $t(\mathbf{x})$ follows a Gaussian distribution and  has variance $\Delta^2 t = 1/[nI(\theta_0)]$, which saturates the Cram{\'e}r-Rao bound. This means that in the asymptotic limit of $n \rightarrow \infty$, the Cram{\'e}r-Rao bound can always be saturated, and the optimal strategy is a maximum likelihood estimator.
\end{proof}

Theorem~\ref{thm::crbAsymp} illustrates how the Cram{\'e}r-Rao bound is in fact reachable, so long as a sufficient number of independent experiments are performed, and a sufficient number of data points are gathered. The fact that the bound can be saturated allows us to directly quantify how useful a given statistical distribution $f(x\mid \theta)$ is for the estimation of an unknown parameter $\theta$ via the Fisher information $I(\theta)$. We just have to keep in mind that we need to make many repeated measurements in order to make this connection.

\subsection{Elements of quantum metrology} \label{sec::qMetrology}

Thus far, the problem of parameter estimation has revolved around around what is essentially a classical information processing problem -- there is some probability distribution that depends on $\theta$, and we figure out what are the best ways to extract information about $\theta$ from the classical statistics.

This section will introduce quantum mechanical elements to the parameter estimation problem. The most fundamental element of quantum metrology is the {\it probe} which is represented by some density operator $\rho$. The parameter $\theta$ which we are interested to measure is encoded onto some quantum channel $\Phi_\theta$. Information about $\theta$ is extracted by passing the state $\rho$ through the quantum channel $\Phi_\theta$, resulting in the transformation of state $\Phi_\theta(\rho) \coloneqq \rho_\theta.$

Information about $\theta$ is therefore imprinted onto the probe $\rho_\theta$. In order to perform our estimate of $\theta$, we perform a measurement on $\rho_\theta$, which is represented by some set of positive operator value measures (POVM)\cite{NielsenChuang} $ M \coloneqq \{ \Pi_x \}$ satisfying $\Pi_x \geq 0$ and $\int \dd{x}\Pi_x = \openone.$ By performing a measurement, we obtain the statistical distribution $\Tr(\Pi_x \rho_\theta) = f(x \mid \theta).$ In principle, this is the end of the quantum aspect of quantum metrology. After performing the measurement and obtaining the statistics, what remains is to perform your best estimate of $\theta$ given $f(x \mid \theta)$, which is the standard parameter estimation problem described in the previous section. 

There is an infinite repertoire of possible POVMs that we can consider in quantum mechanics. It is therefore natural to ask what is the optimal measurement $M$ that we should perform on the state $\rho_\theta$. It is somewhat of a small miracle that this question can actually be answered using only fairly elementary arguments.

In order to address the previous question properly, we introduce an operator call the symmetric logarithmic derivative\cite{Helstrom1967, Helstrom1968}. 
\begin{definition} [Symmetric logarithmic derivative]
Consider the eigendecomposition of $\rho_\theta$ such that $\rho_\theta = \sum_i \lambda_i \ketbra{i}$ and $\bra{i}\ket{j} = \delta_{ij}$. The symmetric logarithmic derivative of $\rho_\theta$ is the operator $D_\theta$ that satisfies the equation $$\pdv{\theta} \rho_\theta = \acomm{\rho_\theta}{D_\theta}/2,$$ where $\acomm{\cdot}{\cdot}$ is the anticommutator $\acomm{A}{B}\coloneqq AB+BA.$
\end{definition} We are guaranteed that the solution $D_\theta$ will always exist as for any matrix $A = \acomm{\rho_\theta}{D_\theta}/2$, we can verify by direct substitution that $D_\theta = 2\sum_{i,j}[A_{ij}/(p_i+p_j)]\ketbra{i}{j}$ is a solution. Furthermore, we see that as $\rho$ and hence $\pdv{\theta} \rho_\theta$ are both Hermitian, $D_\theta$ is also Hermitian.

Based on the symmetric logarithmic derivative, we can then introduce the quantum  Fisher information.
 
\begin{definition} [Quantum Fisher information] \label{def::genQFI}
For a given symmetric logarithmic derivative $D_\theta$, the quantum Fisher information is defined as the quantity
$$I_Q(\rho, \theta) \coloneqq \Tr(\rho_\theta D_\theta^2),$$ and $\Phi_\theta(\rho) \coloneqq \rho_\theta.$
\end{definition}

The quantum Fisher information $I_Q(\rho, \theta)$ is given a physical significance via the following theorem, which is quantum version of the Cram{\'e}r-Rao bound\cite{Helstrom1976, Holevo1982, Braunstein1994}.

\begin{theorem} [Quantum Cram{\'e}r-Rao bound] \label{thm::qCramerRao}
Let $ M \coloneqq \{ \Pi_x \}$ be any measurement satisfying $\int \dd{x}\Pi_x = \openone,$ and $f(x \mid \theta) \coloneqq \Tr(\rho_\theta \Pi_x)$ where $\rho_\theta \coloneqq \Phi_\theta(\rho)$. For a given probe $\rho$ and measurement $M$  we denote the Fisher information of the probability  distribution $f(x \mid \theta)$ as $I(\theta \mid \rho, M)$.

Then for any $M$, we have $$I(\theta \mid \rho, M) \leq I_Q(\rho,\theta) \coloneqq \Tr(\rho_\theta D_\theta^2). $$

This directly implies that for any unbiased estimator $t(x)$ we have $$\Delta^2t \geq \frac{1}{nI_Q(\rho, \theta)},$$ where $n$ is the number of independent samples/experiments performed. 
\end{theorem}

\begin{proof}
Recall that for a given $M$, the Fisher information is $I(\theta \mid \rho, M) \coloneqq \expval{[\pdv{\theta}L(x \mid \theta)]^2}_\theta$ where $L(x \mid \theta)$ is the likelihood function $\log f(x \mid \theta).$

We see that based on the definition of $D_\theta$, we have
\begin{align}
I(\theta \mid \rho, M) &= \expval{[\pdv{\theta}\log f(x \mid  \theta) ]^2}_\theta  \\ 
&=\expval{[\pdv{\theta}\log \Tr(\rho_\theta \Pi_x) ]^2}_\theta  \label{eq::qcrb1}\\
&=\expval{[\Tr(\pdv{\theta}\rho_\theta \Pi_x)/\Tr(\rho_\theta \Pi_x) ]^2}_\theta \label{eq::qcrb2}\\
&\leq \int \dd{x} [\abs{\Tr(\rho_\theta D_\theta \Pi_x)}^2/\Tr(\rho_\theta \Pi_x) ] \label{eq::qcrb3} \\
&= \int \dd{x} \abs{\Tr(\sqrt{\Pi_x} \sqrt{\rho_\theta}  \sqrt{\rho_\theta} D_\theta \sqrt{\Pi_x})}^2/\Tr(\rho_\theta \Pi_x)  \label{eq::qcrb4}\\
&\leq \int \dd{x} [\Tr(\Pi_x \rho_\theta)  \Tr( \Pi_x D_\theta  \rho_\theta D_\theta )/\Tr(\rho_\theta \Pi_x) ] \label{eq::qcrb5}\\
&= \int \dd{x}  \Tr( \Pi_x D_\theta  \rho_\theta D_\theta ) \\
&=  \Tr(  \rho_\theta D_\theta^2 ). \label{eq::qcrb6}
\end{align} In Eq~\ref{eq::qcrb1}, we used the identity $f(x \mid \theta) \coloneqq \Tr(\rho_\theta \Pi_x)$. Eq~\ref{eq::qcrb2}, directly results from computing the partial derivative. In Eq~\ref{eq::qcrb3}, we substituted the expression $\pdv{\theta} \rho_\theta = \acomm{\rho_\theta}{D_\theta}/2$. We then set $A = \rho_\theta D_\theta$ and observe that since $\Pi_x$ is positive, $\Tr(A \Pi_x) + \Tr(A^\dag \Pi_x) = 2\Re[\Tr(A \Pi_x)] \leq 2 \abs{\Tr(A \Pi_x)}.$ In Eq~\ref{eq::qcrb4}, we used the cyclic property of the trace to write $\Tr(\rho_\theta D_\theta \Pi_x) = \Tr(\sqrt{\Pi_x} \sqrt{\rho_\theta}  \sqrt{\rho_\theta} D_\theta \sqrt{\Pi_x}).$ In Eq~\ref{eq::qcrb5}, we used the Cauchy-Schwarz inequality for the Hilbert-Schmidt norm $\abs{\Tr(A^\dag B)} \leq \Tr(A^\dag A)\Tr(B^\dag B)$ and set $\sqrt{\Pi_x} \sqrt{\rho_\theta}$ and $B= \sqrt{\rho_\theta} D_\theta \sqrt{\Pi_x}).$. Finally, in Eq~\ref{eq::qcrb6}, we used the identity $\int \dd{x}\Pi_x = \openone,$ which gives us the required inequality and proves the first part of the theorem. The inequality $\Delta^2t \geq 1/[nI_Q(\rho, \theta)]$ then follows directly from Theorem~\ref{thm::CramerRao}. 
\end{proof}

The quantum Cram{\'e}r-Rao bound extends the result of Theorem~\ref{thm::CramerRao} to the quantum regime. It sets ultimate limits on our ability to extract information about an unknown variable $\theta$ via a quantum measurement. We see that this bound does not depend on the measurement $M$ being performed, but does depend on the probe $\rho$, as well as the unknown parameter $\theta$.

We can show that the quantum Cram{\'e}r-Rao bound may always
 be saturated by some measurement\cite{Helstrom1968, Braunstein1994}, at least in principle. Recall that the symmetric logarithmic derivative $D_\theta$ is a Hermitian matrix, which is diagonalizable. As such, we can consider its eigendecomposition $D_\theta = \sum_x \lambda_{\theta, x} \ketbra{\phi_{\theta,x}}$ and choose the measurement $M = \{ \Pi_x = \ketbra{\phi_{\theta,x}} \}$, which is a projective measurement onto the eigenbasis of $D_\theta$. One may then directly verify by substitution into Eq.~\ref{eq::qcrb2} that $I(\theta \mid \rho, M) = \Tr(\rho_\theta D_\theta^2)$. Since the bound can be saturated, this suggests that the quantum Fisher information $I_Q(\rho, \theta)$ precisely quantifies the usefulness of a probe $\rho$ for the measurement of a given $\theta$.

There are however, several important caveats to keep in mind. First, while the bound may be saturated via a projection onto the eigenbasis of $D_\theta$, this by itself does not inform us of a way to physically implement the measurement in a laboratory. The optimal measurement is also in general not unique, and more technologically feasible measurements may exist.

Second, note that in general both $I_Q(\rho,\theta)$ and $D_\theta$ depends on the value of $\theta$. There is therefore no guarantee that a single fixed measurement $M$ will be able to saturate the  quantum Cram{\'e}r-Rao bound for every value of $\theta.$\cite{Cochran1973, Barndorff2000} In some sense, this suggests that we need to somehow know the value of $\theta$ before we can decide what measurement to perform, which clearly goes against our initial objective of measuring some unknown but fixed quantity $\theta$. This issue is surmountable, however, by considering adaptive schemes\cite{Wiseman1995,Berry2000, Berry2002,Armen2002} performed over multiple measurements. Conditioned on prior measurement outcomes, the measurement $M$ can be made to eventually converge to the optimal case over a sufficiently large number of correlated experiments\cite{Fujiwara2006, Fujiwara2011}. Recall from Theorem~\ref{thm::crbAsymp} that the Fisher information can be saturated under the assumption that a large number $n$ of independent experiments are performed. In the quantum case, the situation is more complicated because it may be necessary to perform some adaptive scheme over a large number of correlated experiments to allow $M$ to converge first. One may then subsequently obtain independent samples using the optimal $M$ to saturate the quantum Cram{\'e}r-Rao bound.

For similar reasons, the dependence on $\theta$ implies there is no guarantee that a given probe $\rho$ will equally useful for every value of $\theta$ except in special cases. For unitary evolutions however, this turns out to not be an issue as the Fisher information can be shown to be the same along any point in the probe's unitary orbit. We therefore see that the interpretation of quantum Fisher information as a measure of a probe $\rho$'s usefulness for metrology is especially well suited for unitary encodings. This will be further discussed in the subsequent section.

One may also remove potential issues arising from the dependence of $I_Q(\rho,\theta)$ on $\theta$ by assuming that $\theta$ is unknown but varies over only a very small region in the vicinity of some value $\theta_0$. This is the local estimation approach, where one effectively only considers $I_Q(\rho,\theta_0)$ since $\theta \approx \theta_0$. The Fisher information then becomes solely a function of the probe $\rho$. Physically, it corresponds to the high precision measurement regime, where we are only interested in measuring very small differences in physical parameters. This allows us to generally interpret $I_Q(\rho,\theta_0)$  as a measure of the usefulness of the probe $\rho$ for high precision measurements. However, this presupposes strong \textit{a priori} knowledge about the distribution of $\theta$ before hand, and such an assumption may not always be valid.

\subsection{Unitary quantum metrology} \label{sec::unitaryQMet}

In the previous section, we discussed quantum metrology in very general terms, where the the quantum channel $\Phi_\theta$ may in general be any quantum map. The corresponding quantum Fisher information $I(\rho, \theta)$ is difficult to compute under such general scenarios. 

We can however, greatly simplify the problem by considering only unitary encodings. Suppose that $\Phi_\theta(\rho) = \rho_\theta = U_\theta \rho U^\dag_\theta$, where $U = e^{-i\theta G}$. $G$ is a Hermitian operator, and is sometimes called the generator of the unitary transformation. Writing $\rho = \sum_i p_i \ketbra{i}$ in its diagonal form, we can directly evaluate the symmetric logarithmic derivative: 

\begin{align*}
D_\theta &= 2\sum_{i,j}[\mel{i}{\pdv{\rho_\theta} {\theta}}{j}/(p_i+p_j)]\ketbra{i}{j} \\
&= 2i\sum_{i,j}[\mel{i}{\comm{\rho}{G}}{j}/(p_i+p_j)]\ketbra{i}{j} \\
&= 2i\sum_{i,j}[(p_i-p_j)/(p_i+p_j)]\mel{i}{G}{j}\ketbra{i}{j},
\end{align*} where we used the von Neumann equation $i \pdv{\rho_\theta} {\theta} = [G, \rho].$

We can use this to evaluate the quantum Fisher information, resulting in the following series of inequalities:
\begin{align*}
\Tr(\rho D_\theta^2) &= 4\sum_{i,j}p_i\frac{(p_i-p_j)^2}{(p_i+p_j)^2}\abs{\mel{i}{G}{j}}^2 \\
&= 4\sum_{i,j}\frac{(p_i-p_j)^2}{p_i+p_j}\abs{\mel{i}{G}{j}}^2\frac{p_i}{p_i+p_j} \\
&= 4\sum_{i,j}\frac{(p_i-p_j)^2}{(p_i+p_j)}\abs{\mel{i}{G}{j}}^2 \left (1- \frac{p_j}{p_i+p_j} \right ) \\
&= 4\sum_{i,j}\frac{(p_i-p_j)^2}{p_i+p_j}\abs{\mel{i}{G}{j}}^2 - \Tr(\rho D_\theta^2).
\end{align*}

This leads to the following definition of Fisher information for unitary processes\cite{Helstrom1976, Holevo1982, Braunstein1994, Braunstein1996}.

\begin{definition} [Quantum Fisher information, unitary encoding] \label{def::unitQFI}
For any unitary $U_\theta = e^{-i\theta G}$ with generator $G$, the quantum Fisher information is $$I_Q(\rho, G) \coloneqq 2\sum_{i,j}\frac{(p_i-p_j)^2}{p_i+p_j}\abs{\mel{i}{G}{j}}^2,$$ where $p_i$ and $\ket{i}$ are the eigenvalues and eigenvectors of $\rho$.
\end{definition}

Notice that we have dropped the dependence on $\theta$, compared to the more general version of the quantum Fisher information in Definition~\ref{def::genQFI}. This is because the Fisher information is actually invariant under the unitary $U_\theta$. It is not difficult to verify that this is true. If $\ket{i}$ is the eigenvector of $\rho$, then $U_\theta\ket{i}$ is the eigenvector of $\rho_\theta$. However, since $U_\theta = e^{-i\theta G}$ and  $[U_\theta, G] = 0$, we have $\mel{i}{U_\theta^\dag GU_\theta}{j}=\mel{i}{G}{j}$. This shows that the Fisher information is always constant along for every $\rho_\theta$. For general quantum channels, this property does not necessarily hold (see discussion at end of Section~\ref{sec::qMetrology}). 

Below is a collection of some elementary properties\cite{Toth2014} of $I_Q(\rho, G)$:

\begin{gather*}
I_Q(\ket{\psi}, G) = 4\Delta^2 G, \; \text{where $\ket{\psi}$ is a pure state} \\ 
I_Q(\rho, G) \leq 4\Delta^2 G, \; \text{for general mixed state $\rho$} \\ 
I_Q(e^{-\theta G}\rho e^{\theta G}, G ) =  I_Q(\rho , G )\\
I_Q(U\rho U^\dag, G ) =  I_Q(\rho,U^\dag G U ), \; \text{where $U$ is unitary}  \\
 I_Q(\sum_i p_i\rho_i, G) \leq \sum_i p_i I_Q(\rho_i, G) , \; \text{where } \sum_i p_i =1 \\
I_Q(\rho_1 \otimes \sigma_2, G_1\otimes \openone_2 + \openone_1 \otimes H_2) =  I_Q(\rho_1 , G_1 ) +I_Q(\sigma_2 , H_2 ) \\ 
I_Q(\oplus_i p_i \rho_i , \oplus_i G_i) = \sum_i p_i I_Q(\rho_i , G_i), \; \text{where  $\Tr(\rho_i)=1$}  \\
I_Q(\rho_{12} , G \otimes \openone_2  )  \geq I_Q(\Tr_2(\rho_{12}) , G )\\
\end{gather*} 

For unitary dynamics where the eigenvalues of $G$ is bounded, one may also additionaly identify the optimal quantum states maximizing the quantum Fisher information\cite{Braunstein1996,Giovannetti2006}. The optimal probe in this case is an equal superposition of the form $\ket{\psi} = (\ket{\lambda_{\max}} + \ket{\lambda_{\min}})/\sqrt{2}$, where $\ket{\lambda_{\max}}$ and $\ket{\lambda_{\min}})$ are the eigenvectors corresponding to maximum and minimum eigenvalues respectively.

Finally, a recent result\cite{Yu2013, Toth2013} proved that for unitary encodings, the quantum Fisher information is the convex roof of the variance of $G$: 
\begin{align} \label{eq::convexRoofFisher}
I_Q(\rho, G)= 4 \min_{\{ p_i, \ket{\psi_i} \}} \sum_i p_i \Delta^2_{\ket{\psi_i}} G, 
\end{align} 
where the minimization is over all possible pure states decompositions $\{ p_i, \ket{\psi_i} \}$ satisfying $\rho =\sum_i p_i \ketbra{\psi_i}$ and $\Delta^2_{\ket{\psi_i}} G$ is the variance of $G$ for the state $\ket{\psi_i}$. In this case, one may interpret the convex roof as the useful "quantum" part of the variance that is left over after statistical mixing.

The extraction of an unknown parameter from a unitary quantum channel of this type is probably the most well studied and understood of all the problems in quantum metrology. In the next section, we will discuss several physically relevant examples of such unitary channels, which provides strong evidence that nonclassical light is a useful resource in making precision measurements.

\subsection{Extracting metrological power from nonclassical states}

In this section, we will discuss the role that nonclassical states play in quantum enhanced metrology. 

Our general strategy to demonstrate that metrological power may be extracted is quite simple. Suppose we have a generator $G$ with corresponding Fisher information $I_Q(\rho, G)$. If we are able to demonstrate the existence of some state $\rho$ satisfying $I_Q(\rho, G) \geq \sup_{\rho_\text{cl}} I_Q(\rho_\text{cl}, G),$ where the optimization is over the set of classical states $\rho_\text{cl}$, then clearly $\rho$ must be nonclassical and nonclassicality can be exploited to improve measurement precision. A similar strategy was also employed in Ref~\onlinecite{Rivas2010}, where the quantum Fisher information was used to test whether a state is nonclassical. We begin by demonstrating this possibility for a parameter estimation problem called phase estimation. 

\subsubsection{Single mode phase estimation} \label{sec::phaseEst}

Let us consider a very simple choice for the generator $G$. Recall the number operator $n= a^\dag a$, which is the Hermitian observable measuring the number of photons in a system. Let us choose $G = n/2$. From this, we can construct the unitary encoding $U_\theta = e^{-i\theta n/2 },$ which causes a clockwise rotation of angle $\theta/2$ in phase space (see also Section~\ref{sec::cohStates}). One may also verify that $U_\theta^\dag a U_\theta = e^{-i\theta/2}a $. $U_\theta \rho U^\dag_\theta$ therefore induces a change in the phase of $\rho$ relative to some reference clock. For this reason, we can call the problem of measuring the parameter $\theta$ quantum phase estimation. 

Let us consider the quantum Fisher information for a coherent state $\ket{\alpha}$. For pure states, this is just four times the variance of the observable $G = n/2$ (see Section~\ref{sec::unitaryQMet}). Since the number distribution of the coherent state is a Poisson distribution, the variance and mean of the number distribution is the same, so we have $4\Delta^2_{\ket{\alpha}}( n/2) = \Delta^2_{\ket{\alpha}}\, n = \expval{n}{\alpha}$. The astute reader may have wondered about the $1/2$ factor in $G = n/2$. For the moment, it is just for convenience as it removes the constant factor $4$ from the quantum Fisher information, but we shall see that a similar factor will also appear in the problem of interferometry, which we will discuss in the subsequent section. 

Now, let us consider a classical mixed state $\rho_\text{cl} = \int \dd[2]{\alpha} P_\text{cl}(\alpha) \ketbra{\alpha}$, where $P_\text{cl}(\alpha)$ is a positive probability distribution function. Note that since $P_\text{cl}(\alpha)$ is a proper probability distribution, $\rho_\text{cl} = \int \dd[2]{\alpha} P_\text{cl}(\alpha) \ketbra{\alpha}$ is an example of a pure state decomposition of the state $\rho_\text{cl}$. Recall from Eq.~\ref{eq::convexRoofFisher} that the quantum Fisher information of a mixed state is actually four times the convex roof of the variance, i.e. the minimum average variance over all possible pure state decompositions. As such we may write $$I_Q(\rho_\text{cl}, n/2) \leq \int \dd[2]{\alpha}P_\text{cl}(\alpha) \expval{n}{\alpha} = \Tr(\rho_\text{cl} n ) = \expval{n}_{\rho_\text{cl}}.$$ If we apply the Quantum Cram{\'e}r-Rao bound (Theorem~\ref{thm::qCramerRao}), we then get the following lower bound on the standard deviation for our unbiased estimate $t$:
\begin{align} \label{eq::shotNoiseLimit}
\Delta t \geq \frac{1}{\sqrt{\expval{n}_{\rho_\text{cl}}}}
\end{align}

The above expression is called the standard quantum limit, also sometimes called the shot noise limit. It essentially states that for classical light sources the measurement precision scales with the inverse square root of the mean photon number at best. Since the mean photon number reflects the energy content of your light source, one may reinterpret this to mean that for classical light sources, energy can be traded for measurement precision. Importantly, since for classical states $I_Q(\rho_\text{cl}, n/2)$ scales with the mean photon number $\expval{n}_{\rho_\text{cl}}$ at best, $I_Q(\rho, n/2) > \expval{n}_{\rho}$ implies that $\rho$ must be nonclassical and useful metrological power may be extracted from it.

We recall from the list of elementary properties in Section~\ref{sec::unitaryQMet} that the quantum Fisher information is bounded by the variance of $G$, i.e. $I_Q(\rho, G) \leq 4 \Delta^2 G$. For $G=n/2$, this translates to  $I_Q(\rho, n/2) \leq \Delta_{\rho}^2 \, n $. We can combine this with the condition for nonclassicality $I_Q(\rho, n/2) > \expval{n}_{\rho}$ to obtain the following necessary, but insufficient, condition to beat the standard quantum limit: $$ \frac{\Delta_{\rho}^2 \, n}{\expval{n}_{\rho}}-1 \geq 0.$$

Curiously, this is exactly the set of nonclassical states which are \emph{not} identified by the Mandel Q parameter (see Section~\ref{sec::mandelQ}. Any nonclassical state that is potentially useful for this phase estimation problem must be super-Poissonian. 

Let us consider the subspace spanned by the Fock states $\{ \ket{0}, \ldots, \ket{n_{\max}}\}$ for some finite $n_{\max}.$ Within this subspace, $n$ is a bounded operator, so the state $\ket{\psi} = \frac{1}{\sqrt{2}}(\ket{0}+\ket{n_{\max}} )$ must maximize $I_Q(\ket{\psi}, n)$\cite{Braunstein1996,Giovannetti2006}. One may verify that $\expval{n}_{\ket{\psi}} = n_{\max}/2$ and that $\Delta^2 n = n_{\max}^2/4 = \expval{n}_{\ket{\psi}}^2 = I_Q(\ket{\psi}, n/2)$. This leads to the following bound on the unbiased estimate:
\begin{align} \label{eq::heisenbergLimit}
\Delta t \geq \frac{1}{\expval{n}_{\rho}}.
\end{align} 

Eq.~\ref{eq::heisenbergLimit} is referred to the Heisenberg limit. The name is somewhat of a misnomer as it is not really a fundamental quantum limit. For $N$ distinguishable particles where every particle may be addressed, one may indeed make general arguments to demonstrate that the quantum limit\cite{Giovannetti2006} is $~1/N$, but for systems of identical particles such arguments do not apply. Indeed, quantum states beating the Heisenberg limit in Eq.~\ref{eq::heisenbergLimit} have been studied. Somewhat confusingly, strategies beating the Heisenberg limit are sometimes said to have achieved sub-Heisenberg sensitivity\cite{Anisimov2010, Rivas2012, Zhang2012}. There are doubts as to whether such sub-Heisenberg strategies are truly useful, as the Heisenberg limit is retrieved once the performance is averaged over all possible values of $\theta$\cite{Berry2012}. There are also arguments suggesting that without prior knowledge of $\theta$, the advantages of sub-Heisenberg strategies, while possible, are limited\cite{Giovannetti2012}. In any case, a more general quantum limit\cite{Hofmann2009, Zhang2012} is $$\Delta t \geq \frac{1}{\sqrt{\expval{n^2}_\rho}}.$$ Note that the square is on the number operator $n$, and not the expectation value. In general $\expval{n^2}_\rho \neq \expval{n}_\rho^2.$ This comes from the observation that $I_Q(\rho, n/2) \leq \Delta_\rho^2 \, n \leq \expval{n^2}_\rho$. While the expression appears superficially similar to Eq.~\ref{eq::heisenbergLimit}, it cannot be interpreted directly as the energy of the system. Quantum advantages in phase estimation problems are typically compared in an energy adjusted scenario.

One doesn't have to look far to find examples of sub-Heisenberg sensitivity. Consider the squeezed vacuum state $\ket{0, \abs{\epsilon}}.$ Its mean photon number and variance can be verified to be $\expval{n}_{\ket{0, \abs{\epsilon}}} = \sinh^2(\abs{\epsilon})$ and $\Delta^2 \, n = 2\cosh^2(\abs{\epsilon})\sinh^2(\abs{\epsilon})= 2(\expval{n}_{\ket{0, \abs{\epsilon}}}^2+\expval{n}_{\ket{0, \abs{\epsilon}}}) = I_Q(\ket{0, \abs{\epsilon}}, n/2)$. Clearly, we have $I_Q(\ket{0, \abs{\epsilon}}, n/2) > \expval{n}_{\ket{0, \abs{\epsilon}}}^2$, so the Heisenberg limit has been exceeded.

Table~\ref{tab:comparisonPhaseEst} compares the achievable Fisher information for coherent states, Fock states, squeezed states and cat states that were introduced in Section~\ref{sec::examplesPfun}. Among the states compared, we see that the squeezed vacuum and even cat states are able to beat the standard quantum limit. The Fock and odd cat states are unable to do so despite being nonclassical. This is because they are sub-Poissonian. See Section~\ref{sec::mandelQ}. 

\begin{table}
\caption{\label{tab:comparisonPhaseEst} A comparison of the achievable quantum Fisher information for coherent states, Fock states, squeezed vacuum, even cat states and odd cat states in single mode phase estimation. See Section~\ref{sec::examplesPfun} for more detailed discussion of such states. }
\begin{ruledtabular}
\begin{tabular}{lcr}
State $\ket{\psi}$ & Fisher Information $I_Q(\ket{\psi}, n/2)$ \\
\hline
Coherent states, $\ket{\alpha}$ & $\expval{n}_{\ket{\alpha}}$ \\
Fock States, $\ket{n}$ & $0$\\
Squeezed vacuum, $\ket{0, \abs{\epsilon}}$  & $2(\expval{n}_{\ket{0, \abs{\epsilon}}}^2+\expval{n}_{\ket{0, \abs{\epsilon}}})$ \\
Even cat states, $ \ket{\mathrm{\psi_{+}}}$ & $\expval{n}_{\ket{\mathrm{\psi_{+}}}} + \abs{\beta}^4\sech^2\abs{\beta^2} $\\
Odd cat states, $\ket{\mathrm{\psi_{-}}}$ & $ \expval{n}_{\ket{\mathrm{\psi_{-}}}} - \abs{\beta}^4\csch^{2}\abs{\beta^2} $\\

\end{tabular}
\end{ruledtabular}
\end{table}

\subsubsection{Optical interferometry}

In Section~\ref{sec::phaseEst}), we introduced the single mode phase estimation problem, where the goal is to perform precise measurements of the change in phase $\theta/2$ in a single mode of light. We now extend the problem slightly and consider a two mode setup. Let the corresponding creation and annihilation operators of the first and second modes be $a^\dagger,a$ and $b^\dagger,b$. Their respective number operators are then $n_a \coloneqq a^\dag a$ and $n_b \coloneqq b^\dag b$. The total number operator is then $n_\text{total} = n_a+n_b$
We then choose the generator to be $G= (n_a - n_b)/2$, resulting in the unitary evolution $U_\theta = e^{-i\theta n_a / 2 }e^{i\theta n_b / 2 }$.

We see that this unitary dynamic corresponds to an clockwise rotation in the phase space of mode $a$ by an angle of $\theta/2$, together with another rotation of angle $\theta/2$ in the opposite, anti-clockwise direction in mode $b$. The angle $\theta$ then measures the \emph{relative} phase difference between the two modes. This mirrors the situation modelled by a Mach-Zehnder interferometer, which is shown in Fig.~\ref{fig::MZI}, where the interaction in the middle box is funtionally equivalent to $U_\theta = e^{-i\theta n_a / 2 }e^{i\theta n_b / 2 }$. We now consider the quantum Fisher information for several classes of states under the unitary dynamic $U_\theta$. This corresponds to calculating the Fisher information of the state \emph{after} exiting the first beam splitter of the interferometer, which is represented by $\ket{\psi_\text{mid}}$ in Fig.~\ref{fig::MZI}.

\begin{figure}
\includegraphics[width = 0.9\linewidth ]{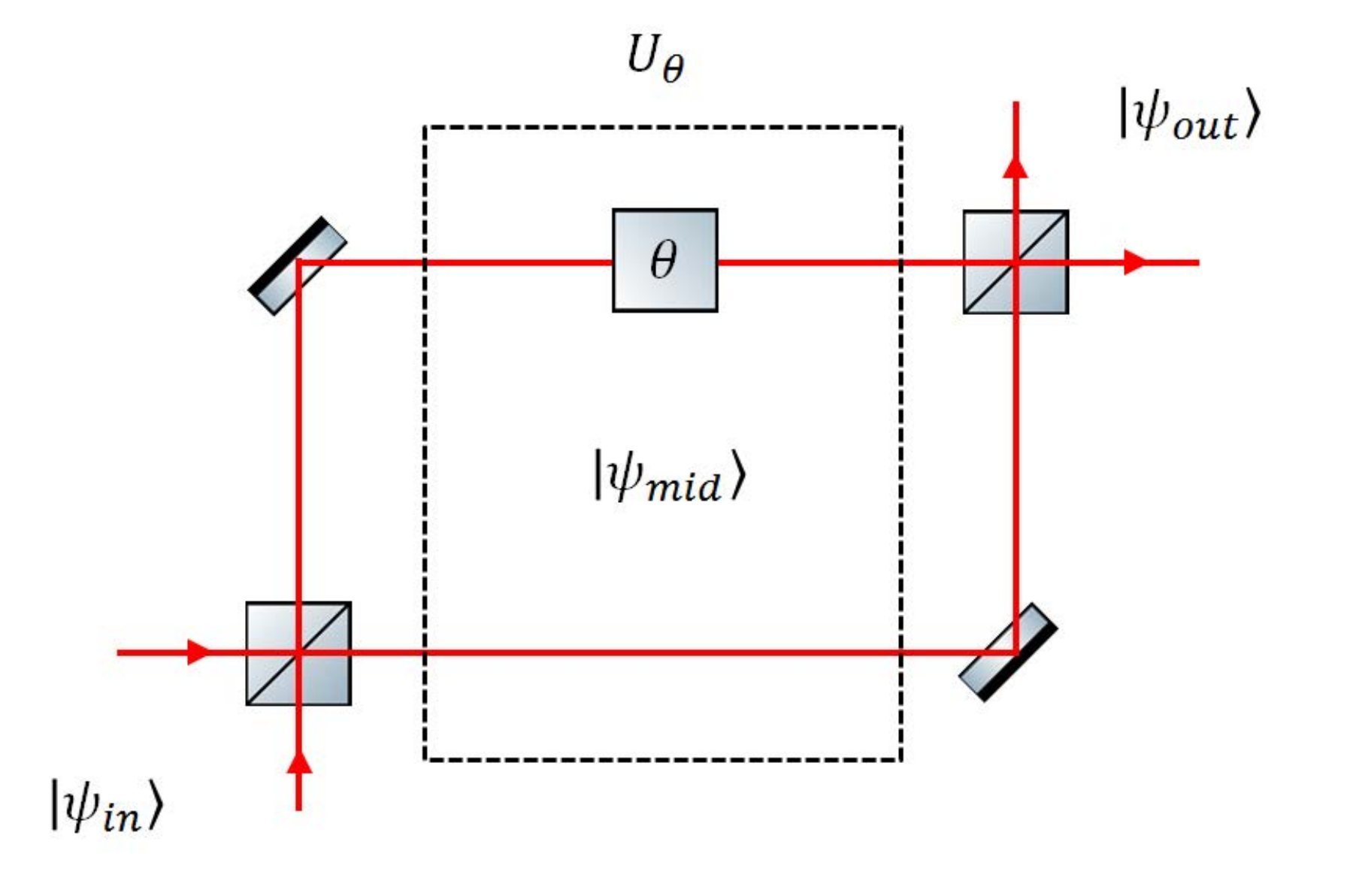}
\caption{\label{fig::MZI} A Mach-Zehnder interferometer. A two mode input state $\ket{\psi_\text{in}}$ enters a beam splitter. Within the interferometer, the beam $\ket{\psi_\text{mid}}$ experiences a phase shift of $\theta$ in the upper path relative to the lower path. The beams exit the interferometer after passing through a second beam splitter. The goal is to perform a measurement on $\ket{\psi_\text{out}}$ to estimate the value of $\theta$.}
\end{figure}

We first perform a similar analysis that was performed for the single mode case over the set of classical states $\rho_\text{cl} = \int \dd[2]{\alpha}\dd[2]{\beta} P_\text{cl}(\alpha,\beta) \ketbra{\alpha} \otimes \ketbra{\beta}$, where $P_\text{cl}(\alpha,\beta)$ is a positive probability distribution function. For a product of coherent states $\ket{\alpha}\ket{\beta}$, we can verify that $I_Q(\ket{\alpha}\ket{\beta}, G) = \Delta^2_{\ket{\alpha}\ket{\beta}} G = (\Delta^2_{\ket{\alpha}} n_a +\Delta_{\ket{\beta}}^2 n_b)/4 = \expval{n_\text{total}}_{\ket{\alpha}\ket{\beta}}/4$. We then observe that since $P_\text{cl}(\alpha,\beta)$ is just a positive classical distribution, the expression $\rho_\text{cl} = \int \dd[2]{\alpha}\dd[2]{\beta} P_\text{cl}(\alpha,\beta) \ketbra{\alpha} \otimes \ketbra{\beta}$ is just a pure state decomposition in terms of products of coherent states. We then combine this with Eq.~\ref{eq::convexRoofFisher} to obtain: 
\begin{align*}
I_Q(\rho_\text{cl}, G) &\leq 4 \int \dd[2]{\alpha}\dd[2]{\beta} P_\text{cl}(\alpha,\beta)  \expval{n_\text{total}}_{\ket{\alpha}\ket{\beta}}/4 \\
&=\int \dd[2]{\alpha}\dd[2]{\beta} P_\text{cl}(\alpha,\beta)  \bra{\alpha}\bra{\beta} n_\text{total}\ket{\alpha}\ket{\beta} \\
&= \Tr(\rho_\text{cl}n_\text{total}) \\
&= \expval{n_\text{total}}_{\rho_\text{cl}}.
\end{align*} From the Quantum Cram{\'e}r-Rao bound (Theorem~\ref{thm::qCramerRao}), we get the standard quantum limit \begin{align*} 
\Delta t \geq \frac{1}{\sqrt{\expval{n_\text{total}}_{\rho_\text{cl}}}},
\end{align*} 
which we see is basically identical to the single mode case discussed in Section~\ref{sec::phaseEst}.

Let us consider the subspace spanned by the product of Fock states $\{\ket{i_a}\ket{i_b} \}$ where $i_a,i_b \leq n_\text{max}$. Within this subspace, the eigenvector of $G$ with the maximum eigenvalue is $G \ket{n_\text{max}}\ket{0} = n_\text{max}/2 \ket{n_\text{max}}\ket{0}$, and the eigenvector with the minimum eigenvalue is $G \ket{0}\ket{n_\text{max}} = -n_\text{max}/2 \ket{0}\ket{n_\text{max}}$. The state achieving the largest quantum Fisher information within this subspace is then $\ket{\psi} = (\ket{n_\text{max}}\ket{0} + \ket{0}\ket{n_\text{max}})/\sqrt{2}.$ This is the famous NOON state\cite{Dowling2008}. 

One may further verify that for the NOON state $\Delta^2_{\ket{\psi}}G = n_\text{max}^2/4 = \expval{n_\text{total}}_{\ket{\psi}}^2/4$ such that we have $$I_Q(\ket{\psi}, G) = \expval{n_\text{total}}_{\ket{\psi}}^2,$$ which gives us the measurement sensitivity 
\begin{align*} 
\Delta t \geq \frac{1}{\expval{n_\text{total}}_{\ket{\psi}}}.
\end{align*} We see that the NOON state achieves Heisenberg limited sensitivity. Again, we are careful to note that the Heisenberg limit is not a truly fundamental quantum limit for optical quantum systems (see discussion in Section~\ref{sec::phaseEst}).

Thus far, we have discussed what happens after the first beam splitter in the interferometer, which isolates the effect of the unitary encoding of $
\theta$ onto the quantum probe. Another more traditional convention is to consider the state $\rho_\text{in}$ that is being fed into the input ports of the interferometer and the state $\rho_\text{out}$ emerging from the output ports. These correspond to the state before the first beam splitter and after the second beam splitter in Fig.~\ref{fig::MZI}. In principle, one may perform any quantum measurement on the output state $\rho_\text{out}$, but traditional interferometry typically measures the intensity difference $n_a-n_b$ at the output, corresponding to the visibility of the interference fringes.

Under such settings, there exists a useful formalism\cite{ Yurke1986a} for analysing the relationship between the input state and the visibility at the output port of the interferometer. Let us define the operators 
\begin{align*}
J_x \coloneqq \frac{1}{2}(a^\dag b + b^\dag a), \quad J_y \coloneqq \frac{i}{2}(b^\dag a - a^\dag b), \quad J_z \coloneqq \frac{1}{2}(a^\dag a + b^\dag b).
\end{align*} One may verify that these operators satisfy the commutation relations $\comm{J_i}{J_j} = i \varepsilon_{ijk}J_k$ so they analogous to angular momentum operators\cite{Schwinger1965}. Based on these definitions, we can obtain the following convenient input-output relations: $$\expval{J_z}_{\rho_\text{out}} = \cos\theta \expval{J_z}_{\rho_\text{in}} - \sin\theta \expval{J_x}_{\rho_\text{in}} $$ and 
\begin{align*}
\Delta^2_{\rho_\text{out}}J_z &= \cos^2\theta\Delta^2_{\rho_\text{in}}J_z + \cos^2\theta\Delta^2_{\rho_\text{in}}J_x \\
& \hspace{0.15\textwidth} -2\sin\theta\cos \theta \text{cov}(J_x,J_z)_{\rho_\text{in}},
\end{align*} where $\text{cov}(J_x,J_z)_{\rho_\text{in}} \coloneqq \expval{J_x J_z + J_z J_x}_{\rho_\text{in}} - \expval{J_x}_{\rho_\text{in}}\expval{J_z}_{\rho_\text{in}}$ is the covariance between $J_x$ and $J_z.$ Notice that $J_z = (n_a - n_b)/2$, so $\expval{J_z}_{\rho_\text{out}}$ and  $\Delta^2_{\rho_\text{out}}J_z$ are, up constant factors, just the mean and variance of a visibility measurement at the output. Assuming that $\Delta_{\rho_\text{out}}J_z $ is sufficiently small, the following error propagation formula gives a good approximation of the measurement sensitivity one may expect from visibility measurements: $$\Delta\theta \vert_{\rho_\text{in}} = \frac{\Delta_{\rho_\text{out}}J_z}{\abs{\pdv{\expval{J_z}_{\rho_\text{out}}}{\theta}}}.$$

Now, suppose the input state is a single laser beam with the other input port empty, i.e. $\ket{\psi_\text{in}} = \ket{\alpha}\ket{0}$. Using the above input-output relations, we can show that $$\Delta\theta \vert_{\ket{\psi_\text{in}}} = \frac{1}{\sqrt{\expval{n_\text{total}}_{\ket{\psi_\text{in}}}}\abs{\sin\theta}}, $$ which is in line with what we expect from the standard quantum limit. It was the contribution of Caves\cite{Caves1981} who realized that one may beat this by replacing the vacuum with a nonclassical state such as squeezed vacuum. More specifically, he showed that if we choose $\ket{\psi_\text{in}} = \ket{\alpha}\ket{0, \abs{\epsilon}}$, then in the limit of $\expval{n_\text{total}}_{\ket{\psi_\text{in}}} \rightarrow \infty$ it is possible to achieve $$\Delta\theta \vert_{\ket{\psi_\text{in}}} \approx \frac{1}{\expval{n_\text{total}}_{\ket{\psi_\text{in}}}^{3/4}}$$ for certain combinations of $\alpha$ and $\abs{\epsilon}$. Note that this is below the Heisenberg limit in Eq.~\ref{eq::heisenbergLimit}, but this is only a limitation of a visibility measurement and not a fundamental limit. In principle, the input state $\ket{\psi_\text{in}} = \ket{\alpha}\ket{0, \abs{\epsilon}}$ is able to do better, potentially reaching the Heisenberg limit, if one considers more general measurements such as photon number based\cite{Pezze2008,Seshadreesan2011} or homodyne measurements\cite{dAriano1995,Oh2017, Oh2019}.

Curiously, it turns out that the coherent state $\ket{\alpha}$, despite being classical, plays an important role in determining whether sub shot noise sensitivities can be achieve. More specifically, if one of the input ports is a vacuum, i.e. the input state has the form $\rho_\text{in} = \rho\otimes \ketbra{0}$, then no matter what state you choose for $\rho$, it is not possible to beat the standard quantum limit.

One may demonstrate this using the Glauber-Sudarshar $P$-representation $\rho= \int\dd[2]{\alpha}P(\alpha)\ketbra{\alpha}$. If the input state has the form $\rho_\text{in} = \rho\otimes \ketbra{0}$, then after the first beam splitter, the state is $\rho_\text{mid} = U_\text{BS}\rho_\text{in} U_\text{BS}^\dag = \int\dd[2]{\alpha}P(\alpha)\ketbra{\frac{\alpha}{\sqrt{2}}} \otimes \ketbra{\frac{\alpha}{\sqrt{2}}}$. Let us consider the variance $\Delta^2_{\rho_\text{mid}} G$, which can be expanded in the following way: 
\begin{align*}
&\Delta^2_{\rho_\text{mid}} G \\
&= \expval{G^2}_ {\rho_\text{mid}} - \expval{G}^2_{\rho_\text{mid}} \\
&=  \frac{1}{4}\left [\expval{(n_a-n_b)^2}_ {\rho_\text{mid}} - \expval{n_a-n_b}^2_{\rho_\text{mid}} \right ] \\
&=\frac{1}{4}\left(\Delta^2_ {\rho_\text{mid}}n_a+ \Delta^2_ {\rho_\text{mid}}n_b -2\expval{n_a n_b}_ {\rho_\text{mid}}+2\expval{n_a}_ {\rho_\text{mid}}\expval{n_b}_ {\rho_\text{mid}}\right) \\
&=\frac{1}{2}\left(\Delta^2_ {\rho_\text{mid}}n_a-\expval{n_a n_b}_ {\rho_\text{mid}}+\expval{n_a}_ {\rho_\text{mid}}^2\right).
\end{align*} In the last line, we used the fact that $\Delta^2_ {\rho_\text{mid}}n_a = \Delta^2_ {\rho_\text{mid}}n_b$ and $\Delta^2_ {\rho_\text{mid}}n_a = \Delta^2_ {\rho_\text{mid}}n_b$ since $\rho_\text{mid}$ is symmetric on both modes. We now compute the individual terms in the sum. We can verify that 
\begin{align*}
&\expval{n_a n_b}_ {\rho_\text{mid}} \\
&= \Tr(\rho_\text{mid}n_an_b) \\
&= \Tr[\int\dd[2]{\alpha}P(\alpha)\ketbra{\frac{\alpha}{\sqrt{2}}} \otimes \ketbra{\frac{\alpha}{\sqrt{2}}}n_a n_b] \\
&=\int\dd[2]{\alpha}P(\alpha) \abs{\frac{\alpha}{\sqrt{2}}}^4.
\end{align*} 
We then substitute this into $\Delta^2_ {\rho_\text{mid}}n_a$ and  use the identity $n_a^2 = (a^\dag)^2a^2 + n_a$ to get 
\begin{align*}
&\Delta^2_ {\rho_\text{mid}}n_a \\
&= \expval{n_a^2}_{\rho_\text{mid}} - \expval{n_a}_ {\rho_\text{mid}}^2 \\
&= \Tr(\rho_\text{mid}n_a^2)  - \expval{n_a}_ {\rho_\text{mid}}^2\\
&= \Tr[\rho_\text{mid} (a^\dag)^2a^2 + n_a]  - \expval{n_a}_ {\rho_\text{mid}}^2\\
&= \Tr[\int\dd[2]{\alpha}P(\alpha)\ketbra{\frac{\alpha}{\sqrt{2}}} \otimes \ketbra{\frac{\alpha}{\sqrt{2}}}(a^\dag)^2a^2 ] \\
&\hspace{0.25\textwidth} +\expval{n_a}_ {\rho_\text{mid}} -\expval{n_a}_ {\rho_\text{mid}}^2\\
&=\int\dd[2]{\alpha}P(\alpha) \abs{\frac{\alpha}{\sqrt{2}}}^4 +\expval{n_a}_ {\rho_\text{mid}} -\expval{n_a}_ {\rho_\text{mid}}^2\\
&=\expval{n_a n_b}_ {\rho_\text{mid}}+\expval{n_a}_ {\rho_\text{mid}} -\expval{n_a}_ {\rho_\text{mid}}^2.
\end{align*} Finally, we substitute this expression back into $\Delta^2_{\rho_\text{mid}} G$ to get 
\begin{align*}
&\Delta^2_{\rho_\text{mid}} G \\
&=\frac{1}{2}\left(\Delta^2_ {\rho_\text{mid}}n_a-\expval{n_a n_b}_ {\rho_\text{mid}}+\expval{n_a}_ {\rho_\text{mid}}^2\right) \\
&=\frac{1}{2}\expval{n_a}_ {\rho_\text{mid}} \\ 
&= \frac{1}{4}\expval{n_a+n_b}_ {\rho_\text{mid}} \\ 
&= \frac{1}{4}\expval{n_\text{total}}_ {\rho_\text{mid}}.
\end{align*}

Using the fact that the quantum Fisher information is bounded by 4 times the variance, we get $I_Q(\rho_\text{mid}, G) \geq 4 \Delta^2_{\rho_\text{mid}} G = \expval{n_\text{total}}_ {\rho_\text{mid}}$. For any unbiased estimator $t$, the quantum Cram{\'e}r-Rao bound says $$\Delta t \geq \frac{1}{\sqrt{\expval{n_\text{total}}_ {\rho_\text{mid}}}}.$$ Caves\cite{Caves1980,Caves1981} argued that zero point fluctuations, i.e. the fluctuations of the vacuum energy entering the the interferometer, leads to shot noise limited sensitivity. Random energy fluctuations of the vacuum can cause random photons to enter the interferometer and diminish the visibility. From the previous arguments, we see that so long as one of the input ports is empty, the vacuum noise is sufficient to dominate any attempts to improve the situation by injecting nonclassical light into the other port\cite{Lang2013, Takeoka2017}. This conclusion is not limited to just visibility measurements, but applies to all possible measurements performed on the output state.

The proposal by Caves\cite{Caves1981} to replace the vacuum with a squeezed state was the first of such such proposals to use nonclassical states of light to improve interferometry. Broadly speaking, the modern interpretation of quantum metrology can be said to have started from this work. Cave's proposal is now being adopted in gravitational wave detectors\cite{Schnabel2010,LIGO2011, LIGO2013}. Other than the aforementioned NOON and squeezed states, many other nonclassical quantum states have also been considered as potential inputs. These include highly nonclassical states such as such as two-mode squeezed states\cite{Bondurant1984}, entangled coherent states\cite{Joo2011, Joo2012} and definite photon number states\cite{Yurke1986a, Holland1993,Sanders1995, Berry2000}.

\subsubsection{Estimating phase space displacements} \label{sec::estDisplacements}

Previously, we have considered both single mode and two mode phase estimation problems. Such problems are equivalent to measuring changes in angular rotation in phase space. The natural counterpart to angular rotations are the set of linear displacements. It turns out that nonclassical quantum states also demonstrate intrinsic superiority over classical states when the task is to measure the extent of the linear displacement\cite{Munro2002}. This was recently proposed as a nonclassicality test in Ref.~\onlinecite{Yadin2018} and  Ref.~\onlinecite{Kwon2019}. Both were able to show that for pure states, more Fisher information can be extracted from nonclassical states compared to classical states. However in Ref.~\onlinecite{Yadin2018}, the extension to mixed quantum states was achieved using a convex roof approach which does not have a direct operational interpretation. In contrast, Ref.~\onlinecite{Kwon2019} focused more on the amount of Fisher information that is extractable from mixed quantum states.

Consider a system consisting of $N$ optical modes. The corresponding creation and annihilation operators $a_i^\dag$ and $a_i$ where $i = 1, \ldots , N$. An $N$ mode annihilation operator can be defined as $a_{\boldsymbol{\mu}} \coloneqq \sum_{i=1}^N \mu_i a_i$ where $\boldsymbol{\mu} = [\Re(\mu_1), \Im(\mu_1) \ldots, \Re(\mu_N), \Im(\mu_N)]$ is a $2N$ dimensional real vector of unit length, i.e. $\abs{\boldsymbol{\mu}}^2 = \sum_{i=1}^N \abs{\mu_i}^2 = 1$. We can also define the $N$ mode field quadrature $$X_{\boldsymbol{\mu}} \coloneqq\frac{a_{\boldsymbol{\mu}} + a^\dag_{\boldsymbol{\mu}}}{\sqrt{2}}$$ as well as the $N$ mode displacement operator $$D(\theta, \boldsymbol{\mu}) \coloneqq e^{-i\theta X_{\boldsymbol{\mu}}}. $$ For a single mode, this reduces to $D(\theta, \mu=e^{i\phi} ) \coloneqq \exp[-i\theta (e^{i\phi} a + e^{-i\phi} a^\dag)/\sqrt{2}] = \exp[\alpha a^\dag - \alpha^* a],$ where $\alpha = i\theta e^{-i \phi }/\sqrt{2}.$ We therefore see that other than a re-parametrization, the single mode displacement operator defined in Section~\ref{sec::cohStates} is retrieved when $N=1$. The parameter $\theta$ determines the magnitude of the displacement, while $\boldsymbol{\mu}$ determines the direction.

Suppose we are interested to estimate the magnitude of the displacement $\theta$. This is equivalent to choosing the generator $G = X_{\boldsymbol{\mu}}$ (see Section~\ref{sec::unitaryQMet}). The fundamental limits of this parameter estimation problem is given by the quantum Fisher information $I_Q(\rho, X_{\boldsymbol{\mu}})$. Writing $\rho = \sum_i p_i \ketbra{i}$ in terms of its eigenbasis $\{ \ket{i} \}$, we can compute the quantum Fisher information and verify that it simplifies to the following:
\begin{align*}
I_Q(\rho, X_{\boldsymbol\mu}) 
&= 2 \sum_{i,j} \frac{(p_i - p_j)^2}{p_i + p_j}|\bra{i} X_{\boldsymbol\mu} \ket{j}|^2 \\
&= \boldsymbol{\mu}^T \boldsymbol{F} \boldsymbol{\mu},
\end{align*}
where ${\boldsymbol F}$ is called the quantum Fisher information matrix. It is a real symmetric $2N \times 2N$ matrix  with elements 
\begin{align} \label{eq::QFIMatrix}
F_{kl} = 2 \sum_{i,j} \frac{(p_i - p_j)^2}{p_i + p_j} \bra{i}X^{(k)}\ket{j} \bra{j}X^{(l)}\ket{i},
\end{align}
and $X^{(2i-1)} = (a_n + a^\dagger_n) / \sqrt{2}$ and $X^{(2i)} = (a_n - a^\dagger_n) / (\sqrt{2}i)$ are the local canonical quadrature operators for the $i$th mode.

The resulting quantum Cram{\'e}r Rao bound may therefore be written as $$\Delta t \geq \frac{1}{\sqrt{\boldsymbol \mu^T \boldsymbol F \boldsymbol \mu}}.$$ We see that the relevant quantities can be computed from the quantum Fisher information matrix $\boldsymbol{F}.$

For any state $\rho$, let us consider the average Fisher information over all possible quadrature directions ${\boldsymbol\mu}$: 
\begin{align} \label{eq::aveMetPow}
M_\text{ave}(\rho) := \frac{1}{2A} \int_{S} \dd[2N]{\boldsymbol\mu} I_F(\ket{\psi}, X_{\boldsymbol\mu} )
\end{align}
 where $S = \{ \boldsymbol{\mu} : \abs{\boldsymbol{\mu}} = 1 \}$ is the surface of the unit sphere, and $A = \int_S d^{2N}\boldsymbol\mu = 2\pi^N/(N-1)! $.

It is possible to simplify the above expression for $M_\text{ave}(\rho)$. Since $I_Q(\rho, X_{\boldsymbol\mu}) = \boldsymbol{\mu}^T \boldsymbol{F} \boldsymbol{\mu}$, we can write $I_Q(\rho, X_{\boldsymbol\mu}) = \boldsymbol{e}_i^\dag O_{\boldsymbol\mu,i}^\dag F O_{\boldsymbol\mu,i} \boldsymbol{e}_i$ for any complete set of basis vectors $\{ \boldsymbol{e}_i \}_{i=1}^{2N}$.  $O_{\boldsymbol\mu,i}$ is some orthorgonal matrix in $2N$ dimensional vector space satisfying $O_{\boldsymbol\mu,i} \boldsymbol{e}_i = \boldsymbol{\mu}$. However, because the integration is over every direction $\boldsymbol{\mu}$, it is equivalent to integrating over every possible orthogonal matrix $O_\mu$, so we can drop the index $i$ and write $\int_S d^{2N} \boldsymbol{\mu}^T \boldsymbol{F} \boldsymbol{\mu} = \int_S \dd[2]{\boldsymbol{\mu}} \boldsymbol{e}_i^\dag O_{\boldsymbol\mu,i}^\dag F O_{\boldsymbol\mu,i} \boldsymbol{e}_i = \int_S \dd[2]{\boldsymbol{\mu}} \boldsymbol{e}_i^\dag O_{\boldsymbol\mu}^\dag F O_{\boldsymbol\mu} \boldsymbol{e}_i$ for every $i$, so we get 
\begin{align*}
2N \int_S d^{2N} \boldsymbol{\mu}^T \boldsymbol{F} \boldsymbol{\mu} &= \sum_{i=1}^{2N} \int_S \dd[2N]{\boldsymbol{\mu}} \boldsymbol{e}_i^\dag O_{\boldsymbol\mu}^\dag F O_{\boldsymbol\mu} \boldsymbol{e}_i \\ 
&= \int_S \dd[2N]{\boldsymbol{\mu}} \Tr(O_{\boldsymbol\mu}^\dag F O_{\boldsymbol\mu} ) \\
 &= A \Tr \boldsymbol{F}.
\end{align*} Substituting back into Eq.~\ref{eq::aveMetPow}, we get the expression 
\begin{align} \label{eq::aveMetPowTrace}
M_\text{ave}(\rho) = \frac{\Tr \boldsymbol{F}}{4N}.
\end{align}

It is instructive to consider the $N = 1$ case. One may verify using Eq.~\ref{eq::QFIMatrix} that for single mode states, $M_\text{ave}(\rho) = \frac{\Tr \boldsymbol{F}}{4N} = [I_Q(\rho, x)+ I_Q(\rho, p)]/4$ for general mixed states and that  $M_\text{ave}(\ket{\psi}) = \Delta_{\ket{\psi}}^2 x +\Delta_{\ket{\psi}}^2 p$ for pure states. For a coherent state $\ket{\alpha}$, we see that $M_\text{ave}(\ket{\alpha}) = 1$ (see also Section~\ref{sec::cohStates}). As such, for any classical state $\rho_\text{cl} = \int \dd[2]{\alpha}P_\text{cl}(\alpha) \ketbra{\alpha}$ where $P_\text{cl}(\alpha)$ is a positive probability distribution, due to the convexity of the quantum Fisher information, we have the classical bound $$ M_\text{ave}(\rho_\text{cl}) \leq 1 $$ Any state $\rho$ surpassing this limit must clearly be nonclassical. Over the set of pure states, only the coherent states can satisfy $\Delta_{\ket{\psi}}^2 x +\Delta_{\ket{\psi}}^2 p =1$ , so $M_\text{ave}(\rho)$ is able to identify every nonclassical pure states. Identical arguments also apply for $N > 1$, so the classical bound also applies for multimode systems. 

Let us now consider a different quantity. Suppose instead of the average, we compute the maximum Fisher information over all possible quadrature directions $\boldsymbol\mu$.
\begin{equation}
\label{QF1}
M_\text{opt} (\rho) := \frac{1}{2} \max_{\boldsymbol\mu \in S} I_Q(\rho, X_{\boldsymbol\mu}) = \frac{ \lambda_{\max} ({\boldsymbol F}) }{2},
\end{equation}
where $\lambda_{\max} ({\boldsymbol F})$ is the maximum eigenvalue of $\boldsymbol F$. The last equality comes from the direct observation that $I_Q(\rho, X_{\boldsymbol\mu}) = \boldsymbol{\mu}^T \boldsymbol{F} \boldsymbol{\mu}$ is maximized when $\boldsymbol{\mu}$ is the eigenvector corresponding to $\lambda_{\max} ({\boldsymbol F})$. $M_\text{opt} (\rho)$ is perhaps a more operational quantity, because it directly quantifies the maximum metrological power that you can extract from the state in some parameter estimation problem, rather than some hypothetical average performance.

Again, it is instructive to consider the single mode case $N=1$. For the coherent state $\ket{\alpha}$, the general field quadrature $x_\phi \coloneqq (e^{-i\phi}a + e^{i\phi}a^\dag)\sqrt{2}$ has the same variance in every direction so $\Delta^2_{\ket{\alpha}} x_\phi = 1/2$. From the convexity of the quantum Fisher information, we obtain the classical bound $$M_\text{opt} (\rho_\text{cl}) \leq 1,$$ so any quantum state $\rho$ that exceeds this bound has to be nonclassical. For pure states $\ket{\psi}$, since every nonclassical state satisfies $\Delta^2_{\ket{\psi}}x + \Delta^2_{\ket{\psi}}p > 1$, there must be at least one quadrature direction where $\Delta^2_{\ket{\alpha}} x_\phi > 1/2$, so $M_\text{opt} (\rho)$ is also able to able to identify every nonclassical pure state. 

As every nonclassical pure state will beat both classical limits discussed in this section, examples are plentiful. An example of states of states exceeding the classical limits are the Cat states $\ket{\mathrm{\psi_{\pm}}} \coloneqq \frac{1}{\sqrt{\mathcal{N}}} (\ket{\beta} \pm \ket{-\beta})$, which achieves $M_{\text{ave}} (\ket{\mathrm{\psi_{\pm}}}) = 2\expval{n}_{\ket{\mathrm{\psi_{\pm}}}}$ and $M_{\text{opt}} (\ket{\mathrm{\psi_{\pm}}}) = 2(\expval{n}_{\ket{\mathrm{\psi_{\pm}}}} +\abs{\beta}^2)$. Interestingly, entanglement does not necessarily help the sensitivity in this case since both Fock states $\ket{n}$ and NOON states $(\ket{n}\ket{0}+\ket{0}\ket{n})/\sqrt{2}$ achieves $M_{\text{ave}} (\ket{\psi}) = 2n/N $ and $M_{\text{opt}} (\ket{\psi}) = 2n$.

It is also worth mentioning that in Refs.~\onlinecite{Yadin2018,Kwon2019} , the authors were also motivated to construct a nonclassicality measure in the resource theory of nonclassicality (see Section~\ref{sec::resTheoryNonclass}). One may also show that $M_{\text{ave}} (\ket{\psi})$ for pure states $\ket{\psi}$, and $M_{\text{opt}} (\rho)$ for general mixed states $\rho$, are both nonclassicality measures under the resource theoretical approach.

\subsubsection{Quantum illumination and reflectivity measurements}

Quantum illumination is a target detection scheme first proposed by Lloyd\cite{Lloyd2008}. The idea is to be able to detect the presence of a weakly reflective target by exploiting the properties of entanglement. By sending out a probe beam that is entangled to the receiver, we can potentially discriminate between receiving a random photon from a noisy environment, or a photon that was genuinely reflected back from the unknown object. One may show that a two mode squeezed state is able to beat strategies using only classical light sources\cite{Tan2008,Shapiro2009}. Experimental realizations of the quantum illumination protocol have recently been performed \cite{Lopaeva2013, Lopaeva2014, Zhang2015}. There have also been proposals to perform quantum illumination in the microwave regime where a radar typically operate\cite{Barzanjeh2015}, as well as repurpose the protocol for quantum communication\cite{Shapiro2009a}. A somewhat surprising fact is that even if the entanglement of the initial state is broken after the signal is sent out, the quantum advantage may still survive\cite{Zhang2013}.

At first glance, quantum illumination appears closer to a remote sensing problem rather than a parameter estimation problem. Indeed, in quantum illumination, the figure of merit is typically the error probability of discriminating between a prepared photon versus a photon from background noise. This does not appear at first directly related to the metrological problems we have considered thus far. However, the problem of quantum illumination has recently been rephrased as a parameter estimation problem in order to provide upper bounds to the general quantum illumination problem\cite{Sanz2017}. This section will mainly discuss this approach.

A quantum illumination strategy consists of the preparation of a two mode state, called the signal-idler system $\ket{\Psi}_\text{SI}$. One half of the system (signal) is sent out as a probe, while the other half (idler) is kept in the laboratory. In addition to the signal-idler system, we also need to consider a source of noisy photons, modelled as a thermal bath $\rho_B$, which is in the thermal state $\rho_B = (1-e^{-\beta}) \sum_{n\geq 0} e^{-\beta n}\ket{n}_B\bra{n},$ where $\beta$ is the inverse temperature. The inverse temperature is related to the mean photon number via the relation $\expval{n_b}_{\rho_B} = (1-e^{-\beta})^{-1}$. 

The weakly reflective object can be modelled as a beam splitter with low reflectivity. Let $s,s^\dagger$ and $b,b^\dag$ be the annihilation and creation operators for the signal and the bath modes respectively. The beam splitter interaction is $U_\theta = \exp[\theta(s^\dag b - s b^\dag)]$. Defining our generator as $G = i(s b^\dag - s^\dag b)$, we can consider the problem of estimating the physical parameter $\theta$, which corresponds to the relectivity of beam splitter. This setup is shown in Fig.~\ref{fig::QIModel}. For an object with low reflectivity, $\theta \approx 0$ so the objective is to send a quantum probe that is able to measure very small changes in $\theta$. One figure of merit here is therefore $I_Q( \rho_\text{SI}\otimes \rho_B, G).$ Note that the input state has three modes, but the interaction is only between the signal and the bath modes. Also note that because $\rho_B$ is a thermal state, the combined state is never a pure state.

\begin{figure}
\includegraphics[width = 0.9\linewidth ]{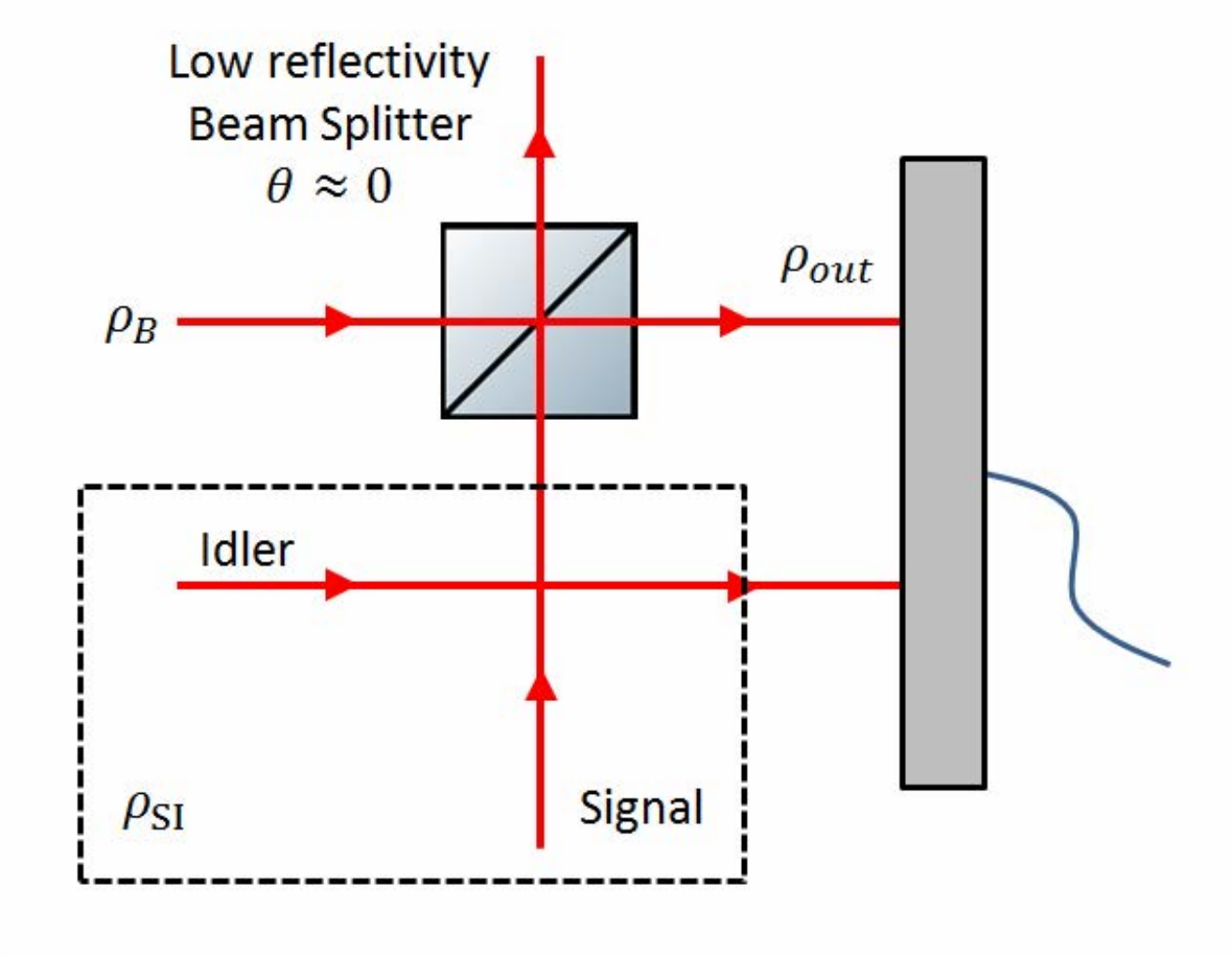}
\caption{\label{fig::QIModel} A simplified model of quantum illumination. A signal idler state $\rho_\text{SI}$ is prepared. The signal state is sent out and interacts with a thermal state $\rho_B$ via a beam splitter with low reflectivity $\theta \approx 0$. A final measurement is performed on the output state $\rho_\text{out}$. }
\end{figure}

Nevertheless, we can still consider a product of coherent states $\ket{\alpha}_S\ket{\beta}_I\ket{\gamma}_B$ for the input. Evaluating the quantum Fisher information for the generator $G$, we can verify that $I_Q(\ket{\alpha}_S\ket{\beta}_I\ket{\gamma}_B, G) = 4( \expval{n_\text{S}}_{\ket{\alpha}} + \expval{n_B}_{\ket{\gamma}})$. Using Eq.~\ref{eq::convexRoofFisher} again, we have for any classical state $\rho_\text{cl} = \int \dd[2]{\alpha}\dd[2]{\beta}\dd[2]{\gamma} P_\text{cl}(\alpha,\beta,\gamma) \ketbra{\alpha}\otimes \ketbra{\beta} \otimes \ketbra{\gamma}$, the classical bound $$I_Q(\rho_\text{cl}, G) = 4( \expval{n_\text{S}}_{\rho_\text{cl}} + \expval{n_B}_{\rho_\text{cl}}).$$ The maximum measurement sensitivity from the quantum Cram{\'e}r-Rao bound is then $$\Delta t \geq \frac{1}{2\sqrt{\expval{n_\text{S}}_{\rho_\text{cl}} + \expval{n_B}_{\rho_\text{cl}}}}.$$ This is, up to a constant factor, broadly similar to the standard quantum limit from Eq.~\ref{eq::shotNoiseLimit}. 

We note that the above is actually the classical limit for a general parameter estimation of the reflectivity $\theta$, where we allow for any input state. For the quantum illumination problem, we have to impose the condition that the input state has the form $\rho_\text{SI}\otimes \rho_B$, where $\rho_B$ is a thermal state. It is also generally assumed that the transmitted part of the signal is lost, and that one only receives the reflected signal. This extra assumption complicates the problem and can lead to very different bounds.

To simplify the problem, we assume the prepared state is a pure state. From the Schmidt decomposition~\cite{NielsenChuang}, we can always write it in the form $$\ket{\psi}_\text{SI} =  \sum_i \sqrt{\lambda_i} \ket{w_i}_S\ket{v_i}_I.$$ For the input state $ \ket{\psi}_{\text{SI}}\bra{\psi}\bra\otimes \rho_B$, we compute the Fisher information that can be extracted by performing a measurement on the reflected signal and and idler system (Fig.~\ref{fig::QIModel}). This corresponds to computing the Fisher information of $\rho_\theta \coloneqq \Tr_B[U_\theta \ket{\psi}_{\text{SI}}\bra{\psi}\bra\otimes \rho_B U^\dag_\theta]$. Using Definition~\ref{def::genQFI}, we can get the expression:
\begin{align*}
I_Q(\rho, \theta) = \frac{4}{1+\expval{n_B}_\rho} \sum_{i,j} \frac{\lambda_i \lambda_j}{\lambda_j + \lambda_i \frac{\expval{n_B}_\rho}{\expval{n_B}_\rho+1}} \abs{\mel{w_i}{s}{w_j}_S}^2.
\end{align*}
In particular, we see that for a product state $\ket{\psi}_\text{SI} =  \ket{w}_S\ket{v}_I,$ this simplifies to $$I_Q(\rho, \theta) = \frac{4 \abs{\mel{w}{s}{w}_I}^2}{1+2\expval{n_B}_\rho}.$$ Furthermore, if the signal is a coherent state $\ket{w} = \ket{\alpha}$, this gives $$I_Q(\rho, \theta) = \frac{4 \expval{n_S}_\rho}{1+2\expval{n_B}_\rho}.$$ This further suggests that if the signal-idler state is a classical state $\rho_\text{cl} = \int \dd[2]{\alpha}\dd[2]{\beta} P_\text{cl}(\alpha,\beta) \ketbra{\alpha} \otimes \ketbra{\beta}$, where $P_\text{cl}(\alpha,\beta)$ is a positive probability distribution function, then, from the convexity of the quantum Fisher information, we have the following classical bound for the quantum illumination problem $$I_Q(\rho_\text{cl}, \theta) \leq \frac{4 \expval{n_S}_{\rho_\text{cl}}}{1+2\expval{n_B}_{\rho_B}}.$$ The quantum Cram{\'e}r-Rao bound then gives $$\Delta t \geq \sqrt{\frac{1+2\expval{n_B}_{\rho_B}}{4 \expval{n_S}_{\rho_\text{cl}}}}.$$ If the environment is sufficiently cool, corresponding to $\expval{n_B}_{\rho_B} \ll 1$, then it turns out that the classical bound is the optimal quantum mechanical bound, so no advantage can be extracted from a nonclassical probe. If the mean photon numbers of the probe $\rho_\text{SI}$ is moderate and the environment $\rho_B$ is sufficiently warm, then it can be demonstrated that entangled coherent states or two mode squeezed vacuum can beat the classical limit.

We are careful to note that the scenario in Fig.~\ref{fig::QIModel} is a highly idealized one. In an actual implementation of quantum illumination, the signal will acquire an additional unknown phase relative to the idler, which is not modelled here. Nonetheless, it is sufficient to establish some fundamental limits on the performance to be gained from using nonclassical states. Also related is Refs.~\onlinecite{Jakeman1986,SabinesChesterking2019}, where nonclassical states are used to improve measurements of the transmissivity of an object in scenarios where photons can be lost.

\subsection{Sources of nonclassical light}

It should be clear at this stage that nonclassical states are a valuable resource of metrological power. All the previously discussed schemes exploiting nonclassical effects assumes that some source of nonclassical states are readily available. In practice however, not all states are created equal, and some nonclassical states are more readily produced than others. In this section we touch upon some methods of generating travelling nonclassical states of light.

Squeezed states may be may be produced by passing light through a nonlinear optical medium via a process called spontaneous parametric down-conversion (SPDC)\cite{Burnham1970, Shen1984, Hong1987, Shih1988, Shih1994, Kwiat1995}. If the process produces a single beam, we say that it is degenerate. If it produces two correlated beams, we say that it is non-degenerate. Degenerate SPDCs produce single mode squeezed vacuum states while non-degenerate SPDCs produce two mode squeezed vacuum. Two mode squeezed vacuum states may also be produced by passing two single mode squeezed vacuum states through a beam splitter.

Two mode squeezed states are also a source of heralded single photons. At low intensities, they emit correlated photon pairs. By detecting one photon in a beam, we know that the other beam must contain a single photon. Heralded low number Fock states may also be similarly generated in this way\cite{Cooper2013}. Other single photon sources include quantum dots\cite{Ohnesorge1997,Gerard1998} and single atoms in resonant microcavities\cite{Kuhn2002,McKeever1992}. Superconducting quantum circuits\cite{Hofheinz2008} are a source of low number Fock states.

A two-photon NOON state can be generated from two single photon sources using the Hong-Ou-Mandel effect\cite{Hong1987}. SPDCs can also be exploited to generate NOON states with low photon numbers\cite{Mitchell2004,Walther2004,Afek2010}. They can also be generated via superconducting circuits \cite{Wang2011}.

Cat states and entangled coherent states may be produced via a nonlinear medium\cite{Yurke1986, Sanders1992, Mecozzi1987, Gerry1999, Jeong2004}, but the required levels of nonlinearities are extremely demanding. Approximate cat states states may also be generated by probabilistically subtracting photons from a squeezed state \cite{Wenger2004, Ourjoumtsev2006, NeergaardNielsen2006, Wakui2007}, by squeezing a single photon\cite{Lund2004}, and by postselecting with a homodyne measurement performed on a number state\cite{Ourjoumtsev2007}. Entangled coherent states may then be generated by passing a cat state through a beam splitter, or be probabilistically prepared via two cat states and postselecting on a photon subtraction event\cite{Ourjoumtsev2009}.  

There are also schemes to generate any arbitrary superpositions of Fock states~\cite{Vogel1993, Hofheinz2009} for low photon numbers. This in principle allows for an infinite variety of nonclassical states with low photon numbers to be approximated. Exotic states involving superpositions/entanglement of classical and nonclassical states of light have also be produced in the laboratory\cite{Jeong2014,Morin2014}.

There is a huge variety of possible nonclassical states\cite{Dodonov2002} and a plethora of possible techniques to generate them. We point the interested reader to other dedicated reviews of the subject for more in-depth discussions\cite{Davidovich1996, Braunstein2005, Sanders2012,DellAnno2006, Boyd2019}.

\section{Conclusion}

In this review, we discussed the notion of nonclassicality in light. By arguing that the most classical states of light are the set of coherent states, we discussed many approaches of identifying and quantifying the nonclassicality of a system. We then discussed how nonclassical states may be exploited to beat classical bounds. This elevates the notion of nonclassicality from something that is purely of fundamental interest, to a resource with practical utility.

The primary goal of the authors is to provide a convincing and thorough demonstration of the utility of nonclassical states in one very specific application: parameter estimation. We have mostly done this by discussing several idealized scenarios, typically under noiseless conditions and under perfect conditions. Naturally, we can ask whether the quantum advantages persist even under more realistic assumptions of noise and imperfect detection. The answer is yes, but rather unsurprisingly, the quantum advantage is  severely diminished. In particular, for phase estimation problems, when the noise models are identically and independently distributed, the quadratic scaling of the Heisenberg limit is not reachable and the quantum advantage is limited to a constant factor. See Ref.~\onlinecite{Escher2011,Dobrzanski2012, Dobrzanski2015} for more discussion of the achievable quantum limits under noisy scenarios.

We have also restricted most of our discussion to the estimation of a single parameter $\theta$. More generally we can also consider situations where multiple parameters are estimated at the same time. See Ref.~\onlinecite{Sidhu2019} for a recent discussion on such generalized quantum parameter estimation problems. However, while a Cram{\'e}r-Rao bound for multi-parameter estimation problems can also be proven, this bound cannot be saturated in general so its interpretation is not as strong as in the single parameter case. 

On the topic of nonclassicality, it is also worth mentioning that there is considerable interest in generating macroscopic superpositions of quantum states. The primary motivation behind this is to push the boundaries of quantum mechanics to the macroscopic regime and to continue to test its validity in the macroscopic limit. Macroscopic superposition is a separate topic in its own right and encompasses more than just systems of light. However, when limited to optical systems, such macroscopic superpositions are necessarily a subset of nonclassicality. See Refs.~\onlinecite{Jeong2015, Frowis2018} for recent reviews of the subject.

Finally, the authors hope that the collection of topics discussed in this review proved helpful to the reader in understanding some of the key concepts concerning nonclassicality as well as the extraction of useful metrological power from quantum states.

\begin{acknowledgments}
This work was supported by the National Research Foundation of Korea (NRF) through a grant funded by the the Ministry of Science and ICT (Grant No. NRF-2019R1H1A3079890). K.C. Tan was supported by Korea Research Fellowship Program through the National Research Foundation of Korea (NRF) funded by the Ministry of Science and ICT (Grant No. 2016H1D3A1938100).
\end{acknowledgments}

\bibliography{aipsamp}

\providecommand{\noopsort}[1]{}\providecommand{\singleletter}[1]{#1}%
\begin{thebibliography}{222}%
\makeatletter
\providecommand \@ifxundefined [1]{%
 \@ifx{#1\undefined}
}%
\providecommand \@ifnum [1]{%
 \ifnum #1\expandafter \@firstoftwo
 \else \expandafter \@secondoftwo
 \fi
}%
\providecommand \@ifx [1]{%
 \ifx #1\expandafter \@firstoftwo
 \else \expandafter \@secondoftwo
 \fi
}%
\providecommand \natexlab [1]{#1}%
\providecommand \enquote  [1]{``#1''}%
\providecommand \bibnamefont  [1]{#1}%
\providecommand \bibfnamefont [1]{#1}%
\providecommand \citenamefont [1]{#1}%
\providecommand \href@noop [0]{\@secondoftwo}%
\providecommand \href [0]{\begingroup \@sanitize@url \@href}%
\providecommand \@href[1]{\@@startlink{#1}\@@href}%
\providecommand \@@href[1]{\endgroup#1\@@endlink}%
\providecommand \@sanitize@url [0]{\catcode `\\12\catcode `\$12\catcode
  `\&12\catcode `\#12\catcode `\^12\catcode `\_12\catcode `\%12\relax}%
\providecommand \@@startlink[1]{}%
\providecommand \@@endlink[0]{}%
\providecommand \url  [0]{\begingroup\@sanitize@url \@url }%
\providecommand \@url [1]{\endgroup\@href {#1}{\urlprefix }}%
\providecommand \urlprefix  [0]{URL }%
\providecommand \Eprint [0]{\href }%
\providecommand \doibase [0]{http://dx.doi.org/}%
\providecommand \selectlanguage [0]{\@gobble}%
\providecommand \bibinfo  [0]{\@secondoftwo}%
\providecommand \bibfield  [0]{\@secondoftwo}%
\providecommand \translation [1]{[#1]}%
\providecommand \BibitemOpen [0]{}%
\providecommand \bibitemStop [0]{}%
\providecommand \bibitemNoStop [0]{.\EOS\space}%
\providecommand \EOS [0]{\spacefactor3000\relax}%
\providecommand \BibitemShut  [1]{\csname bibitem#1\endcsname}%
\let\auto@bib@innerbib\@empty
\bibitem [{\citenamefont {Braunstein}\ and\ \citenamefont {van
  Loock}(2005)}]{Braunstein2005}%
  \BibitemOpen
  \bibfield  {author} {\bibinfo {author} {\bibfnamefont {S.~L.}\ \bibnamefont
  {Braunstein}}\ and\ \bibinfo {author} {\bibfnamefont {P.}~\bibnamefont {van
  Loock}},\ }\href@noop {} {\bibfield  {journal} {\bibinfo  {journal} {Rev.
  Mod. Phys.}\ }\textbf {\bibinfo {volume} {77}},\ \bibinfo {pages} {513}
  (\bibinfo {year} {2005})}\BibitemShut {NoStop}%
\bibitem [{\citenamefont {Taylor}\ and\ \citenamefont
  {Bowen}(2016)}]{Taylor2016}%
  \BibitemOpen
  \bibfield  {author} {\bibinfo {author} {\bibfnamefont {M.~A.}\ \bibnamefont
  {Taylor}}\ and\ \bibinfo {author} {\bibfnamefont {W.~P.}\ \bibnamefont
  {Bowen}},\ }\href@noop {} {\bibfield  {journal} {\bibinfo  {journal} {Phys.
  Rep.}\ }\textbf {\bibinfo {volume} {615}},\ \bibinfo {pages} {1} (\bibinfo
  {year} {2016})}\BibitemShut {NoStop}%
\bibitem [{\citenamefont {Berchera}\ and\ \citenamefont
  {Degiovanni}(2019)}]{Berchera2019}%
  \BibitemOpen
  \bibfield  {author} {\bibinfo {author} {\bibfnamefont {I.~R.}\ \bibnamefont
  {Berchera}}\ and\ \bibinfo {author} {\bibfnamefont {I.~P.}\ \bibnamefont
  {Degiovanni}},\ }\href@noop {} {\bibfield  {journal} {\bibinfo  {journal}
  {Metrologia}\ }\textbf {\bibinfo {volume} {56}},\ \bibinfo {pages} {024001}
  (\bibinfo {year} {2019})}\BibitemShut {NoStop}%
\bibitem [{\citenamefont {Grynberg}, \citenamefont {Aspect},\ and\
  \citenamefont {Fabre}(2010)}]{Grynberg2010}%
  \BibitemOpen
  \bibfield  {author} {\bibinfo {author} {\bibfnamefont {G.}~\bibnamefont
  {Grynberg}}, \bibinfo {author} {\bibfnamefont {A.}~\bibnamefont {Aspect}}, \
  and\ \bibinfo {author} {\bibfnamefont {C.}~\bibnamefont {Fabre}},\
  }\href@noop {} {\emph {\bibinfo {title} {Introduction to quantum optics}}}\
  (\bibinfo  {publisher} {Cambridge University Press},\ \bibinfo {address}
  {Cambridge},\ \bibinfo {year} {2010})\BibitemShut {NoStop}%
\bibitem [{\citenamefont {Yuen}\ and\ \citenamefont
  {Shapiro}(1980)}]{Yuen1980}%
  \BibitemOpen
  \bibfield  {author} {\bibinfo {author} {\bibfnamefont {H.~P.}\ \bibnamefont
  {Yuen}}\ and\ \bibinfo {author} {\bibfnamefont {J.~H.}\ \bibnamefont
  {Shapiro}},\ }\href@noop {} {\bibfield  {journal} {\bibinfo  {journal} {IEEE
  Trans. Inf. Theory}\ }\textbf {\bibinfo {volume} {IT-26}},\ \bibinfo {pages}
  {78} (\bibinfo {year} {1980})}\BibitemShut {NoStop}%
\bibitem [{\citenamefont {Yuen}\ and\ \citenamefont {Chan}(1983)}]{Yuen1983}%
  \BibitemOpen
  \bibfield  {author} {\bibinfo {author} {\bibfnamefont {H.~P.}\ \bibnamefont
  {Yuen}}\ and\ \bibinfo {author} {\bibfnamefont {V.~W.~S.}\ \bibnamefont
  {Chan}},\ }\href@noop {} {\bibfield  {journal} {\bibinfo  {journal} {Opt.
  Lett.}\ }\textbf {\bibinfo {volume} {8}},\ \bibinfo {pages} {177} (\bibinfo
  {year} {1983})}\BibitemShut {NoStop}%
\bibitem [{\citenamefont {Schumaker}(1984)}]{Schumaker1984}%
  \BibitemOpen
  \bibfield  {author} {\bibinfo {author} {\bibfnamefont {B.~L.}\ \bibnamefont
  {Schumaker}},\ }\href@noop {} {\bibfield  {journal} {\bibinfo  {journal}
  {Opt. Lett.}\ }\textbf {\bibinfo {volume} {9}},\ \bibinfo {pages} {189}
  (\bibinfo {year} {1984})}\BibitemShut {NoStop}%
\bibitem [{\citenamefont {Yurke}\ and\ \citenamefont
  {Stoler}(1987)}]{Yurke1987}%
  \BibitemOpen
  \bibfield  {author} {\bibinfo {author} {\bibfnamefont {B.}~\bibnamefont
  {Yurke}}\ and\ \bibinfo {author} {\bibfnamefont {D.}~\bibnamefont {Stoler}},\
  }\href@noop {} {\bibfield  {journal} {\bibinfo  {journal} {Phys. Rev. A}\
  }\textbf {\bibinfo {volume} {36}},\ \bibinfo {pages} {1955} (\bibinfo {year}
  {1987})}\BibitemShut {NoStop}%
\bibitem [{\citenamefont {Vogel}\ and\ \citenamefont
  {Welsch}(2006)}]{Vogel2006}%
  \BibitemOpen
  \bibfield  {author} {\bibinfo {author} {\bibfnamefont {W.}~\bibnamefont
  {Vogel}}\ and\ \bibinfo {author} {\bibfnamefont {D.~G.}\ \bibnamefont
  {Welsch}},\ }\href@noop {} {\emph {\bibinfo {title} {Quantum Optics}}},\
  \bibinfo {edition} {3rd}\ ed.\ (\bibinfo  {publisher} {Wiley-VCH},\ \bibinfo
  {address} {Weinheim},\ \bibinfo {year} {2006})\BibitemShut {NoStop}%
\bibitem [{\citenamefont {Planck}(1901)}]{Planck1901}%
  \BibitemOpen
  \bibfield  {author} {\bibinfo {author} {\bibfnamefont {M.}~\bibnamefont
  {Planck}},\ }\href@noop {} {\bibfield  {journal} {\bibinfo  {journal} {Ann.
  Phys.}\ }\textbf {\bibinfo {volume} {309}},\ \bibinfo {pages} {553} (\bibinfo
  {year} {1901})}\BibitemShut {NoStop}%
\bibitem [{\citenamefont {Einstein}(1905)}]{Einstein1905}%
  \BibitemOpen
  \bibfield  {author} {\bibinfo {author} {\bibfnamefont {A.}~\bibnamefont
  {Einstein}},\ }\href@noop {} {\bibfield  {journal} {\bibinfo  {journal} {Ann.
  Phys.}\ }\textbf {\bibinfo {volume} {17}},\ \bibinfo {pages} {132} (\bibinfo
  {year} {1905})}\BibitemShut {NoStop}%
\bibitem [{\citenamefont {Heisenberg}(1927)}]{Heisenberg1927}%
  \BibitemOpen
  \bibfield  {author} {\bibinfo {author} {\bibfnamefont {W.}~\bibnamefont
  {Heisenberg}},\ }\href@noop {} {\bibfield  {journal} {\bibinfo  {journal} {Z.
  Phys.}\ }\textbf {\bibinfo {volume} {43}},\ \bibinfo {pages} {172} (\bibinfo
  {year} {1927})}\BibitemShut {NoStop}%
\bibitem [{\citenamefont {Robertson}(1929)}]{Robertson1929}%
  \BibitemOpen
  \bibfield  {author} {\bibinfo {author} {\bibfnamefont {H.~P.}\ \bibnamefont
  {Robertson}},\ }\href@noop {} {\bibfield  {journal} {\bibinfo  {journal}
  {Phys. Rev.}\ }\textbf {\bibinfo {volume} {34}},\ \bibinfo {pages} {163}
  (\bibinfo {year} {1929})}\BibitemShut {NoStop}%
\bibitem [{\citenamefont {Glauber}(1963)}]{Glauber1963}%
  \BibitemOpen
  \bibfield  {author} {\bibinfo {author} {\bibfnamefont {R.~J.}\ \bibnamefont
  {Glauber}},\ }\href@noop {} {\bibfield  {journal} {\bibinfo  {journal} {Phys.
  Rev.}\ }\textbf {\bibinfo {volume} {131}},\ \bibinfo {pages} {2766} (\bibinfo
  {year} {1963})}\BibitemShut {NoStop}%
\bibitem [{\citenamefont {Schr{\"o}dinger}(1926)}]{Schrodinger1926}%
  \BibitemOpen
  \bibfield  {author} {\bibinfo {author} {\bibfnamefont {E.}~\bibnamefont
  {Schr{\"o}dinger}},\ }\href@noop {} {\bibfield  {journal} {\bibinfo
  {journal} {Naturwissenschaften}\ }\textbf {\bibinfo {volume} {14}},\ \bibinfo
  {pages} {664} (\bibinfo {year} {1926})}\BibitemShut {NoStop}%
\bibitem [{\citenamefont {Sudarshan}(1963)}]{Sudarshan1963}%
  \BibitemOpen
  \bibfield  {author} {\bibinfo {author} {\bibfnamefont {E.~C.~G.}\
  \bibnamefont {Sudarshan}},\ }\href@noop {} {\bibfield  {journal} {\bibinfo
  {journal} {Phs. Rev. Lett.}\ }\textbf {\bibinfo {volume} {10}},\ \bibinfo
  {pages} {277} (\bibinfo {year} {1963})}\BibitemShut {NoStop}%
\bibitem [{\citenamefont {M\o{}lmer}(1997)}]{Molmer1997}%
  \BibitemOpen
  \bibfield  {author} {\bibinfo {author} {\bibfnamefont {K.}~\bibnamefont
  {M\o{}lmer}},\ }\href@noop {} {\bibfield  {journal} {\bibinfo  {journal}
  {Phys. Rev. A}\ }\textbf {\bibinfo {volume} {55}},\ \bibinfo {pages} {3195}
  (\bibinfo {year} {1997})}\BibitemShut {NoStop}%
\bibitem [{\citenamefont {Rudolph}\ and\ \citenamefont
  {Sanders}(2001)}]{Rudolph2001}%
  \BibitemOpen
  \bibfield  {author} {\bibinfo {author} {\bibfnamefont {T.}~\bibnamefont
  {Rudolph}}\ and\ \bibinfo {author} {\bibfnamefont {B.~C.}\ \bibnamefont
  {Sanders}},\ }\href@noop {} {\bibfield  {journal} {\bibinfo  {journal} {Phys.
  Rev. Lett.}\ }\textbf {\bibinfo {volume} {87}},\ \bibinfo {pages} {077903}
  (\bibinfo {year} {2001})}\BibitemShut {NoStop}%
\bibitem [{\citenamefont {van Enk}\ and\ \citenamefont
  {Fuchs}(2001)}]{Enk2001}%
  \BibitemOpen
  \bibfield  {author} {\bibinfo {author} {\bibfnamefont {S.~J.}\ \bibnamefont
  {van Enk}}\ and\ \bibinfo {author} {\bibfnamefont {C.~A.}\ \bibnamefont
  {Fuchs}},\ }\href@noop {} {\bibfield  {journal} {\bibinfo  {journal} {Phys.
  Rev. Lett.}\ }\textbf {\bibinfo {volume} {88}},\ \bibinfo {pages} {027902}
  (\bibinfo {year} {2001})}\BibitemShut {NoStop}%
\bibitem [{\citenamefont {Wiseman}(2003)}]{Wiseman2003}%
  \BibitemOpen
  \bibfield  {author} {\bibinfo {author} {\bibfnamefont {H.~M.}\ \bibnamefont
  {Wiseman}},\ }\href@noop {} {\bibfield  {journal} {\bibinfo  {journal} {J.
  Mod. Opt.}\ }\textbf {\bibinfo {volume} {50}},\ \bibinfo {pages} {1797}
  (\bibinfo {year} {2003})}\BibitemShut {NoStop}%
\bibitem [{\citenamefont {Kiesel}\ and\ \citenamefont
  {Vogel}(2010)}]{Kiesel2010}%
  \BibitemOpen
  \bibfield  {author} {\bibinfo {author} {\bibfnamefont {T.}~\bibnamefont
  {Kiesel}}\ and\ \bibinfo {author} {\bibfnamefont {W.}~\bibnamefont {Vogel}},\
  }\href@noop {} {\bibfield  {journal} {\bibinfo  {journal} {Phys. Rev. A}\
  }\textbf {\bibinfo {volume} {82}},\ \bibinfo {pages} {032107} (\bibinfo
  {year} {2010})}\BibitemShut {NoStop}%
\bibitem [{\citenamefont {K{\" u}hn}\ and\ \citenamefont
  {Vogel}(2018)}]{Kuhn2018}%
  \BibitemOpen
  \bibfield  {author} {\bibinfo {author} {\bibfnamefont {B.}~\bibnamefont {K{\"
  u}hn}}\ and\ \bibinfo {author} {\bibfnamefont {W.}~\bibnamefont {Vogel}},\
  }\href@noop {} {\bibfield  {journal} {\bibinfo  {journal} {Phys. Rev. A}\
  }\textbf {\bibinfo {volume} {97}},\ \bibinfo {pages} {053823} (\bibinfo
  {year} {2018})}\BibitemShut {NoStop}%
\bibitem [{\citenamefont {Tan}, \citenamefont {Choi},\ and\ \citenamefont
  {Jeong}(2019)}]{Tan2019}%
  \BibitemOpen
  \bibfield  {author} {\bibinfo {author} {\bibfnamefont {K.~C.}\ \bibnamefont
  {Tan}}, \bibinfo {author} {\bibfnamefont {S.}~\bibnamefont {Choi}}, \ and\
  \bibinfo {author} {\bibfnamefont {H.}~\bibnamefont {Jeong}},\ }\href@noop {}
  {} (\bibinfo {year} {2019}),\ \Eprint {http://arxiv.org/abs/arXiv:1906.05579}
  {arXiv:1906.05579} \BibitemShut {NoStop}%
\bibitem [{\citenamefont {Cahill}\ and\ \citenamefont
  {Glauber}(1969{\natexlab{a}})}]{Cahill1969a}%
  \BibitemOpen
  \bibfield  {author} {\bibinfo {author} {\bibfnamefont {K.~E.}\ \bibnamefont
  {Cahill}}\ and\ \bibinfo {author} {\bibfnamefont {R.~J.}\ \bibnamefont
  {Glauber}},\ }\href@noop {} {\bibfield  {journal} {\bibinfo  {journal} {Phys.
  Rev.}\ }\textbf {\bibinfo {volume} {177}},\ \bibinfo {pages} {1857} (\bibinfo
  {year} {1969}{\natexlab{a}})}\BibitemShut {NoStop}%
\bibitem [{\citenamefont {Cahill}\ and\ \citenamefont
  {Glauber}(1969{\natexlab{b}})}]{Cahill1969b}%
  \BibitemOpen
  \bibfield  {author} {\bibinfo {author} {\bibfnamefont {K.~E.}\ \bibnamefont
  {Cahill}}\ and\ \bibinfo {author} {\bibfnamefont {R.~J.}\ \bibnamefont
  {Glauber}},\ }\href@noop {} {\bibfield  {journal} {\bibinfo  {journal} {Phys.
  Rev.}\ }\textbf {\bibinfo {volume} {177}},\ \bibinfo {pages} {1882} (\bibinfo
  {year} {1969}{\natexlab{b}})}\BibitemShut {NoStop}%
\bibitem [{\citenamefont {Wigner}(1932)}]{Wigner1932}%
  \BibitemOpen
  \bibfield  {author} {\bibinfo {author} {\bibfnamefont {E.}~\bibnamefont
  {Wigner}},\ }\href@noop {} {\bibfield  {journal} {\bibinfo  {journal} {Phys.
  Rev.}\ }\textbf {\bibinfo {volume} {40}},\ \bibinfo {pages} {749} (\bibinfo
  {year} {1932})}\BibitemShut {NoStop}%
\bibitem [{\citenamefont {Husimi}(1940)}]{Husimi1940}%
  \BibitemOpen
  \bibfield  {author} {\bibinfo {author} {\bibfnamefont {K.}~\bibnamefont
  {Husimi}},\ }\href@noop {} {\bibfield  {journal} {\bibinfo  {journal} {Proc.
  Phys. Math. Soc. Jpn.}\ }\textbf {\bibinfo {volume} {22}},\ \bibinfo {pages}
  {264--314} (\bibinfo {year} {1940})}\BibitemShut {NoStop}%
\bibitem [{\citenamefont {Kenfack}\ and\ \citenamefont
  {{\.Z}yczkowski}(2004)}]{Kenfack2004}%
  \BibitemOpen
  \bibfield  {author} {\bibinfo {author} {\bibfnamefont {A.}~\bibnamefont
  {Kenfack}}\ and\ \bibinfo {author} {\bibfnamefont {K.}~\bibnamefont
  {{\.Z}yczkowski}},\ }\href@noop {} {\bibfield  {journal} {\bibinfo  {journal}
  {J. Opt. B: Quantum Semiclass. Opt.}\ }\textbf {\bibinfo {volume} {6}},\
  \bibinfo {pages} {396} (\bibinfo {year} {2004})}\BibitemShut {NoStop}%
\bibitem [{\citenamefont {Schleich}(2001)}]{Schleich2001a}%
  \BibitemOpen
  \bibfield  {author} {\bibinfo {author} {\bibfnamefont {W.~P.}\ \bibnamefont
  {Schleich}},\ }\enquote {\bibinfo {title} {Quantum optics in phase space},}\
  \ (\bibinfo  {publisher} {Wiley-VCH},\ \bibinfo {address} {Berlin},\ \bibinfo
  {year} {2001})\ Chap.~\bibinfo {chapter} {12}, pp.\ \bibinfo {pages}
  {335--336},\ \bibinfo {edition} {1st}\ ed.\BibitemShut {Stop}%
\bibitem [{\citenamefont {Hofheinz}\ \emph {et~al.}(2008)\citenamefont
  {Hofheinz}, \citenamefont {Weig}, \citenamefont {Ansmann}, \citenamefont
  {Bialczak}, \citenamefont {Lucero}, \citenamefont {Neeley}, \citenamefont
  {O’Connell}, \citenamefont {Wang}, \citenamefont {Martinis},\ and\
  \citenamefont {Cleland}}]{Hofheinz2008}%
  \BibitemOpen
  \bibfield  {author} {\bibinfo {author} {\bibfnamefont {M.}~\bibnamefont
  {Hofheinz}}, \bibinfo {author} {\bibfnamefont {E.~M.}\ \bibnamefont {Weig}},
  \bibinfo {author} {\bibfnamefont {M.}~\bibnamefont {Ansmann}}, \bibinfo
  {author} {\bibfnamefont {R.~C.}\ \bibnamefont {Bialczak}}, \bibinfo {author}
  {\bibfnamefont {E.}~\bibnamefont {Lucero}}, \bibinfo {author} {\bibfnamefont
  {M.}~\bibnamefont {Neeley}}, \bibinfo {author} {\bibfnamefont {A.~D.}\
  \bibnamefont {O’Connell}}, \bibinfo {author} {\bibfnamefont
  {H.}~\bibnamefont {Wang}}, \bibinfo {author} {\bibfnamefont {J.~M.}\
  \bibnamefont {Martinis}}, \ and\ \bibinfo {author} {\bibfnamefont {A.~N.}\
  \bibnamefont {Cleland}},\ }\href@noop {} {\bibfield  {journal} {\bibinfo
  {journal} {Nature}\ }\textbf {\bibinfo {volume} {454}},\ \bibinfo {pages}
  {310} (\bibinfo {year} {2008})}\BibitemShut {NoStop}%
\bibitem [{\citenamefont {Cirac}\ \emph {et~al.}(1993)\citenamefont {Cirac},
  \citenamefont {Blatt}, \citenamefont {Parkins},\ and\ \citenamefont
  {Zoller}}]{Cirac1993}%
  \BibitemOpen
  \bibfield  {author} {\bibinfo {author} {\bibfnamefont {J.~I.}\ \bibnamefont
  {Cirac}}, \bibinfo {author} {\bibfnamefont {R.}~\bibnamefont {Blatt}},
  \bibinfo {author} {\bibfnamefont {A.~S.}\ \bibnamefont {Parkins}}, \ and\
  \bibinfo {author} {\bibfnamefont {P.}~\bibnamefont {Zoller}},\ }\href@noop {}
  {\bibfield  {journal} {\bibinfo  {journal} {Phys. Rev. Lett.}\ }\textbf
  {\bibinfo {volume} {70}},\ \bibinfo {pages} {762} (\bibinfo {year}
  {1993})}\BibitemShut {NoStop}%
\bibitem [{\citenamefont {Varcoe}\ \emph {et~al.}(2000)\citenamefont {Varcoe},
  \citenamefont {Brattke}, \citenamefont {Weidinger},\ and\ \citenamefont
  {Walther}}]{Varcoe2000}%
  \BibitemOpen
  \bibfield  {author} {\bibinfo {author} {\bibfnamefont {B.~T.~H.}\
  \bibnamefont {Varcoe}}, \bibinfo {author} {\bibfnamefont {S.}~\bibnamefont
  {Brattke}}, \bibinfo {author} {\bibfnamefont {M.}~\bibnamefont {Weidinger}},
  \ and\ \bibinfo {author} {\bibfnamefont {H.}~\bibnamefont {Walther}},\
  }\href@noop {} {\bibfield  {journal} {\bibinfo  {journal} {Nature}\ }\textbf
  {\bibinfo {volume} {403}},\ \bibinfo {pages} {743} (\bibinfo {year}
  {2000})}\BibitemShut {NoStop}%
\bibitem [{\citenamefont {Bertet}\ \emph {et~al.}(2002)\citenamefont {Bertet},
  \citenamefont {Osnaghi}, \citenamefont {Milman}, \citenamefont {Auffeves},
  \citenamefont {Maioli}, \citenamefont {Brune}, \citenamefont {Raimond},\ and\
  \citenamefont {Haroche}}]{Bertet2002}%
  \BibitemOpen
  \bibfield  {author} {\bibinfo {author} {\bibfnamefont {P.}~\bibnamefont
  {Bertet}}, \bibinfo {author} {\bibfnamefont {S.}~\bibnamefont {Osnaghi}},
  \bibinfo {author} {\bibfnamefont {P.}~\bibnamefont {Milman}}, \bibinfo
  {author} {\bibfnamefont {A.}~\bibnamefont {Auffeves}}, \bibinfo {author}
  {\bibfnamefont {P.}~\bibnamefont {Maioli}}, \bibinfo {author} {\bibfnamefont
  {M.}~\bibnamefont {Brune}}, \bibinfo {author} {\bibfnamefont {J.~M.}\
  \bibnamefont {Raimond}}, \ and\ \bibinfo {author} {\bibfnamefont
  {S.}~\bibnamefont {Haroche}},\ }\href@noop {} {\bibfield  {journal} {\bibinfo
   {journal} {Phys. Rev. Lett.}\ }\textbf {\bibinfo {volume} {88}},\ \bibinfo
  {pages} {143601} (\bibinfo {year} {2002})}\BibitemShut {NoStop}%
\bibitem [{\citenamefont {Andrews}(2014)}]{Andrews2014}%
  \BibitemOpen
  \bibinfo {editor} {\bibfnamefont {D.~L.}\ \bibnamefont {Andrews}},\ ed.,\
  \href@noop {} {\emph {\bibinfo {title} {Photonics, Volume 1: Fundamentals of
  photonics and physics}}}\ (\bibinfo  {publisher} {Wiley},\ \bibinfo {address}
  {Hoboken},\ \bibinfo {year} {2014})\BibitemShut {NoStop}%
\bibitem [{\citenamefont {Lvovsky}(2016)}]{Lvovsky2016}%
  \BibitemOpen
  \bibfield  {author} {\bibinfo {author} {\bibfnamefont {A.~I.}\ \bibnamefont
  {Lvovsky}},\ }\href@noop {} {} (\bibinfo {year} {2016}),\ \Eprint
  {http://arxiv.org/abs/arXiv:1401.4118} {arXiv:1401.4118} \BibitemShut
  {NoStop}%
\bibitem [{\citenamefont {Loudon}\ and\ \citenamefont
  {Knight}(1987)}]{Loudon1987}%
  \BibitemOpen
  \bibfield  {author} {\bibinfo {author} {\bibfnamefont {R.}~\bibnamefont
  {Loudon}}\ and\ \bibinfo {author} {\bibfnamefont {P.}~\bibnamefont
  {Knight}},\ }\href@noop {} {\bibfield  {journal} {\bibinfo  {journal} {J.
  Mod. Opt.}\ }\textbf {\bibinfo {volume} {34}},\ \bibinfo {pages} {709--759}
  (\bibinfo {year} {1987})}\BibitemShut {NoStop}%
\bibitem [{\citenamefont {Slusher}\ \emph {et~al.}(1985)\citenamefont
  {Slusher}, \citenamefont {Hollberg}, \citenamefont {Yurke}, \citenamefont
  {Mertz},\ and\ \citenamefont {Valley}}]{Slusher1985}%
  \BibitemOpen
  \bibfield  {author} {\bibinfo {author} {\bibfnamefont {R.~E.}\ \bibnamefont
  {Slusher}}, \bibinfo {author} {\bibfnamefont {L.~W.}\ \bibnamefont
  {Hollberg}}, \bibinfo {author} {\bibfnamefont {B.}~\bibnamefont {Yurke}},
  \bibinfo {author} {\bibfnamefont {J.~C.}\ \bibnamefont {Mertz}}, \ and\
  \bibinfo {author} {\bibfnamefont {J.~F.}\ \bibnamefont {Valley}},\
  }\href@noop {} {\bibfield  {journal} {\bibinfo  {journal} {Phys. Rev. Lett.}\
  }\textbf {\bibinfo {volume} {55}},\ \bibinfo {pages} {2409} (\bibinfo {year}
  {1985})}\BibitemShut {NoStop}%
\bibitem [{\citenamefont {Slusher}\ \emph {et~al.}(1987)\citenamefont
  {Slusher}, \citenamefont {Grangier}, \citenamefont {LaPorta}, \citenamefont
  {Yurke},\ and\ \citenamefont {Potasek}}]{Slusher1987}%
  \BibitemOpen
  \bibfield  {author} {\bibinfo {author} {\bibfnamefont {R.~E.}\ \bibnamefont
  {Slusher}}, \bibinfo {author} {\bibfnamefont {P.}~\bibnamefont {Grangier}},
  \bibinfo {author} {\bibfnamefont {A.}~\bibnamefont {LaPorta}}, \bibinfo
  {author} {\bibfnamefont {B.}~\bibnamefont {Yurke}}, \ and\ \bibinfo {author}
  {\bibfnamefont {M.~J.}\ \bibnamefont {Potasek}},\ }\href@noop {} {\bibfield
  {journal} {\bibinfo  {journal} {Phys. Rev. Lett.}\ }\textbf {\bibinfo
  {volume} {59}},\ \bibinfo {pages} {2566} (\bibinfo {year}
  {1987})}\BibitemShut {NoStop}%
\bibitem [{\citenamefont {Kim}\ and\ \citenamefont {Kumar}(1994)}]{Kim1994}%
  \BibitemOpen
  \bibfield  {author} {\bibinfo {author} {\bibfnamefont {C.}~\bibnamefont
  {Kim}}\ and\ \bibinfo {author} {\bibfnamefont {P.}~\bibnamefont {Kumar}},\
  }\href@noop {} {\bibfield  {journal} {\bibinfo  {journal} {Phys. Rev. Lett.}\
  }\textbf {\bibinfo {volume} {73}},\ \bibinfo {pages} {1605} (\bibinfo {year}
  {1994})}\BibitemShut {NoStop}%
\bibitem [{\citenamefont {Yurke}\ and\ \citenamefont
  {Stoler}(1986)}]{Yurke1986}%
  \BibitemOpen
  \bibfield  {author} {\bibinfo {author} {\bibfnamefont {B.}~\bibnamefont
  {Yurke}}\ and\ \bibinfo {author} {\bibfnamefont {D.}~\bibnamefont {Stoler}},\
  }\href@noop {} {\bibfield  {journal} {\bibinfo  {journal} {Phys. Rev. Lett.}\
  }\textbf {\bibinfo {volume} {57}},\ \bibinfo {pages} {13} (\bibinfo {year}
  {1986})}\BibitemShut {NoStop}%
\bibitem [{\citenamefont {Milburn}(1986)}]{Milburn1986}%
  \BibitemOpen
  \bibfield  {author} {\bibinfo {author} {\bibfnamefont {G.~J.}\ \bibnamefont
  {Milburn}},\ }\href@noop {} {\bibfield  {journal} {\bibinfo  {journal} {Phys.
  Rev. A}\ }\textbf {\bibinfo {volume} {33}},\ \bibinfo {pages} {674} (\bibinfo
  {year} {1986})}\BibitemShut {NoStop}%
\bibitem [{\citenamefont {Milburn}\ and\ \citenamefont
  {Holmes}(1986)}]{Milburn1986a}%
  \BibitemOpen
  \bibfield  {author} {\bibinfo {author} {\bibfnamefont {G.~J.}\ \bibnamefont
  {Milburn}}\ and\ \bibinfo {author} {\bibfnamefont {C.~A.}\ \bibnamefont
  {Holmes}},\ }\href@noop {} {\bibfield  {journal} {\bibinfo  {journal} {Phys.
  Rev. Lett.}\ }\textbf {\bibinfo {volume} {56}},\ \bibinfo {pages} {2237}
  (\bibinfo {year} {1986})}\BibitemShut {NoStop}%
\bibitem [{\citenamefont {Schleich}, \citenamefont {Pernigo},\ and\
  \citenamefont {Kien}(1991)}]{Schleich1991}%
  \BibitemOpen
  \bibfield  {author} {\bibinfo {author} {\bibfnamefont {W.}~\bibnamefont
  {Schleich}}, \bibinfo {author} {\bibfnamefont {M.}~\bibnamefont {Pernigo}}, \
  and\ \bibinfo {author} {\bibfnamefont {F.~L.}\ \bibnamefont {Kien}},\
  }\href@noop {} {\bibfield  {journal} {\bibinfo  {journal} {Phys. Rev. A}\
  }\textbf {\bibinfo {volume} {44}},\ \bibinfo {pages} {2172} (\bibinfo {year}
  {1991})}\BibitemShut {NoStop}%
\bibitem [{\citenamefont {Brune}\ \emph {et~al.}(1992)\citenamefont {Brune},
  \citenamefont {Haroche}, \citenamefont {Raimond}, \citenamefont
  {Davidovich},\ and\ \citenamefont {Zachary}}]{Brune1992}%
  \BibitemOpen
  \bibfield  {author} {\bibinfo {author} {\bibfnamefont {M.}~\bibnamefont
  {Brune}}, \bibinfo {author} {\bibfnamefont {S.}~\bibnamefont {Haroche}},
  \bibinfo {author} {\bibfnamefont {J.~M.}\ \bibnamefont {Raimond}}, \bibinfo
  {author} {\bibfnamefont {L.}~\bibnamefont {Davidovich}}, \ and\ \bibinfo
  {author} {\bibfnamefont {N.}~\bibnamefont {Zachary}},\ }\href@noop {}
  {\bibfield  {journal} {\bibinfo  {journal} {Phys. Rev. A}\ }\textbf {\bibinfo
  {volume} {45}},\ \bibinfo {pages} {7} (\bibinfo {year} {1992})}\BibitemShut
  {NoStop}%
\bibitem [{\citenamefont {Schr{\"o}dinger}(1935)}]{Schrodinger1935}%
  \BibitemOpen
  \bibfield  {author} {\bibinfo {author} {\bibfnamefont {E.}~\bibnamefont
  {Schr{\"o}dinger}},\ }\href@noop {} {\bibfield  {journal} {\bibinfo
  {journal} {Naturwissenschaften}\ }\textbf {\bibinfo {volume} {23}},\ \bibinfo
  {pages} {807} (\bibinfo {year} {1935})}\BibitemShut {NoStop}%
\bibitem [{\citenamefont {Lee}\ and\ \citenamefont {Jeong}(2011)}]{Lee2011}%
  \BibitemOpen
  \bibfield  {author} {\bibinfo {author} {\bibfnamefont {C.-W.}\ \bibnamefont
  {Lee}}\ and\ \bibinfo {author} {\bibfnamefont {H.}~\bibnamefont {Jeong}},\
  }\href@noop {} {\bibfield  {journal} {\bibinfo  {journal} {Phys. Rev. Lett.}\
  }\textbf {\bibinfo {volume} {106}},\ \bibinfo {pages} {220401} (\bibinfo
  {year} {2011})}\BibitemShut {NoStop}%
\bibitem [{\citenamefont {Agarwal}\ and\ \citenamefont
  {Tara}(1992)}]{Agarwal1992}%
  \BibitemOpen
  \bibfield  {author} {\bibinfo {author} {\bibfnamefont {G.~S.}\ \bibnamefont
  {Agarwal}}\ and\ \bibinfo {author} {\bibfnamefont {K.}~\bibnamefont {Tara}},\
  }\href@noop {} {\bibfield  {journal} {\bibinfo  {journal} {Phys. Rev. A}\
  }\textbf {\bibinfo {volume} {46}},\ \bibinfo {pages} {485} (\bibinfo {year}
  {1992})}\BibitemShut {NoStop}%
\bibitem [{\citenamefont {Kiesel}\ \emph {et~al.}(2008)\citenamefont {Kiesel},
  \citenamefont {Vogel}, \citenamefont {Parigi}, \citenamefont {Zavatta},\ and\
  \citenamefont {Bellini}}]{Kiesel2008}%
  \BibitemOpen
  \bibfield  {author} {\bibinfo {author} {\bibfnamefont {T.}~\bibnamefont
  {Kiesel}}, \bibinfo {author} {\bibfnamefont {W.}~\bibnamefont {Vogel}},
  \bibinfo {author} {\bibfnamefont {V.}~\bibnamefont {Parigi}}, \bibinfo
  {author} {\bibfnamefont {A.}~\bibnamefont {Zavatta}}, \ and\ \bibinfo
  {author} {\bibfnamefont {M.}~\bibnamefont {Bellini}},\ }\href@noop {}
  {\bibfield  {journal} {\bibinfo  {journal} {Phys. Rev. A}\ }\textbf {\bibinfo
  {volume} {78}},\ \bibinfo {pages} {021804} (\bibinfo {year}
  {2008})}\BibitemShut {NoStop}%
\bibitem [{\citenamefont {Damanet}\ \emph {et~al.}(2018)\citenamefont
  {Damanet}, \citenamefont {K{\"u}bler}, \citenamefont {Martin},\ and\
  \citenamefont {Braun}}]{Damanet2018}%
  \BibitemOpen
  \bibfield  {author} {\bibinfo {author} {\bibfnamefont {F.}~\bibnamefont
  {Damanet}}, \bibinfo {author} {\bibfnamefont {J.}~\bibnamefont {K{\"u}bler}},
  \bibinfo {author} {\bibfnamefont {J.}~\bibnamefont {Martin}}, \ and\ \bibinfo
  {author} {\bibfnamefont {D.}~\bibnamefont {Braun}},\ }\href@noop {}
  {\bibfield  {journal} {\bibinfo  {journal} {Phys. Rev. A}\ }\textbf {\bibinfo
  {volume} {97}},\ \bibinfo {pages} {023832} (\bibinfo {year}
  {2018})}\BibitemShut {NoStop}%
\bibitem [{\citenamefont {Lee}(1991)}]{Lee1991}%
  \BibitemOpen
  \bibfield  {author} {\bibinfo {author} {\bibfnamefont {C.~T.}\ \bibnamefont
  {Lee}},\ }\href@noop {} {\bibfield  {journal} {\bibinfo  {journal} {Phys.
  Rev. A}\ }\textbf {\bibinfo {volume} {44}},\ \bibinfo {pages} {R2775}
  (\bibinfo {year} {1991})}\BibitemShut {NoStop}%
\bibitem [{\citenamefont {K{\"u}hn}\ and\ \citenamefont
  {Vogel}(2018)}]{Kuhn2018a}%
  \BibitemOpen
  \bibfield  {author} {\bibinfo {author} {\bibfnamefont {B.}~\bibnamefont
  {K{\"u}hn}}\ and\ \bibinfo {author} {\bibfnamefont {W.}~\bibnamefont
  {Vogel}},\ }\href@noop {} {\bibfield  {journal} {\bibinfo  {journal} {Phys.
  Rev. A}\ }\textbf {\bibinfo {volume} {98}},\ \bibinfo {pages} {053807}
  (\bibinfo {year} {2018})}\BibitemShut {NoStop}%
\bibitem [{\citenamefont {Sperling}(2016)}]{Sperling2016}%
  \BibitemOpen
  \bibfield  {author} {\bibinfo {author} {\bibfnamefont {J.}~\bibnamefont
  {Sperling}},\ }\href@noop {} {\bibfield  {journal} {\bibinfo  {journal}
  {Phys. Rev. A}\ }\textbf {\bibinfo {volume} {94}},\ \bibinfo {pages} {013814}
  (\bibinfo {year} {2016})}\BibitemShut {NoStop}%
\bibitem [{\citenamefont {Mandel}(1979)}]{Mandel1979}%
  \BibitemOpen
  \bibfield  {author} {\bibinfo {author} {\bibfnamefont {L.}~\bibnamefont
  {Mandel}},\ }\href@noop {} {\bibfield  {journal} {\bibinfo  {journal} {Opt.
  Lett.}\ }\textbf {\bibinfo {volume} {4}},\ \bibinfo {pages} {205} (\bibinfo
  {year} {1979})}\BibitemShut {NoStop}%
\bibitem [{\citenamefont {Hillery}(1985)}]{Hillery1985}%
  \BibitemOpen
  \bibfield  {author} {\bibinfo {author} {\bibfnamefont {M.}~\bibnamefont
  {Hillery}},\ }\href@noop {} {\bibfield  {journal} {\bibinfo  {journal} {Phys.
  Lett. A}\ }\textbf {\bibinfo {volume} {111}},\ \bibinfo {pages} {409}
  (\bibinfo {year} {1985})}\BibitemShut {NoStop}%
\bibitem [{\citenamefont {Rudin}(1987)}]{Rudin1987}%
  \BibitemOpen
  \bibfield  {author} {\bibinfo {author} {\bibfnamefont {W.}~\bibnamefont
  {Rudin}},\ }\href@noop {} {\emph {\bibinfo {title} {Real and complex
  analysis}}},\ \bibinfo {edition} {3rd}\ ed.\ (\bibinfo  {publisher}
  {McGraw-Hill},\ \bibinfo {address} {New York},\ \bibinfo {year}
  {1987})\BibitemShut {NoStop}%
\bibitem [{\citenamefont {Hanbury-Brown}\ and\ \citenamefont
  {Twiss}(1956)}]{HanburyBrown1956}%
  \BibitemOpen
  \bibfield  {author} {\bibinfo {author} {\bibfnamefont {R.}~\bibnamefont
  {Hanbury-Brown}}\ and\ \bibinfo {author} {\bibfnamefont {R.~Q.}\ \bibnamefont
  {Twiss}},\ }\href@noop {} {\bibfield  {journal} {\bibinfo  {journal}
  {Nature}\ }\textbf {\bibinfo {volume} {177}},\ \bibinfo {pages} {27}
  (\bibinfo {year} {1956})}\BibitemShut {NoStop}%
\bibitem [{\citenamefont {Short}\ and\ \citenamefont
  {Mandel}(1983)}]{Short1983}%
  \BibitemOpen
  \bibfield  {author} {\bibinfo {author} {\bibfnamefont {R.}~\bibnamefont
  {Short}}\ and\ \bibinfo {author} {\bibfnamefont {L.}~\bibnamefont {Mandel}},\
  }\href@noop {} {\bibfield  {journal} {\bibinfo  {journal} {Phys. Rev. Lett.}\
  }\textbf {\bibinfo {volume} {51}},\ \bibinfo {pages} {384} (\bibinfo {year}
  {1983})}\BibitemShut {NoStop}%
\bibitem [{\citenamefont {Hong}\ and\ \citenamefont {Mandel}(1986)}]{Hong1986}%
  \BibitemOpen
  \bibfield  {author} {\bibinfo {author} {\bibfnamefont {C.~K.}\ \bibnamefont
  {Hong}}\ and\ \bibinfo {author} {\bibfnamefont {L.}~\bibnamefont {Mandel}},\
  }\href@noop {} {\bibfield  {journal} {\bibinfo  {journal} {Phys. Rev. Lett.}\
  }\textbf {\bibinfo {volume} {56}},\ \bibinfo {pages} {58} (\bibinfo {year}
  {1986})}\BibitemShut {NoStop}%
\bibitem [{\citenamefont {Kimble}, \citenamefont {Dagenais},\ and\
  \citenamefont {Mandel}(1977)}]{Kimble1977}%
  \BibitemOpen
  \bibfield  {author} {\bibinfo {author} {\bibfnamefont {H.~J.}\ \bibnamefont
  {Kimble}}, \bibinfo {author} {\bibfnamefont {M.}~\bibnamefont {Dagenais}}, \
  and\ \bibinfo {author} {\bibfnamefont {L.}~\bibnamefont {Mandel}},\
  }\href@noop {} {\bibfield  {journal} {\bibinfo  {journal} {Phys. Rev. Lett.}\
  }\textbf {\bibinfo {volume} {39}},\ \bibinfo {pages} {691} (\bibinfo {year}
  {1977})}\BibitemShut {NoStop}%
\bibitem [{\citenamefont {Hillery}(1987)}]{Hillery1987}%
  \BibitemOpen
  \bibfield  {author} {\bibinfo {author} {\bibfnamefont {M.}~\bibnamefont
  {Hillery}},\ }\href@noop {} {\bibfield  {journal} {\bibinfo  {journal} {Phys.
  Rev. A}\ }\textbf {\bibinfo {volume} {35}},\ \bibinfo {pages} {725} (\bibinfo
  {year} {1987})}\BibitemShut {NoStop}%
\bibitem [{\citenamefont {Nielsen}\ and\ \citenamefont
  {Chuang}(2000)}]{NielsenChuang}%
  \BibitemOpen
  \bibfield  {author} {\bibinfo {author} {\bibfnamefont {M.~A.}\ \bibnamefont
  {Nielsen}}\ and\ \bibinfo {author} {\bibfnamefont {I.~L.}\ \bibnamefont
  {Chuang}},\ }\href@noop {} {\emph {\bibinfo {title} {Quantum computation and
  quantum information}}}\ (\bibinfo  {publisher} {Cambridge University Press},\
  \bibinfo {address} {New York},\ \bibinfo {year} {2000})\BibitemShut {NoStop}%
\bibitem [{\citenamefont {Dodonov}\ \emph {et~al.}(2000)\citenamefont
  {Dodonov}, \citenamefont {Man'ko}, \citenamefont {Man'ko},\ and\
  \citenamefont {W{\" u}nsche}}]{Dodonov2000}%
  \BibitemOpen
  \bibfield  {author} {\bibinfo {author} {\bibfnamefont {V.~V.}\ \bibnamefont
  {Dodonov}}, \bibinfo {author} {\bibfnamefont {O.~V.}\ \bibnamefont {Man'ko}},
  \bibinfo {author} {\bibfnamefont {V.~I.}\ \bibnamefont {Man'ko}}, \ and\
  \bibinfo {author} {\bibfnamefont {A.}~\bibnamefont {W{\" u}nsche}},\
  }\href@noop {} {\bibfield  {journal} {\bibinfo  {journal} {J. Mod. Opt.}\
  }\textbf {\bibinfo {volume} {47}},\ \bibinfo {pages} {633} (\bibinfo {year}
  {2000})}\BibitemShut {NoStop}%
\bibitem [{\citenamefont {Marian}, \citenamefont {Marian},\ and\ \citenamefont
  {Scutaru}(2002)}]{Marian2002}%
  \BibitemOpen
  \bibfield  {author} {\bibinfo {author} {\bibfnamefont {P.}~\bibnamefont
  {Marian}}, \bibinfo {author} {\bibfnamefont {T.~A.}\ \bibnamefont {Marian}},
  \ and\ \bibinfo {author} {\bibfnamefont {H.}~\bibnamefont {Scutaru}},\
  }\href@noop {} {\bibfield  {journal} {\bibinfo  {journal} {Phys. Rev. Lett.}\
  }\textbf {\bibinfo {volume} {88}},\ \bibinfo {pages} {153601} (\bibinfo
  {year} {2002})}\BibitemShut {NoStop}%
\bibitem [{\citenamefont {Malbouisson}\ and\ \citenamefont
  {Baseia}(2003)}]{Malbouisson2003}%
  \BibitemOpen
  \bibfield  {author} {\bibinfo {author} {\bibfnamefont {J.~M.~C.}\
  \bibnamefont {Malbouisson}}\ and\ \bibinfo {author} {\bibfnamefont
  {B.}~\bibnamefont {Baseia}},\ }\href@noop {} {\bibfield  {journal} {\bibinfo
  {journal} {Phys. Scr.}\ }\textbf {\bibinfo {volume} {67}},\ \bibinfo {pages}
  {93} (\bibinfo {year} {2003})}\BibitemShut {NoStop}%
\bibitem [{\citenamefont {Vogel}\ and\ \citenamefont
  {Risken}(1989)}]{Vogel1989}%
  \BibitemOpen
  \bibfield  {author} {\bibinfo {author} {\bibfnamefont {K.}~\bibnamefont
  {Vogel}}\ and\ \bibinfo {author} {\bibfnamefont {H.}~\bibnamefont {Risken}},\
  }\href@noop {} {\bibfield  {journal} {\bibinfo  {journal} {Phys. Rev. A}\
  }\textbf {\bibinfo {volume} {40}},\ \bibinfo {pages} {2847(R)} (\bibinfo
  {year} {1989})}\BibitemShut {NoStop}%
\bibitem [{\citenamefont {L{\"u}tkenhaus}\ and\ \citenamefont
  {Barnett}(1995)}]{Lutkenhaus1995}%
  \BibitemOpen
  \bibfield  {author} {\bibinfo {author} {\bibfnamefont {N.}~\bibnamefont
  {L{\"u}tkenhaus}}\ and\ \bibinfo {author} {\bibfnamefont {S.~M.}\
  \bibnamefont {Barnett}},\ }\href@noop {} {\bibfield  {journal} {\bibinfo
  {journal} {Phys. Rev. A}\ }\textbf {\bibinfo {volume} {51}},\ \bibinfo
  {pages} {3340} (\bibinfo {year} {1995})}\BibitemShut {NoStop}%
\bibitem [{\citenamefont {Lee}(1995)}]{Lee1995}%
  \BibitemOpen
  \bibfield  {author} {\bibinfo {author} {\bibfnamefont {C.~T.}\ \bibnamefont
  {Lee}},\ }\href@noop {} {\bibfield  {journal} {\bibinfo  {journal} {Phys.
  Rev. A}\ }\textbf {\bibinfo {volume} {52}},\ \bibinfo {pages} {3374}
  (\bibinfo {year} {1995})}\BibitemShut {NoStop}%
\bibitem [{\citenamefont {Lvovsky}\ and\ \citenamefont
  {Raymer}(2009)}]{Lvovsky2009}%
  \BibitemOpen
  \bibfield  {author} {\bibinfo {author} {\bibfnamefont {A.~I.}\ \bibnamefont
  {Lvovsky}}\ and\ \bibinfo {author} {\bibfnamefont {M.~G.}\ \bibnamefont
  {Raymer}},\ }\href@noop {} {\bibfield  {journal} {\bibinfo  {journal} {Rev.
  Mod. Phys.}\ }\textbf {\bibinfo {volume} {81}},\ \bibinfo {pages} {299}
  (\bibinfo {year} {2009})}\BibitemShut {NoStop}%
\bibitem [{\citenamefont {Hudson}(1974)}]{Hudson1974}%
  \BibitemOpen
  \bibfield  {author} {\bibinfo {author} {\bibfnamefont {R.}~\bibnamefont
  {Hudson}},\ }\href@noop {} {\bibfield  {journal} {\bibinfo  {journal} {Rep.
  Math. Phys.}\ }\textbf {\bibinfo {volume} {6}},\ \bibinfo {pages} {249}
  (\bibinfo {year} {1974})}\BibitemShut {NoStop}%
\bibitem [{\citenamefont {Horodecki}\ \emph {et~al.}(2009)\citenamefont
  {Horodecki}, \citenamefont {Horodecki}, \citenamefont {Horodecki},\ and\
  \citenamefont {Horodecki}}]{Horodecki2009}%
  \BibitemOpen
  \bibfield  {author} {\bibinfo {author} {\bibfnamefont {R.}~\bibnamefont
  {Horodecki}}, \bibinfo {author} {\bibfnamefont {P.}~\bibnamefont
  {Horodecki}}, \bibinfo {author} {\bibfnamefont {M.}~\bibnamefont
  {Horodecki}}, \ and\ \bibinfo {author} {\bibfnamefont {K.}~\bibnamefont
  {Horodecki}},\ }\href@noop {} {\bibfield  {journal} {\bibinfo  {journal}
  {Rev. Mod. Phys.}\ }\textbf {\bibinfo {volume} {81}},\ \bibinfo {pages} {865}
  (\bibinfo {year} {2009})}\BibitemShut {NoStop}%
\bibitem [{\citenamefont {Aharonov}\ \emph {et~al.}(1966)\citenamefont
  {Aharonov}, \citenamefont {Falkoff}, \citenamefont {Lerner},\ and\
  \citenamefont {Pendleton}}]{Aharonov1966}%
  \BibitemOpen
  \bibfield  {author} {\bibinfo {author} {\bibfnamefont {Y.}~\bibnamefont
  {Aharonov}}, \bibinfo {author} {\bibfnamefont {D.}~\bibnamefont {Falkoff}},
  \bibinfo {author} {\bibfnamefont {E.}~\bibnamefont {Lerner}}, \ and\ \bibinfo
  {author} {\bibfnamefont {H.}~\bibnamefont {Pendleton}},\ }\href@noop {}
  {\bibfield  {journal} {\bibinfo  {journal} {Ann. Phys. (N.Y.)}\ }\textbf
  {\bibinfo {volume} {39}},\ \bibinfo {pages} {498} (\bibinfo {year}
  {1966})}\BibitemShut {NoStop}%
\bibitem [{\citenamefont {Kim}\ \emph {et~al.}(2002)\citenamefont {Kim},
  \citenamefont {Son}, \citenamefont {Bu{\v z}ek},\ and\ \citenamefont
  {Knight}}]{kim2002}%
  \BibitemOpen
  \bibfield  {author} {\bibinfo {author} {\bibfnamefont {M.~S.}\ \bibnamefont
  {Kim}}, \bibinfo {author} {\bibfnamefont {W.}~\bibnamefont {Son}}, \bibinfo
  {author} {\bibfnamefont {V.}~\bibnamefont {Bu{\v z}ek}}, \ and\ \bibinfo
  {author} {\bibfnamefont {P.~L.}\ \bibnamefont {Knight}},\ }\href@noop {}
  {\bibfield  {journal} {\bibinfo  {journal} {Phys. Rev. A}\ }\textbf {\bibinfo
  {volume} {65}},\ \bibinfo {pages} {032323} (\bibinfo {year}
  {2002})}\BibitemShut {NoStop}%
\bibitem [{\citenamefont {{-}b. Wang}(2002)}]{Wang2002}%
  \BibitemOpen
  \bibfield  {author} {\bibinfo {author} {\bibfnamefont {X.}~\bibnamefont
  {{-}b. Wang}},\ }\href@noop {} {\bibfield  {journal} {\bibinfo  {journal}
  {Phys. Rev. A}\ }\textbf {\bibinfo {volume} {66}},\ \bibinfo {pages} {024303}
  (\bibinfo {year} {2002})}\BibitemShut {NoStop}%
\bibitem [{\citenamefont {Asb{\'o}th}, \citenamefont {Calsamiglia},\ and\
  \citenamefont {Ritsch}(2005)}]{Asboth2005}%
  \BibitemOpen
  \bibfield  {author} {\bibinfo {author} {\bibfnamefont {J.~K.}\ \bibnamefont
  {Asb{\'o}th}}, \bibinfo {author} {\bibfnamefont {J.}~\bibnamefont
  {Calsamiglia}}, \ and\ \bibinfo {author} {\bibfnamefont {H.}~\bibnamefont
  {Ritsch}},\ }\href@noop {} {\bibfield  {journal} {\bibinfo  {journal} {Phys.
  Rev. Lett.}\ }\textbf {\bibinfo {volume} {94}},\ \bibinfo {pages} {173602}
  (\bibinfo {year} {2005})}\BibitemShut {NoStop}%
\bibitem [{\citenamefont {Horodecki}, \citenamefont {Horodecki},\ and\
  \citenamefont {Horodecki}(1998)}]{Horodecki1998}%
  \BibitemOpen
  \bibfield  {author} {\bibinfo {author} {\bibfnamefont {M.}~\bibnamefont
  {Horodecki}}, \bibinfo {author} {\bibfnamefont {P.}~\bibnamefont
  {Horodecki}}, \ and\ \bibinfo {author} {\bibfnamefont {R.}~\bibnamefont
  {Horodecki}},\ }\href@noop {} {\bibfield  {journal} {\bibinfo  {journal}
  {Phys. Rev. Lett.}\ }\textbf {\bibinfo {volume} {80}},\ \bibinfo {pages}
  {5239} (\bibinfo {year} {1998})}\BibitemShut {NoStop}%
\bibitem [{\citenamefont {Horodecki}\ and\ \citenamefont
  {Lewenstein}(2000)}]{Horodecki2000}%
  \BibitemOpen
  \bibfield  {author} {\bibinfo {author} {\bibfnamefont {P.}~\bibnamefont
  {Horodecki}}\ and\ \bibinfo {author} {\bibfnamefont {M.}~\bibnamefont
  {Lewenstein}},\ }\href@noop {} {\bibfield  {journal} {\bibinfo  {journal}
  {Phys. Rev. Lett.}\ }\textbf {\bibinfo {volume} {85}},\ \bibinfo {pages}
  {2657} (\bibinfo {year} {2000})}\BibitemShut {NoStop}%
\bibitem [{\citenamefont {Shchukin}, \citenamefont {Richter},\ and\
  \citenamefont {Vogel}(2005)}]{Shchukin2005}%
  \BibitemOpen
  \bibfield  {author} {\bibinfo {author} {\bibfnamefont {E.}~\bibnamefont
  {Shchukin}}, \bibinfo {author} {\bibfnamefont {T.}~\bibnamefont {Richter}}, \
  and\ \bibinfo {author} {\bibfnamefont {W.}~\bibnamefont {Vogel}},\
  }\href@noop {} {\bibfield  {journal} {\bibinfo  {journal} {Phys. Rev. A}\
  }\textbf {\bibinfo {volume} {71}},\ \bibinfo {pages} {011802(R)} (\bibinfo
  {year} {2005})}\BibitemShut {NoStop}%
\bibitem [{\citenamefont {Gehrke}, \citenamefont {Sperling},\ and\
  \citenamefont {Vogel}(2012)}]{Gehrke2012}%
  \BibitemOpen
  \bibfield  {author} {\bibinfo {author} {\bibfnamefont {C.}~\bibnamefont
  {Gehrke}}, \bibinfo {author} {\bibfnamefont {J.}~\bibnamefont {Sperling}}, \
  and\ \bibinfo {author} {\bibfnamefont {W.}~\bibnamefont {Vogel}},\
  }\href@noop {} {\bibfield  {journal} {\bibinfo  {journal} {Phys. Rev. A}\
  }\textbf {\bibinfo {volume} {86}},\ \bibinfo {pages} {052118} (\bibinfo
  {year} {2012})}\BibitemShut {NoStop}%
\bibitem [{\citenamefont {Schmidt}(1906)}]{Schmidt1906}%
  \BibitemOpen
  \bibfield  {author} {\bibinfo {author} {\bibfnamefont {E.}~\bibnamefont
  {Schmidt}},\ }\href@noop {} {\bibfield  {journal} {\bibinfo  {journal} {Math.
  Ann.}\ }\textbf {\bibinfo {volume} {63}},\ \bibinfo {pages} {433} (\bibinfo
  {year} {1906})}\BibitemShut {NoStop}%
\bibitem [{\citenamefont {Terhal}\ and\ \citenamefont
  {Horodecki}(1999)}]{Terhal1999}%
  \BibitemOpen
  \bibfield  {author} {\bibinfo {author} {\bibfnamefont {B.~M.}\ \bibnamefont
  {Terhal}}\ and\ \bibinfo {author} {\bibfnamefont {P.}~\bibnamefont
  {Horodecki}},\ }\href@noop {} {\bibfield  {journal} {\bibinfo  {journal}
  {Phys. Rev. A}\ }\textbf {\bibinfo {volume} {61}},\ \bibinfo {pages} {040301}
  (\bibinfo {year} {1999})}\BibitemShut {NoStop}%
\bibitem [{\citenamefont {Bennett}\ \emph {et~al.}(1996)\citenamefont
  {Bennett}, \citenamefont {Vincenzo}, \citenamefont {Smolin},\ and\
  \citenamefont {K.Wootters}}]{Bennett1996}%
  \BibitemOpen
  \bibfield  {author} {\bibinfo {author} {\bibfnamefont {C.~H.}\ \bibnamefont
  {Bennett}}, \bibinfo {author} {\bibfnamefont {D.~P.~D.}\ \bibnamefont
  {Vincenzo}}, \bibinfo {author} {\bibfnamefont {J.~A.}\ \bibnamefont
  {Smolin}}, \ and\ \bibinfo {author} {\bibfnamefont {W.}~\bibnamefont
  {K.Wootters}},\ }\href@noop {} {\bibfield  {journal} {\bibinfo  {journal}
  {Phys. Rev. A}\ }\textbf {\bibinfo {volume} {54}},\ \bibinfo {pages} {3824}
  (\bibinfo {year} {1996})}\BibitemShut {NoStop}%
\bibitem [{\citenamefont {Uhlmann}(1998)}]{Uhlmann1998}%
  \BibitemOpen
  \bibfield  {author} {\bibinfo {author} {\bibfnamefont {A.}~\bibnamefont
  {Uhlmann}},\ }\href@noop {} {\bibfield  {journal} {\bibinfo  {journal} {Open
  Syst. Inf. Dyn.}\ }\textbf {\bibinfo {volume} {5}},\ \bibinfo {pages} {209}
  (\bibinfo {year} {1998})}\BibitemShut {NoStop}%
\bibitem [{\citenamefont {De~Bi\`evre}\ \emph {et~al.}(2019)\citenamefont
  {De~Bi\`evre}, \citenamefont {Horoshko}, \citenamefont {Patera},\ and\
  \citenamefont {Kolobov}}]{Bievre2019}%
  \BibitemOpen
  \bibfield  {author} {\bibinfo {author} {\bibfnamefont {S.}~\bibnamefont
  {De~Bi\`evre}}, \bibinfo {author} {\bibfnamefont {D.~B.}\ \bibnamefont
  {Horoshko}}, \bibinfo {author} {\bibfnamefont {G.}~\bibnamefont {Patera}}, \
  and\ \bibinfo {author} {\bibfnamefont {M.~I.}\ \bibnamefont {Kolobov}},\
  }\href@noop {} {\bibfield  {journal} {\bibinfo  {journal} {Phys. Rev. Lett.}\
  }\textbf {\bibinfo {volume} {122}},\ \bibinfo {pages} {080402} (\bibinfo
  {year} {2019})}\BibitemShut {NoStop}%
\bibitem [{\citenamefont {R{\'e}nyi}(1960)}]{Renyi1960}%
  \BibitemOpen
  \bibfield  {author} {\bibinfo {author} {\bibfnamefont {A.}~\bibnamefont
  {R{\'e}nyi}},\ }\bibfield  {title} {\enquote {\bibinfo {title} {{Proceedings
  of the 4th Berkeley Symposium on Mathematics, Statistics and Probability}},}\
  }\href@noop {} {\ \textbf {\bibinfo {volume} {1}},\ \bibinfo {pages} {547}
  (\bibinfo {year} {1960})}\BibitemShut {NoStop}%
\bibitem [{\citenamefont {Kwon}\ \emph {et~al.}(2019)\citenamefont {Kwon},
  \citenamefont {Tan}, \citenamefont {Volkoff},\ and\ \citenamefont
  {Jeong}}]{Kwon2019}%
  \BibitemOpen
  \bibfield  {author} {\bibinfo {author} {\bibfnamefont {H.}~\bibnamefont
  {Kwon}}, \bibinfo {author} {\bibfnamefont {K.~C.}\ \bibnamefont {Tan}},
  \bibinfo {author} {\bibfnamefont {T.}~\bibnamefont {Volkoff}}, \ and\
  \bibinfo {author} {\bibfnamefont {H.}~\bibnamefont {Jeong}},\ }\href@noop {}
  {\bibfield  {journal} {\bibinfo  {journal} {Phys. Rev. Lett.}\ }\textbf
  {\bibinfo {volume} {122}},\ \bibinfo {pages} {040503} (\bibinfo {year}
  {2019})}\BibitemShut {NoStop}%
\bibitem [{\citenamefont {Adesso}, \citenamefont {Ragy},\ and\ \citenamefont
  {Lee}(2014)}]{Adesso2014}%
  \BibitemOpen
  \bibfield  {author} {\bibinfo {author} {\bibfnamefont {G.}~\bibnamefont
  {Adesso}}, \bibinfo {author} {\bibfnamefont {S.}~\bibnamefont {Ragy}}, \ and\
  \bibinfo {author} {\bibfnamefont {A.~R.}\ \bibnamefont {Lee}},\ }\href@noop
  {} {\bibfield  {journal} {\bibinfo  {journal} {Open Syst. Inf. Dyn.}\
  }\textbf {\bibinfo {volume} {21}},\ \bibinfo {pages} {1440001} (\bibinfo
  {year} {2014})}\BibitemShut {NoStop}%
\bibitem [{\citenamefont {Ma}\ and\ \citenamefont {Rhodes}(1990)}]{Ma1990}%
  \BibitemOpen
  \bibfield  {author} {\bibinfo {author} {\bibfnamefont {X.}~\bibnamefont
  {Ma}}\ and\ \bibinfo {author} {\bibfnamefont {W.}~\bibnamefont {Rhodes}},\
  }\href@noop {} {\bibfield  {journal} {\bibinfo  {journal} {Phys. Rev. A}\
  }\textbf {\bibinfo {volume} {41}},\ \bibinfo {pages} {4625} (\bibinfo {year}
  {1990})}\BibitemShut {NoStop}%
\bibitem [{\citenamefont {Cariolaro}\ and\ \citenamefont
  {Pierobon}(2016)}]{Cariolaro2016}%
  \BibitemOpen
  \bibfield  {author} {\bibinfo {author} {\bibfnamefont {G.}~\bibnamefont
  {Cariolaro}}\ and\ \bibinfo {author} {\bibfnamefont {G.}~\bibnamefont
  {Pierobon}},\ }\href@noop {} {\bibfield  {journal} {\bibinfo  {journal}
  {Phys. Rev. A}\ }\textbf {\bibinfo {volume} {94}},\ \bibinfo {pages} {062109}
  (\bibinfo {year} {2016})}\BibitemShut {NoStop}%
\bibitem [{\citenamefont {Weedbrook}\ \emph {et~al.}(2012)\citenamefont
  {Weedbrook}, \citenamefont {Pirandola}, \citenamefont {Garc\'{\i}a-Patr\'on},
  \citenamefont {Cerf}, \citenamefont {Ralph}, \citenamefont {Shapiro},\ and\
  \citenamefont {Lloyd}}]{Weedbrook2012}%
  \BibitemOpen
  \bibfield  {author} {\bibinfo {author} {\bibfnamefont {C.}~\bibnamefont
  {Weedbrook}}, \bibinfo {author} {\bibfnamefont {S.}~\bibnamefont
  {Pirandola}}, \bibinfo {author} {\bibfnamefont {R.}~\bibnamefont
  {Garc\'{\i}a-Patr\'on}}, \bibinfo {author} {\bibfnamefont {N.~J.}\
  \bibnamefont {Cerf}}, \bibinfo {author} {\bibfnamefont {T.~C.}\ \bibnamefont
  {Ralph}}, \bibinfo {author} {\bibfnamefont {J.~H.}\ \bibnamefont {Shapiro}},
  \ and\ \bibinfo {author} {\bibfnamefont {S.}~\bibnamefont {Lloyd}},\
  }\href@noop {} {\bibfield  {journal} {\bibinfo  {journal} {Rev. Mod. Phys.}\
  }\textbf {\bibinfo {volume} {84}},\ \bibinfo {pages} {621} (\bibinfo {year}
  {2012})}\BibitemShut {NoStop}%
\bibitem [{\citenamefont {Lloyd}\ and\ \citenamefont
  {Braunstein}(1999)}]{Lloyd1999}%
  \BibitemOpen
  \bibfield  {author} {\bibinfo {author} {\bibfnamefont {S.}~\bibnamefont
  {Lloyd}}\ and\ \bibinfo {author} {\bibfnamefont {S.~L.}\ \bibnamefont
  {Braunstein}},\ }\href@noop {} {\bibfield  {journal} {\bibinfo  {journal}
  {Phys. Rev. Lett.}\ }\textbf {\bibinfo {volume} {82}},\ \bibinfo {pages}
  {1784} (\bibinfo {year} {1999})}\BibitemShut {NoStop}%
\bibitem [{\citenamefont {Eisert}, \citenamefont {Scheel},\ and\ \citenamefont
  {Plenio}(2002)}]{Eisert2002}%
  \BibitemOpen
  \bibfield  {author} {\bibinfo {author} {\bibfnamefont {J.}~\bibnamefont
  {Eisert}}, \bibinfo {author} {\bibfnamefont {S.}~\bibnamefont {Scheel}}, \
  and\ \bibinfo {author} {\bibfnamefont {M.~B.}\ \bibnamefont {Plenio}},\
  }\href@noop {} {\bibfield  {journal} {\bibinfo  {journal} {Phys. Rev. Lett.}\
  }\textbf {\bibinfo {volume} {89}},\ \bibinfo {pages} {137903} (\bibinfo
  {year} {2002})}\BibitemShut {NoStop}%
\bibitem [{\citenamefont {Giedke}\ and\ \citenamefont
  {Cirac}(2002)}]{Giedke2002}%
  \BibitemOpen
  \bibfield  {author} {\bibinfo {author} {\bibfnamefont {G.}~\bibnamefont
  {Giedke}}\ and\ \bibinfo {author} {\bibfnamefont {J.~I.}\ \bibnamefont
  {Cirac}},\ }\href@noop {} {\bibfield  {journal} {\bibinfo  {journal} {Phys.
  Rev. A}\ }\textbf {\bibinfo {volume} {66}},\ \bibinfo {pages} {032316}
  (\bibinfo {year} {2002})}\BibitemShut {NoStop}%
\bibitem [{\citenamefont {Fiur\'a\ifmmode~\check{s}\else
  \v{s}\fi{}ek}(2002)}]{Fiurasek2002}%
  \BibitemOpen
  \bibfield  {author} {\bibinfo {author} {\bibfnamefont {J.}~\bibnamefont
  {Fiur\'a\ifmmode~\check{s}\else \v{s}\fi{}ek}},\ }\href@noop {} {\bibfield
  {journal} {\bibinfo  {journal} {Phys. Rev. Lett.}\ }\textbf {\bibinfo
  {volume} {89}},\ \bibinfo {pages} {137904} (\bibinfo {year}
  {2002})}\BibitemShut {NoStop}%
\bibitem [{\citenamefont {Bartlett}\ and\ \citenamefont
  {Sanders}(2002)}]{Bartlett2002}%
  \BibitemOpen
  \bibfield  {author} {\bibinfo {author} {\bibfnamefont {S.~D.}\ \bibnamefont
  {Bartlett}}\ and\ \bibinfo {author} {\bibfnamefont {B.~C.}\ \bibnamefont
  {Sanders}},\ }\href@noop {} {\bibfield  {journal} {\bibinfo  {journal} {Phys.
  Rev. A}\ }\textbf {\bibinfo {volume} {65}},\ \bibinfo {pages} {042304}
  (\bibinfo {year} {2002})}\BibitemShut {NoStop}%
\bibitem [{\citenamefont {Cerf}\ \emph {et~al.}(2005)\citenamefont {Cerf},
  \citenamefont {Kr\"uger}, \citenamefont {Navez}, \citenamefont {Werner},\
  and\ \citenamefont {Wolf}}]{Cerf2005}%
  \BibitemOpen
  \bibfield  {author} {\bibinfo {author} {\bibfnamefont {N.~J.}\ \bibnamefont
  {Cerf}}, \bibinfo {author} {\bibfnamefont {O.}~\bibnamefont {Kr\"uger}},
  \bibinfo {author} {\bibfnamefont {P.}~\bibnamefont {Navez}}, \bibinfo
  {author} {\bibfnamefont {R.~F.}\ \bibnamefont {Werner}}, \ and\ \bibinfo
  {author} {\bibfnamefont {M.~M.}\ \bibnamefont {Wolf}},\ }\href@noop {}
  {\bibfield  {journal} {\bibinfo  {journal} {Phys. Rev. Lett.}\ }\textbf
  {\bibinfo {volume} {95}},\ \bibinfo {pages} {070501} (\bibinfo {year}
  {2005})}\BibitemShut {NoStop}%
\bibitem [{\citenamefont {Menicucci}\ \emph {et~al.}(2006)\citenamefont
  {Menicucci}, \citenamefont {van Loock}, \citenamefont {Gu}, \citenamefont
  {Weedbrook}, \citenamefont {Ralph},\ and\ \citenamefont
  {Nielsen}}]{Menicucci2006}%
  \BibitemOpen
  \bibfield  {author} {\bibinfo {author} {\bibfnamefont {N.~C.}\ \bibnamefont
  {Menicucci}}, \bibinfo {author} {\bibfnamefont {P.}~\bibnamefont {van
  Loock}}, \bibinfo {author} {\bibfnamefont {M.}~\bibnamefont {Gu}}, \bibinfo
  {author} {\bibfnamefont {C.}~\bibnamefont {Weedbrook}}, \bibinfo {author}
  {\bibfnamefont {T.~C.}\ \bibnamefont {Ralph}}, \ and\ \bibinfo {author}
  {\bibfnamefont {M.~A.}\ \bibnamefont {Nielsen}},\ }\href@noop {} {\bibfield
  {journal} {\bibinfo  {journal} {Phys. Rev. Lett.}\ }\textbf {\bibinfo
  {volume} {97}},\ \bibinfo {pages} {110501} (\bibinfo {year}
  {2006})}\BibitemShut {NoStop}%
\bibitem [{\citenamefont {Niset}, \citenamefont {Fiur\'a\ifmmode~\check{s}\else
  \v{s}\fi{}ek},\ and\ \citenamefont {Cerf}(2009)}]{Niset2009}%
  \BibitemOpen
  \bibfield  {author} {\bibinfo {author} {\bibfnamefont {J.}~\bibnamefont
  {Niset}}, \bibinfo {author} {\bibfnamefont {J.}~\bibnamefont
  {Fiur\'a\ifmmode~\check{s}\else \v{s}\fi{}ek}}, \ and\ \bibinfo {author}
  {\bibfnamefont {N.~J.}\ \bibnamefont {Cerf}},\ }\href@noop {} {\bibfield
  {journal} {\bibinfo  {journal} {Phys. Rev. Lett.}\ }\textbf {\bibinfo
  {volume} {102}},\ \bibinfo {pages} {120501} (\bibinfo {year}
  {2009})}\BibitemShut {NoStop}%
\bibitem [{\citenamefont {Zhang}\ and\ \citenamefont {van
  Loock}(2010)}]{Zhang2010}%
  \BibitemOpen
  \bibfield  {author} {\bibinfo {author} {\bibfnamefont {S.~L.}\ \bibnamefont
  {Zhang}}\ and\ \bibinfo {author} {\bibfnamefont {P.}~\bibnamefont {van
  Loock}},\ }\href@noop {} {\bibfield  {journal} {\bibinfo  {journal} {Phys.
  Rev. A}\ }\textbf {\bibinfo {volume} {82}},\ \bibinfo {pages} {062316}
  (\bibinfo {year} {2010})}\BibitemShut {NoStop}%
\bibitem [{\citenamefont {Ohliger}, \citenamefont {Kieling},\ and\
  \citenamefont {Eisert}(2010)}]{Ohliger2010}%
  \BibitemOpen
  \bibfield  {author} {\bibinfo {author} {\bibfnamefont {M.}~\bibnamefont
  {Ohliger}}, \bibinfo {author} {\bibfnamefont {K.}~\bibnamefont {Kieling}}, \
  and\ \bibinfo {author} {\bibfnamefont {J.}~\bibnamefont {Eisert}},\
  }\href@noop {} {\bibfield  {journal} {\bibinfo  {journal} {Phys. Rev. A}\
  }\textbf {\bibinfo {volume} {82}},\ \bibinfo {pages} {042336} (\bibinfo
  {year} {2010})}\BibitemShut {NoStop}%
\bibitem [{\citenamefont {Genoni}, \citenamefont {Paris},\ and\ \citenamefont
  {Banaszek}(2007)}]{Genoni2007}%
  \BibitemOpen
  \bibfield  {author} {\bibinfo {author} {\bibfnamefont {M.~G.}\ \bibnamefont
  {Genoni}}, \bibinfo {author} {\bibfnamefont {M.~G.~A.}\ \bibnamefont
  {Paris}}, \ and\ \bibinfo {author} {\bibfnamefont {K.}~\bibnamefont
  {Banaszek}},\ }\href@noop {} {\bibfield  {journal} {\bibinfo  {journal}
  {Phys. Rev. A}\ }\textbf {\bibinfo {volume} {76}},\ \bibinfo {pages} {042327}
  (\bibinfo {year} {2007})}\BibitemShut {NoStop}%
\bibitem [{\citenamefont {Genoni}, \citenamefont {Paris},\ and\ \citenamefont
  {Banaszek}(2008)}]{Genoni2008}%
  \BibitemOpen
  \bibfield  {author} {\bibinfo {author} {\bibfnamefont {M.~G.}\ \bibnamefont
  {Genoni}}, \bibinfo {author} {\bibfnamefont {M.~G.~A.}\ \bibnamefont
  {Paris}}, \ and\ \bibinfo {author} {\bibfnamefont {K.}~\bibnamefont
  {Banaszek}},\ }\href@noop {} {\bibfield  {journal} {\bibinfo  {journal}
  {Phys. Rev. A}\ }\textbf {\bibinfo {volume} {78}},\ \bibinfo {pages} {060303}
  (\bibinfo {year} {2008})}\BibitemShut {NoStop}%
\bibitem [{\citenamefont {Genoni}\ and\ \citenamefont
  {Paris}(2010)}]{Genoni2010}%
  \BibitemOpen
  \bibfield  {author} {\bibinfo {author} {\bibfnamefont {M.~G.}\ \bibnamefont
  {Genoni}}\ and\ \bibinfo {author} {\bibfnamefont {M.~G.~A.}\ \bibnamefont
  {Paris}},\ }\href@noop {} {\bibfield  {journal} {\bibinfo  {journal} {Phys.
  Rev. A}\ }\textbf {\bibinfo {volume} {82}},\ \bibinfo {pages} {052341}
  (\bibinfo {year} {2010})}\BibitemShut {NoStop}%
\bibitem [{\citenamefont {Ivan}, \citenamefont {Kumar},\ and\ \citenamefont
  {Simon}(2012)}]{Ivan2012}%
  \BibitemOpen
  \bibfield  {author} {\bibinfo {author} {\bibfnamefont {J.~S.}\ \bibnamefont
  {Ivan}}, \bibinfo {author} {\bibfnamefont {M.~S.}\ \bibnamefont {Kumar}}, \
  and\ \bibinfo {author} {\bibfnamefont {R.}~\bibnamefont {Simon}},\
  }\href@noop {} {\bibfield  {journal} {\bibinfo  {journal} {Quantum Inf.
  Process.}\ }\textbf {\bibinfo {volume} {11}},\ \bibinfo {pages} {853}
  (\bibinfo {year} {2012})}\BibitemShut {NoStop}%
\bibitem [{\citenamefont {Marian}\ and\ \citenamefont
  {Marian}(2013)}]{Marian2013}%
  \BibitemOpen
  \bibfield  {author} {\bibinfo {author} {\bibfnamefont {P.}~\bibnamefont
  {Marian}}\ and\ \bibinfo {author} {\bibfnamefont {T.~A.}\ \bibnamefont
  {Marian}},\ }\href@noop {} {\bibfield  {journal} {\bibinfo  {journal} {Phys.
  Rev. A}\ }\textbf {\bibinfo {volume} {88}},\ \bibinfo {pages} {012322}
  (\bibinfo {year} {2013})}\BibitemShut {NoStop}%
\bibitem [{\citenamefont {Ghiu}, \citenamefont {Marian},\ and\ \citenamefont
  {Marian}(2013)}]{Ghiu2013}%
  \BibitemOpen
  \bibfield  {author} {\bibinfo {author} {\bibfnamefont {I.}~\bibnamefont
  {Ghiu}}, \bibinfo {author} {\bibfnamefont {P.}~\bibnamefont {Marian}}, \ and\
  \bibinfo {author} {\bibfnamefont {T.~A.}\ \bibnamefont {Marian}},\
  }\href@noop {} {\bibfield  {journal} {\bibinfo  {journal} {Phys. Scr.}\
  }\textbf {\bibinfo {volume} {T153}},\ \bibinfo {pages} {014028} (\bibinfo
  {year} {2013})}\BibitemShut {NoStop}%
\bibitem [{\citenamefont {Park}\ \emph {et~al.}(2017)\citenamefont {Park},
  \citenamefont {Lee}, \citenamefont {Ji},\ and\ \citenamefont
  {Nha}}]{Park2017}%
  \BibitemOpen
  \bibfield  {author} {\bibinfo {author} {\bibfnamefont {J.}~\bibnamefont
  {Park}}, \bibinfo {author} {\bibfnamefont {J.}~\bibnamefont {Lee}}, \bibinfo
  {author} {\bibfnamefont {S.-W.}\ \bibnamefont {Ji}}, \ and\ \bibinfo {author}
  {\bibfnamefont {H.}~\bibnamefont {Nha}},\ }\href@noop {} {\bibfield
  {journal} {\bibinfo  {journal} {Phys. Rev. A}\ }\textbf {\bibinfo {volume}
  {96}},\ \bibinfo {pages} {052324} (\bibinfo {year} {2017})}\BibitemShut
  {NoStop}%
\bibitem [{\citenamefont {Albarelli}\ \emph {et~al.}(2018)\citenamefont
  {Albarelli}, \citenamefont {Genoni}, \citenamefont {Paris},\ and\
  \citenamefont {Ferraro}}]{Albarelli2018}%
  \BibitemOpen
  \bibfield  {author} {\bibinfo {author} {\bibfnamefont {F.}~\bibnamefont
  {Albarelli}}, \bibinfo {author} {\bibfnamefont {M.~G.}\ \bibnamefont
  {Genoni}}, \bibinfo {author} {\bibfnamefont {M.~G.~A.}\ \bibnamefont
  {Paris}}, \ and\ \bibinfo {author} {\bibfnamefont {A.}~\bibnamefont
  {Ferraro}},\ }\href@noop {} {\bibfield  {journal} {\bibinfo  {journal} {Phys.
  Rev. A}\ }\textbf {\bibinfo {volume} {98}},\ \bibinfo {pages} {052350}
  (\bibinfo {year} {2018})}\BibitemShut {NoStop}%
\bibitem [{\citenamefont {Lami}\ \emph {et~al.}(2018)\citenamefont {Lami},
  \citenamefont {Regula}, \citenamefont {Wang}, \citenamefont {Nichols},
  \citenamefont {Winter},\ and\ \citenamefont {Adesso}}]{Lami2018}%
  \BibitemOpen
  \bibfield  {author} {\bibinfo {author} {\bibfnamefont {L.}~\bibnamefont
  {Lami}}, \bibinfo {author} {\bibfnamefont {B.}~\bibnamefont {Regula}},
  \bibinfo {author} {\bibfnamefont {X.}~\bibnamefont {Wang}}, \bibinfo {author}
  {\bibfnamefont {R.}~\bibnamefont {Nichols}}, \bibinfo {author} {\bibfnamefont
  {A.}~\bibnamefont {Winter}}, \ and\ \bibinfo {author} {\bibfnamefont
  {G.}~\bibnamefont {Adesso}},\ }\href@noop {} {\bibfield  {journal} {\bibinfo
  {journal} {Phys. Rev. A}\ }\textbf {\bibinfo {volume} {98}},\ \bibinfo
  {pages} {022335} (\bibinfo {year} {2018})}\BibitemShut {NoStop}%
\bibitem [{\citenamefont {Takagi}\ and\ \citenamefont
  {Zhuang}(2018)}]{Takagi2018}%
  \BibitemOpen
  \bibfield  {author} {\bibinfo {author} {\bibfnamefont {R.}~\bibnamefont
  {Takagi}}\ and\ \bibinfo {author} {\bibfnamefont {Q.}~\bibnamefont
  {Zhuang}},\ }\href@noop {} {\bibfield  {journal} {\bibinfo  {journal} {Phys.
  Rev. A}\ }\textbf {\bibinfo {volume} {97}},\ \bibinfo {pages} {062337}
  (\bibinfo {year} {2018})}\BibitemShut {NoStop}%
\bibitem [{\citenamefont {Zhuang}, \citenamefont {Shor},\ and\ \citenamefont
  {Shapiro}(2018)}]{Zhuang2018}%
  \BibitemOpen
  \bibfield  {author} {\bibinfo {author} {\bibfnamefont {Q.}~\bibnamefont
  {Zhuang}}, \bibinfo {author} {\bibfnamefont {P.~W.}\ \bibnamefont {Shor}}, \
  and\ \bibinfo {author} {\bibfnamefont {J.~H.}\ \bibnamefont {Shapiro}},\
  }\href@noop {} {\bibfield  {journal} {\bibinfo  {journal} {Phys. Rev. A}\
  }\textbf {\bibinfo {volume} {97}},\ \bibinfo {pages} {052317} (\bibinfo
  {year} {2018})}\BibitemShut {NoStop}%
\bibitem [{\citenamefont {Park}\ \emph {et~al.}(2019)\citenamefont {Park},
  \citenamefont {Lee}, \citenamefont {Baek}, \citenamefont {Ji},\ and\
  \citenamefont {Nha}}]{Park2019}%
  \BibitemOpen
  \bibfield  {author} {\bibinfo {author} {\bibfnamefont {J.}~\bibnamefont
  {Park}}, \bibinfo {author} {\bibfnamefont {J.}~\bibnamefont {Lee}}, \bibinfo
  {author} {\bibfnamefont {K.}~\bibnamefont {Baek}}, \bibinfo {author}
  {\bibfnamefont {S.-W.}\ \bibnamefont {Ji}}, \ and\ \bibinfo {author}
  {\bibfnamefont {H.}~\bibnamefont {Nha}},\ }\href@noop {} {\bibfield
  {journal} {\bibinfo  {journal} {Phys. Rev. A}\ }\textbf {\bibinfo {volume}
  {100}},\ \bibinfo {pages} {012333} (\bibinfo {year} {2019})}\BibitemShut
  {NoStop}%
\bibitem [{\citenamefont {Chitambar}\ and\ \citenamefont
  {Gour}(2019)}]{Chitambar2019}%
  \BibitemOpen
  \bibfield  {author} {\bibinfo {author} {\bibfnamefont {E.}~\bibnamefont
  {Chitambar}}\ and\ \bibinfo {author} {\bibfnamefont {G.}~\bibnamefont
  {Gour}},\ }\href@noop {} {\bibfield  {journal} {\bibinfo  {journal} {Rev.
  Mod. Phys.}\ }\textbf {\bibinfo {volume} {91}},\ \bibinfo {pages} {025001}
  (\bibinfo {year} {2019})}\BibitemShut {NoStop}%
\bibitem [{\citenamefont {Streltsov}, \citenamefont {Adesso},\ and\
  \citenamefont {Plenio}(2017)}]{Streltsov2017}%
  \BibitemOpen
  \bibfield  {author} {\bibinfo {author} {\bibfnamefont {A.}~\bibnamefont
  {Streltsov}}, \bibinfo {author} {\bibfnamefont {G.}~\bibnamefont {Adesso}}, \
  and\ \bibinfo {author} {\bibfnamefont {M.~B.}\ \bibnamefont {Plenio}},\
  }\href@noop {} {\bibfield  {journal} {\bibinfo  {journal} {Rev. Mod. Phys.}\
  }\textbf {\bibinfo {volume} {89}},\ \bibinfo {pages} {041003} (\bibinfo
  {year} {2017})}\BibitemShut {NoStop}%
\bibitem [{\citenamefont {Sperling}\ and\ \citenamefont
  {Vogel}(2015)}]{Sperling2015}%
  \BibitemOpen
  \bibfield  {author} {\bibinfo {author} {\bibfnamefont {J.}~\bibnamefont
  {Sperling}}\ and\ \bibinfo {author} {\bibfnamefont {W.}~\bibnamefont
  {Vogel}},\ }\href@noop {} {\bibfield  {journal} {\bibinfo  {journal} {Phys.
  Scr.}\ }\textbf {\bibinfo {volume} {90}},\ \bibinfo {pages} {074024}
  (\bibinfo {year} {2015})}\BibitemShut {NoStop}%
\bibitem [{\citenamefont {Gilchrist}, \citenamefont {Langford},\ and\
  \citenamefont {Nielsen}(2005)}]{Gilchrist2005}%
  \BibitemOpen
  \bibfield  {author} {\bibinfo {author} {\bibfnamefont {A.}~\bibnamefont
  {Gilchrist}}, \bibinfo {author} {\bibfnamefont {N.~K.}\ \bibnamefont
  {Langford}}, \ and\ \bibinfo {author} {\bibfnamefont {M.~A.}\ \bibnamefont
  {Nielsen}},\ }\href@noop {} {\bibfield  {journal} {\bibinfo  {journal} {Phys.
  Rev. A}\ }\textbf {\bibinfo {volume} {71}},\ \bibinfo {pages} {062310}
  (\bibinfo {year} {2005})}\BibitemShut {NoStop}%
\bibitem [{\citenamefont {Tan}\ \emph {et~al.}(2017)\citenamefont {Tan},
  \citenamefont {Volkoff}, \citenamefont {Kwon},\ and\ \citenamefont
  {Jeong}}]{Tan2017}%
  \BibitemOpen
  \bibfield  {author} {\bibinfo {author} {\bibfnamefont {K.~C.}\ \bibnamefont
  {Tan}}, \bibinfo {author} {\bibfnamefont {T.}~\bibnamefont {Volkoff}},
  \bibinfo {author} {\bibfnamefont {H.}~\bibnamefont {Kwon}}, \ and\ \bibinfo
  {author} {\bibfnamefont {H.}~\bibnamefont {Jeong}},\ }\href@noop {}
  {\bibfield  {journal} {\bibinfo  {journal} {Phys. Rev. Lett.}\ }\textbf
  {\bibinfo {volume} {119}},\ \bibinfo {pages} {190405} (\bibinfo {year}
  {2017})}\BibitemShut {NoStop}%
\bibitem [{\citenamefont {Streltsov}\ \emph {et~al.}(2015)\citenamefont
  {Streltsov}, \citenamefont {Singh}, \citenamefont {Dhar}, \citenamefont
  {Bera},\ and\ \citenamefont {Adesso}}]{Streltsov2015}%
  \BibitemOpen
  \bibfield  {author} {\bibinfo {author} {\bibfnamefont {A.}~\bibnamefont
  {Streltsov}}, \bibinfo {author} {\bibfnamefont {U.}~\bibnamefont {Singh}},
  \bibinfo {author} {\bibfnamefont {H.~S.}\ \bibnamefont {Dhar}}, \bibinfo
  {author} {\bibfnamefont {M.~N.}\ \bibnamefont {Bera}}, \ and\ \bibinfo
  {author} {\bibfnamefont {G.}~\bibnamefont {Adesso}},\ }\href@noop {}
  {\bibfield  {journal} {\bibinfo  {journal} {Phys. Rev. Lett.}\ }\textbf
  {\bibinfo {volume} {115}},\ \bibinfo {pages} {020403} (\bibinfo {year}
  {2015})}\BibitemShut {NoStop}%
\bibitem [{\citenamefont {Tan}\ \emph {et~al.}(2018)\citenamefont {Tan},
  \citenamefont {Choi}, \citenamefont {Kwon},\ and\ \citenamefont
  {Jeong}}]{Tan2018a}%
  \BibitemOpen
  \bibfield  {author} {\bibinfo {author} {\bibfnamefont {K.~C.}\ \bibnamefont
  {Tan}}, \bibinfo {author} {\bibfnamefont {S.}~\bibnamefont {Choi}}, \bibinfo
  {author} {\bibfnamefont {H.}~\bibnamefont {Kwon}}, \ and\ \bibinfo {author}
  {\bibfnamefont {H.}~\bibnamefont {Jeong}},\ }\href@noop {} {\bibfield
  {journal} {\bibinfo  {journal} {Phys. Rev. A}\ }\textbf {\bibinfo {volume}
  {97}},\ \bibinfo {pages} {052304} (\bibinfo {year} {2018})}\BibitemShut
  {NoStop}%
\bibitem [{\citenamefont {Tan}\ \emph {et~al.}(2016)\citenamefont {Tan},
  \citenamefont {Kwon}, \citenamefont {Park},\ and\ \citenamefont
  {Jeong}}]{Tan2016}%
  \BibitemOpen
  \bibfield  {author} {\bibinfo {author} {\bibfnamefont {K.~C.}\ \bibnamefont
  {Tan}}, \bibinfo {author} {\bibfnamefont {H.}~\bibnamefont {Kwon}}, \bibinfo
  {author} {\bibfnamefont {C.-Y.}\ \bibnamefont {Park}}, \ and\ \bibinfo
  {author} {\bibfnamefont {H.}~\bibnamefont {Jeong}},\ }\href@noop {}
  {\bibfield  {journal} {\bibinfo  {journal} {Phys. Rev. A}\ }\textbf {\bibinfo
  {volume} {94}},\ \bibinfo {pages} {022329} (\bibinfo {year}
  {2016})}\BibitemShut {NoStop}%
\bibitem [{\citenamefont {Tan}\ and\ \citenamefont {Jeong}(2018)}]{Tan2018}%
  \BibitemOpen
  \bibfield  {author} {\bibinfo {author} {\bibfnamefont {K.~C.}\ \bibnamefont
  {Tan}}\ and\ \bibinfo {author} {\bibfnamefont {H.}~\bibnamefont {Jeong}},\
  }\href@noop {} {\bibfield  {journal} {\bibinfo  {journal} {Phys. Rev. Lett.}\
  }\textbf {\bibinfo {volume} {121}},\ \bibinfo {pages} {220401} (\bibinfo
  {year} {2018})}\BibitemShut {NoStop}%
\bibitem [{\citenamefont {Yadin}\ \emph {et~al.}(2018)\citenamefont {Yadin},
  \citenamefont {Binder}, \citenamefont {Thompson}, \citenamefont
  {Narasimhachar}, \citenamefont {Gu},\ and\ \citenamefont {Kim}}]{Yadin2018}%
  \BibitemOpen
  \bibfield  {author} {\bibinfo {author} {\bibfnamefont {B.}~\bibnamefont
  {Yadin}}, \bibinfo {author} {\bibfnamefont {F.~C.}\ \bibnamefont {Binder}},
  \bibinfo {author} {\bibfnamefont {J.}~\bibnamefont {Thompson}}, \bibinfo
  {author} {\bibfnamefont {V.}~\bibnamefont {Narasimhachar}}, \bibinfo {author}
  {\bibfnamefont {M.}~\bibnamefont {Gu}}, \ and\ \bibinfo {author}
  {\bibfnamefont {M.~S.}\ \bibnamefont {Kim}},\ }\href@noop {} {\bibfield
  {journal} {\bibinfo  {journal} {Phys. Rev. X}\ }\textbf {\bibinfo {volume}
  {8}},\ \bibinfo {pages} {041038} (\bibinfo {year} {2018})}\BibitemShut
  {NoStop}%
\bibitem [{\citenamefont {Kay}(1993)}]{Kay1993}%
  \BibitemOpen
  \bibfield  {author} {\bibinfo {author} {\bibfnamefont {S.}~\bibnamefont
  {Kay}},\ }\href@noop {} {\emph {\bibinfo {title} {Fundamentals of Statistical
  Signal Processing, Volume I: Estimation Theory}}},\ \bibinfo {edition} {1st}\
  ed.\ (\bibinfo  {publisher} {Prentice Hall},\ \bibinfo {address} {Upper
  Saddle River},\ \bibinfo {year} {1993})\BibitemShut {NoStop}%
\bibitem [{\citenamefont {Lehmann}\ and\ \citenamefont
  {Casella}(1998)}]{Lehmann1998}%
  \BibitemOpen
  \bibfield  {author} {\bibinfo {author} {\bibfnamefont {E.~L.}\ \bibnamefont
  {Lehmann}}\ and\ \bibinfo {author} {\bibfnamefont {G.}~\bibnamefont
  {Casella}},\ }\href@noop {} {\emph {\bibinfo {title} {Theory of point
  estimation, Volume 31}}}\ (\bibinfo  {publisher} {Springer},\ \bibinfo
  {address} {New York},\ \bibinfo {year} {1998})\BibitemShut {NoStop}%
\bibitem [{\citenamefont {Helstrom}(1967)}]{Helstrom1967}%
  \BibitemOpen
  \bibfield  {author} {\bibinfo {author} {\bibfnamefont {C.}~\bibnamefont
  {Helstrom}},\ }\href@noop {} {\bibfield  {journal} {\bibinfo  {journal}
  {Phys. Lett. A}\ }\textbf {\bibinfo {volume} {25}},\ \bibinfo {pages} {101}
  (\bibinfo {year} {1967})}\BibitemShut {NoStop}%
\bibitem [{\citenamefont {Helstrom}(1968)}]{Helstrom1968}%
  \BibitemOpen
  \bibfield  {author} {\bibinfo {author} {\bibfnamefont {C.}~\bibnamefont
  {Helstrom}},\ }\href@noop {} {\bibfield  {journal} {\bibinfo  {journal} {IEEE
  Trans. Inf. Theory}\ }\textbf {\bibinfo {volume} {14}},\ \bibinfo {pages}
  {234} (\bibinfo {year} {1968})}\BibitemShut {NoStop}%
\bibitem [{\citenamefont {Helstrom}(1976)}]{Helstrom1976}%
  \BibitemOpen
  \bibfield  {author} {\bibinfo {author} {\bibfnamefont {C.~W.}\ \bibnamefont
  {Helstrom}},\ }\href@noop {} {\emph {\bibinfo {title} {Quantum Detection and
  Estimation Theory.}}}\ (\bibinfo  {publisher} {Elsevier},\ \bibinfo {address}
  {Burlington},\ \bibinfo {year} {1976})\BibitemShut {NoStop}%
\bibitem [{\citenamefont {Holevo}(1982)}]{Holevo1982}%
  \BibitemOpen
  \bibfield  {author} {\bibinfo {author} {\bibfnamefont {A.~S.}\ \bibnamefont
  {Holevo}},\ }\href@noop {} {\emph {\bibinfo {title} {Probabilistic and
  Statistical Aspects of Quantum Theory}}}\ (\bibinfo  {publisher}
  {North-Holland},\ \bibinfo {address} {Amsterdam},\ \bibinfo {year}
  {1982})\BibitemShut {NoStop}%
\bibitem [{\citenamefont {Braunstein}\ and\ \citenamefont
  {Caves}(1994)}]{Braunstein1994}%
  \BibitemOpen
  \bibfield  {author} {\bibinfo {author} {\bibfnamefont {S.~L.}\ \bibnamefont
  {Braunstein}}\ and\ \bibinfo {author} {\bibfnamefont {C.~M.}\ \bibnamefont
  {Caves}},\ }\href@noop {} {\bibfield  {journal} {\bibinfo  {journal} {Phys.
  Rev. Lett.}\ }\textbf {\bibinfo {volume} {72}},\ \bibinfo {pages} {3439}
  (\bibinfo {year} {1994})}\BibitemShut {NoStop}%
\bibitem [{\citenamefont {Cochran}(1973)}]{Cochran1973}%
  \BibitemOpen
  \bibfield  {author} {\bibinfo {author} {\bibfnamefont {W.~G.}\ \bibnamefont
  {Cochran}},\ }\href@noop {} {\bibfield  {journal} {\bibinfo  {journal} {J.
  Am. Stat. Assoc.}\ }\textbf {\bibinfo {volume} {68}},\ \bibinfo {pages} {771}
  (\bibinfo {year} {1973})}\BibitemShut {NoStop}%
\bibitem [{\citenamefont {Barndorff-Nielsen}\ and\ \citenamefont
  {Gill}(2000)}]{Barndorff2000}%
  \BibitemOpen
  \bibfield  {author} {\bibinfo {author} {\bibfnamefont {O.~E.}\ \bibnamefont
  {Barndorff-Nielsen}}\ and\ \bibinfo {author} {\bibfnamefont {R.~D.}\
  \bibnamefont {Gill}},\ }\href@noop {} {\bibfield  {journal} {\bibinfo
  {journal} {J. Phys. A: Math. Gen.}\ }\textbf {\bibinfo {volume} {33}},\
  \bibinfo {pages} {4481} (\bibinfo {year} {2000})}\BibitemShut {NoStop}%
\bibitem [{\citenamefont {Wiseman}(1995)}]{Wiseman1995}%
  \BibitemOpen
  \bibfield  {author} {\bibinfo {author} {\bibfnamefont {H.~M.}\ \bibnamefont
  {Wiseman}},\ }\href@noop {} {\bibfield  {journal} {\bibinfo  {journal} {Phys.
  Rev. Lett.}\ }\textbf {\bibinfo {volume} {75}},\ \bibinfo {pages} {4587}
  (\bibinfo {year} {1995})}\BibitemShut {NoStop}%
\bibitem [{\citenamefont {Berry}\ and\ \citenamefont
  {Wiseman}(2000)}]{Berry2000}%
  \BibitemOpen
  \bibfield  {author} {\bibinfo {author} {\bibfnamefont {D.~W.}\ \bibnamefont
  {Berry}}\ and\ \bibinfo {author} {\bibfnamefont {H.~M.}\ \bibnamefont
  {Wiseman}},\ }\href@noop {} {\bibfield  {journal} {\bibinfo  {journal} {Phys.
  Rev. Lett.}\ }\textbf {\bibinfo {volume} {85}},\ \bibinfo {pages} {5098}
  (\bibinfo {year} {2000})}\BibitemShut {NoStop}%
\bibitem [{\citenamefont {Berry}\ and\ \citenamefont
  {Wiseman}(2002)}]{Berry2002}%
  \BibitemOpen
  \bibfield  {author} {\bibinfo {author} {\bibfnamefont {D.~W.}\ \bibnamefont
  {Berry}}\ and\ \bibinfo {author} {\bibfnamefont {H.~M.}\ \bibnamefont
  {Wiseman}},\ }\href@noop {} {\bibfield  {journal} {\bibinfo  {journal} {Phys.
  Rev. A}\ }\textbf {\bibinfo {volume} {65}},\ \bibinfo {pages} {043803}
  (\bibinfo {year} {2002})}\BibitemShut {NoStop}%
\bibitem [{\citenamefont {Armen}\ \emph {et~al.}(2002)\citenamefont {Armen},
  \citenamefont {Au}, \citenamefont {Stockton}, \citenamefont {Doherty},\ and\
  \citenamefont {Mabuchi}}]{Armen2002}%
  \BibitemOpen
  \bibfield  {author} {\bibinfo {author} {\bibfnamefont {M.~A.}\ \bibnamefont
  {Armen}}, \bibinfo {author} {\bibfnamefont {J.~K.}\ \bibnamefont {Au}},
  \bibinfo {author} {\bibfnamefont {J.~K.}\ \bibnamefont {Stockton}}, \bibinfo
  {author} {\bibfnamefont {A.~C.}\ \bibnamefont {Doherty}}, \ and\ \bibinfo
  {author} {\bibfnamefont {H.}~\bibnamefont {Mabuchi}},\ }\href@noop {}
  {\bibfield  {journal} {\bibinfo  {journal} {Phys. Rev. Lett.}\ }\textbf
  {\bibinfo {volume} {89}},\ \bibinfo {pages} {133602} (\bibinfo {year}
  {2002})}\BibitemShut {NoStop}%
\bibitem [{\citenamefont {Fujiwara}(2006)}]{Fujiwara2006}%
  \BibitemOpen
  \bibfield  {author} {\bibinfo {author} {\bibfnamefont {A.}~\bibnamefont
  {Fujiwara}},\ }\href@noop {} {\bibfield  {journal} {\bibinfo  {journal} {J.
  Phys. A: Math. Gen.}\ }\textbf {\bibinfo {volume} {39}},\ \bibinfo {pages}
  {12489} (\bibinfo {year} {2006})}\BibitemShut {NoStop}%
\bibitem [{\citenamefont {Fujiwara}(2011)}]{Fujiwara2011}%
  \BibitemOpen
  \bibfield  {author} {\bibinfo {author} {\bibfnamefont {A.}~\bibnamefont
  {Fujiwara}},\ }\href@noop {} {\bibfield  {journal} {\bibinfo  {journal} {J.
  Phys. A: Math. Theor.}\ }\textbf {\bibinfo {volume} {44}},\ \bibinfo {pages}
  {079501} (\bibinfo {year} {2011})}\BibitemShut {NoStop}%
\bibitem [{\citenamefont {Braunstein}, \citenamefont {Caves},\ and\
  \citenamefont {Milburn}(1996)}]{Braunstein1996}%
  \BibitemOpen
  \bibfield  {author} {\bibinfo {author} {\bibfnamefont {S.~L.}\ \bibnamefont
  {Braunstein}}, \bibinfo {author} {\bibfnamefont {C.~M.}\ \bibnamefont
  {Caves}}, \ and\ \bibinfo {author} {\bibfnamefont {G.}~\bibnamefont
  {Milburn}},\ }\href@noop {} {\bibfield  {journal} {\bibinfo  {journal} {Ann.
  Phys.}\ }\textbf {\bibinfo {volume} {247}},\ \bibinfo {pages} {135} (\bibinfo
  {year} {1996})}\BibitemShut {NoStop}%
\bibitem [{\citenamefont {T{\'{o}}th}\ and\ \citenamefont
  {Apellaniz}(2014)}]{Toth2014}%
  \BibitemOpen
  \bibfield  {author} {\bibinfo {author} {\bibfnamefont {G.}~\bibnamefont
  {T{\'{o}}th}}\ and\ \bibinfo {author} {\bibfnamefont {I.}~\bibnamefont
  {Apellaniz}},\ }\href@noop {} {\bibfield  {journal} {\bibinfo  {journal} {J.
  Phys. A: Math. Theor.}\ }\textbf {\bibinfo {volume} {47}},\ \bibinfo {pages}
  {424006} (\bibinfo {year} {2014})}\BibitemShut {NoStop}%
\bibitem [{\citenamefont {Giovannetti}, \citenamefont {Lloyd},\ and\
  \citenamefont {Maccone}(2006)}]{Giovannetti2006}%
  \BibitemOpen
  \bibfield  {author} {\bibinfo {author} {\bibfnamefont {V.}~\bibnamefont
  {Giovannetti}}, \bibinfo {author} {\bibfnamefont {S.}~\bibnamefont {Lloyd}},
  \ and\ \bibinfo {author} {\bibfnamefont {L.}~\bibnamefont {Maccone}},\
  }\href@noop {} {\bibfield  {journal} {\bibinfo  {journal} {Phys. Rev. Lett.}\
  }\textbf {\bibinfo {volume} {96}},\ \bibinfo {pages} {010401} (\bibinfo
  {year} {2006})}\BibitemShut {NoStop}%
\bibitem [{\citenamefont {Yu}(2013)}]{Yu2013}%
  \BibitemOpen
  \bibfield  {author} {\bibinfo {author} {\bibfnamefont {S.}~\bibnamefont
  {Yu}},\ }\href@noop {} {} (\bibinfo {year} {2013}),\ \Eprint
  {http://arxiv.org/abs/arXiv:1302.5311} {arXiv:1302.5311} \BibitemShut
  {NoStop}%
\bibitem [{\citenamefont {T\'oth}\ and\ \citenamefont {Petz}(2013)}]{Toth2013}%
  \BibitemOpen
  \bibfield  {author} {\bibinfo {author} {\bibfnamefont {G.}~\bibnamefont
  {T\'oth}}\ and\ \bibinfo {author} {\bibfnamefont {D.}~\bibnamefont {Petz}},\
  }\href@noop {} {\bibfield  {journal} {\bibinfo  {journal} {Phys. Rev. A}\
  }\textbf {\bibinfo {volume} {87}},\ \bibinfo {pages} {032324} (\bibinfo
  {year} {2013})}\BibitemShut {NoStop}%
\bibitem [{\citenamefont {Rivas}\ and\ \citenamefont {Luis}(2010)}]{Rivas2010}%
  \BibitemOpen
  \bibfield  {author} {\bibinfo {author} {\bibfnamefont {A.}~\bibnamefont
  {Rivas}}\ and\ \bibinfo {author} {\bibfnamefont {A.}~\bibnamefont {Luis}},\
  }\href@noop {} {\bibfield  {journal} {\bibinfo  {journal} {Phys. Rev. Lett.}\
  }\textbf {\bibinfo {volume} {105}},\ \bibinfo {pages} {010403} (\bibinfo
  {year} {2010})}\BibitemShut {NoStop}%
\bibitem [{\citenamefont {Anisimov}\ \emph {et~al.}(2010)\citenamefont
  {Anisimov}, \citenamefont {Raterman}, \citenamefont {Chiruvelli},
  \citenamefont {Plick}, \citenamefont {Huver}, \citenamefont {Lee},\ and\
  \citenamefont {Dowling}}]{Anisimov2010}%
  \BibitemOpen
  \bibfield  {author} {\bibinfo {author} {\bibfnamefont {P.~M.}\ \bibnamefont
  {Anisimov}}, \bibinfo {author} {\bibfnamefont {G.~M.}\ \bibnamefont
  {Raterman}}, \bibinfo {author} {\bibfnamefont {A.}~\bibnamefont
  {Chiruvelli}}, \bibinfo {author} {\bibfnamefont {W.~N.}\ \bibnamefont
  {Plick}}, \bibinfo {author} {\bibfnamefont {S.~D.}\ \bibnamefont {Huver}},
  \bibinfo {author} {\bibfnamefont {H.}~\bibnamefont {Lee}}, \ and\ \bibinfo
  {author} {\bibfnamefont {J.~P.}\ \bibnamefont {Dowling}},\ }\href@noop {}
  {\bibfield  {journal} {\bibinfo  {journal} {Phys. Rev. Lett.}\ }\textbf
  {\bibinfo {volume} {104}},\ \bibinfo {pages} {103602} (\bibinfo {year}
  {2010})}\BibitemShut {NoStop}%
\bibitem [{\citenamefont {Rivas}\ and\ \citenamefont {Luis}(2012)}]{Rivas2012}%
  \BibitemOpen
  \bibfield  {author} {\bibinfo {author} {\bibfnamefont {{\'{A}}.}~\bibnamefont
  {Rivas}}\ and\ \bibinfo {author} {\bibfnamefont {A.}~\bibnamefont {Luis}},\
  }\href@noop {} {\bibfield  {journal} {\bibinfo  {journal} {New J. Phys.}\
  }\textbf {\bibinfo {volume} {14}},\ \bibinfo {pages} {093052} (\bibinfo
  {year} {2012})}\BibitemShut {NoStop}%
\bibitem [{\citenamefont {Zhang}\ \emph {et~al.}(2012)\citenamefont {Zhang},
  \citenamefont {Jin}, \citenamefont {Cao}, \citenamefont {Liu},\ and\
  \citenamefont {Fan}}]{Zhang2012}%
  \BibitemOpen
  \bibfield  {author} {\bibinfo {author} {\bibfnamefont {Y.~R.}\ \bibnamefont
  {Zhang}}, \bibinfo {author} {\bibfnamefont {G.~R.}\ \bibnamefont {Jin}},
  \bibinfo {author} {\bibfnamefont {J.~P.}\ \bibnamefont {Cao}}, \bibinfo
  {author} {\bibfnamefont {W.~M.}\ \bibnamefont {Liu}}, \ and\ \bibinfo
  {author} {\bibfnamefont {H.}~\bibnamefont {Fan}},\ }\href@noop {} {\bibfield
  {journal} {\bibinfo  {journal} {J. Phys. A: Math. Theor.}\ }\textbf {\bibinfo
  {volume} {46}},\ \bibinfo {pages} {035302} (\bibinfo {year}
  {2012})}\BibitemShut {NoStop}%
\bibitem [{\citenamefont {Berry}\ \emph {et~al.}(2012)\citenamefont {Berry},
  \citenamefont {Hall}, \citenamefont {Zwierz},\ and\ \citenamefont
  {Wiseman}}]{Berry2012}%
  \BibitemOpen
  \bibfield  {author} {\bibinfo {author} {\bibfnamefont {D.~W.}\ \bibnamefont
  {Berry}}, \bibinfo {author} {\bibfnamefont {M.~J.~W.}\ \bibnamefont {Hall}},
  \bibinfo {author} {\bibfnamefont {M.}~\bibnamefont {Zwierz}}, \ and\ \bibinfo
  {author} {\bibfnamefont {H.~M.}\ \bibnamefont {Wiseman}},\ }\href@noop {}
  {\bibfield  {journal} {\bibinfo  {journal} {Phys. Rev. A}\ }\textbf {\bibinfo
  {volume} {86}},\ \bibinfo {pages} {053813} (\bibinfo {year}
  {2012})}\BibitemShut {NoStop}%
\bibitem [{\citenamefont {Giovannetti}\ and\ \citenamefont
  {Maccone}(2012)}]{Giovannetti2012}%
  \BibitemOpen
  \bibfield  {author} {\bibinfo {author} {\bibfnamefont {V.}~\bibnamefont
  {Giovannetti}}\ and\ \bibinfo {author} {\bibfnamefont {L.}~\bibnamefont
  {Maccone}},\ }\href@noop {} {\bibfield  {journal} {\bibinfo  {journal} {Phys.
  Rev. Lett.}\ }\textbf {\bibinfo {volume} {108}},\ \bibinfo {pages} {210404}
  (\bibinfo {year} {2012})}\BibitemShut {NoStop}%
\bibitem [{\citenamefont {Hofmann}(2009)}]{Hofmann2009}%
  \BibitemOpen
  \bibfield  {author} {\bibinfo {author} {\bibfnamefont {H.~F.}\ \bibnamefont
  {Hofmann}},\ }\href@noop {} {\bibfield  {journal} {\bibinfo  {journal} {Phys.
  Rev. A}\ }\textbf {\bibinfo {volume} {79}},\ \bibinfo {pages} {033822}
  (\bibinfo {year} {2009})}\BibitemShut {NoStop}%
\bibitem [{\citenamefont {Dowling}(2008)}]{Dowling2008}%
  \BibitemOpen
  \bibfield  {author} {\bibinfo {author} {\bibfnamefont {J.~P.}\ \bibnamefont
  {Dowling}},\ }\href@noop {} {\bibfield  {journal} {\bibinfo  {journal}
  {Contemp. Phys.}\ }\textbf {\bibinfo {volume} {49}},\ \bibinfo {pages} {125}
  (\bibinfo {year} {2008})}\BibitemShut {NoStop}%
\bibitem [{\citenamefont {Yurke}, \citenamefont {McCall},\ and\ \citenamefont
  {Klauder}(1986)}]{Yurke1986a}%
  \BibitemOpen
  \bibfield  {author} {\bibinfo {author} {\bibfnamefont {B.}~\bibnamefont
  {Yurke}}, \bibinfo {author} {\bibfnamefont {S.~L.}\ \bibnamefont {McCall}}, \
  and\ \bibinfo {author} {\bibfnamefont {J.~R.}\ \bibnamefont {Klauder}},\
  }\href@noop {} {\bibfield  {journal} {\bibinfo  {journal} {Phys. Rev. A}\
  }\textbf {\bibinfo {volume} {33}},\ \bibinfo {pages} {4033} (\bibinfo {year}
  {1986})}\BibitemShut {NoStop}%
\bibitem [{\citenamefont {Schwinger}(1965)}]{Schwinger1965}%
  \BibitemOpen
  \bibfield  {author} {\bibinfo {author} {\bibfnamefont {J.}~\bibnamefont
  {Schwinger}},\ }\href@noop {} {\emph {\bibinfo {title} {“Quantum theory of
  angular momentum”}}}\ (\bibinfo  {publisher} {Academic Press},\ \bibinfo
  {address} {New York},\ \bibinfo {year} {1965})\BibitemShut {NoStop}%
\bibitem [{\citenamefont {Caves}(1981)}]{Caves1981}%
  \BibitemOpen
  \bibfield  {author} {\bibinfo {author} {\bibfnamefont {C.~M.}\ \bibnamefont
  {Caves}},\ }\href@noop {} {\bibfield  {journal} {\bibinfo  {journal} {Phys.
  Rev. D}\ }\textbf {\bibinfo {volume} {23}},\ \bibinfo {pages} {1693}
  (\bibinfo {year} {1981})}\BibitemShut {NoStop}%
\bibitem [{\citenamefont {Pezz\'e}\ and\ \citenamefont
  {Smerzi}(2008)}]{Pezze2008}%
  \BibitemOpen
  \bibfield  {author} {\bibinfo {author} {\bibfnamefont {L.}~\bibnamefont
  {Pezz\'e}}\ and\ \bibinfo {author} {\bibfnamefont {A.}~\bibnamefont
  {Smerzi}},\ }\href@noop {} {\bibfield  {journal} {\bibinfo  {journal} {Phys.
  Rev. Lett.}\ }\textbf {\bibinfo {volume} {100}},\ \bibinfo {pages} {073601}
  (\bibinfo {year} {2008})}\BibitemShut {NoStop}%
\bibitem [{\citenamefont {Seshadreesan}\ \emph {et~al.}(2011)\citenamefont
  {Seshadreesan}, \citenamefont {Anisimov}, \citenamefont {Lee},\ and\
  \citenamefont {Dowling}}]{Seshadreesan2011}%
  \BibitemOpen
  \bibfield  {author} {\bibinfo {author} {\bibfnamefont {K.~P.}\ \bibnamefont
  {Seshadreesan}}, \bibinfo {author} {\bibfnamefont {P.~M.}\ \bibnamefont
  {Anisimov}}, \bibinfo {author} {\bibfnamefont {H.}~\bibnamefont {Lee}}, \
  and\ \bibinfo {author} {\bibfnamefont {J.~P.}\ \bibnamefont {Dowling}},\
  }\href@noop {} {\bibfield  {journal} {\bibinfo  {journal} {New J. Phys.}\
  }\textbf {\bibinfo {volume} {13}},\ \bibinfo {pages} {083026} (\bibinfo
  {year} {2011})}\BibitemShut {NoStop}%
\bibitem [{\citenamefont {d'Ariano}, \citenamefont {Macchiavello},\ and\
  \citenamefont {Paris}(1995)}]{dAriano1995}%
  \BibitemOpen
  \bibfield  {author} {\bibinfo {author} {\bibfnamefont {G.}~\bibnamefont
  {d'Ariano}}, \bibinfo {author} {\bibfnamefont {C.}~\bibnamefont
  {Macchiavello}}, \ and\ \bibinfo {author} {\bibfnamefont {M.~G.}\
  \bibnamefont {Paris}},\ }\href@noop {} {\bibfield  {journal} {\bibinfo
  {journal} {Physics Letters A}\ }\textbf {\bibinfo {volume} {198}},\ \bibinfo
  {pages} {286} (\bibinfo {year} {1995})}\BibitemShut {NoStop}%
\bibitem [{\citenamefont {Oh}\ \emph {et~al.}(2017)\citenamefont {Oh},
  \citenamefont {Lee}, \citenamefont {Nha},\ and\ \citenamefont
  {Jeong}}]{Oh2017}%
  \BibitemOpen
  \bibfield  {author} {\bibinfo {author} {\bibfnamefont {C.}~\bibnamefont
  {Oh}}, \bibinfo {author} {\bibfnamefont {S.-Y.}\ \bibnamefont {Lee}},
  \bibinfo {author} {\bibfnamefont {H.}~\bibnamefont {Nha}}, \ and\ \bibinfo
  {author} {\bibfnamefont {H.}~\bibnamefont {Jeong}},\ }\href@noop {}
  {\bibfield  {journal} {\bibinfo  {journal} {Phys. Rev. A}\ }\textbf {\bibinfo
  {volume} {96}},\ \bibinfo {pages} {062304} (\bibinfo {year}
  {2017})}\BibitemShut {NoStop}%
\bibitem [{\citenamefont {Oh}\ \emph {et~al.}(2019)\citenamefont {Oh},
  \citenamefont {Lee}, \citenamefont {Rockstuhl}, \citenamefont {Jeong},
  \citenamefont {Kim}, \citenamefont {Nha},\ and\ \citenamefont
  {Lee}}]{Oh2019}%
  \BibitemOpen
  \bibfield  {author} {\bibinfo {author} {\bibfnamefont {C.}~\bibnamefont
  {Oh}}, \bibinfo {author} {\bibfnamefont {C.}~\bibnamefont {Lee}}, \bibinfo
  {author} {\bibfnamefont {C.}~\bibnamefont {Rockstuhl}}, \bibinfo {author}
  {\bibfnamefont {H.}~\bibnamefont {Jeong}}, \bibinfo {author} {\bibfnamefont
  {J.}~\bibnamefont {Kim}}, \bibinfo {author} {\bibfnamefont {H.}~\bibnamefont
  {Nha}}, \ and\ \bibinfo {author} {\bibfnamefont {S.-Y.}\ \bibnamefont
  {Lee}},\ }\href@noop {} {\bibfield  {journal} {\bibinfo  {journal} {npj
  Quantum Inf.}\ }\textbf {\bibinfo {volume} {5}},\ \bibinfo {pages} {10}
  (\bibinfo {year} {2019})}\BibitemShut {NoStop}%
\bibitem [{\citenamefont {Caves}(1980)}]{Caves1980}%
  \BibitemOpen
  \bibfield  {author} {\bibinfo {author} {\bibfnamefont {C.~M.}\ \bibnamefont
  {Caves}},\ }\href@noop {} {\bibfield  {journal} {\bibinfo  {journal} {Phys.
  Rev. Lett.}\ }\textbf {\bibinfo {volume} {45}},\ \bibinfo {pages} {75}
  (\bibinfo {year} {1980})}\BibitemShut {NoStop}%
\bibitem [{\citenamefont {Lang}\ and\ \citenamefont {Caves}(2013)}]{Lang2013}%
  \BibitemOpen
  \bibfield  {author} {\bibinfo {author} {\bibfnamefont {M.~D.}\ \bibnamefont
  {Lang}}\ and\ \bibinfo {author} {\bibfnamefont {C.~M.}\ \bibnamefont
  {Caves}},\ }\href@noop {} {\bibfield  {journal} {\bibinfo  {journal} {Phys.
  Rev. Lett.}\ }\textbf {\bibinfo {volume} {111}},\ \bibinfo {pages} {173601}
  (\bibinfo {year} {2013})}\BibitemShut {NoStop}%
\bibitem [{\citenamefont {Takeoka}\ \emph {et~al.}(2017)\citenamefont
  {Takeoka}, \citenamefont {Seshadreesan}, \citenamefont {You}, \citenamefont
  {Izumi},\ and\ \citenamefont {Dowling}}]{Takeoka2017}%
  \BibitemOpen
  \bibfield  {author} {\bibinfo {author} {\bibfnamefont {M.}~\bibnamefont
  {Takeoka}}, \bibinfo {author} {\bibfnamefont {K.~P.}\ \bibnamefont
  {Seshadreesan}}, \bibinfo {author} {\bibfnamefont {C.}~\bibnamefont {You}},
  \bibinfo {author} {\bibfnamefont {S.}~\bibnamefont {Izumi}}, \ and\ \bibinfo
  {author} {\bibfnamefont {J.~P.}\ \bibnamefont {Dowling}},\ }\href@noop {}
  {\bibfield  {journal} {\bibinfo  {journal} {Phys. Rev. A}\ }\textbf {\bibinfo
  {volume} {96}},\ \bibinfo {pages} {052118} (\bibinfo {year}
  {2017})}\BibitemShut {NoStop}%
\bibitem [{\citenamefont {Schnabel}\ \emph {et~al.}(2010)\citenamefont
  {Schnabel}, \citenamefont {Mavalvala}, \citenamefont {McClelland},\ and\
  \citenamefont {Lam}}]{Schnabel2010}%
  \BibitemOpen
  \bibfield  {author} {\bibinfo {author} {\bibfnamefont {R.}~\bibnamefont
  {Schnabel}}, \bibinfo {author} {\bibfnamefont {N.}~\bibnamefont {Mavalvala}},
  \bibinfo {author} {\bibfnamefont {D.~E.}\ \bibnamefont {McClelland}}, \ and\
  \bibinfo {author} {\bibfnamefont {P.~K.}\ \bibnamefont {Lam}},\ }\href@noop
  {} {\bibfield  {journal} {\bibinfo  {journal} {Nature Comm.}\ }\textbf
  {\bibinfo {volume} {1}},\ \bibinfo {pages} {121} (\bibinfo {year}
  {2010})}\BibitemShut {NoStop}%
\bibitem [{\citenamefont {Collaboration}(2011{\natexlab{a}})}]{LIGO2011}%
  \BibitemOpen
  \bibfield  {author} {\bibinfo {author} {\bibfnamefont {L.}~\bibnamefont
  {Collaboration}},\ }\href@noop {} {\bibfield  {journal} {\bibinfo  {journal}
  {Nat. Phys.}\ }\textbf {\bibinfo {volume} {7}},\ \bibinfo {pages} {962}
  (\bibinfo {year} {2011}{\natexlab{a}})}\BibitemShut {NoStop}%
\bibitem [{\citenamefont {Collaboration}(2011{\natexlab{b}})}]{LIGO2013}%
  \BibitemOpen
  \bibfield  {author} {\bibinfo {author} {\bibfnamefont {L.}~\bibnamefont
  {Collaboration}},\ }\href@noop {} {\bibfield  {journal} {\bibinfo  {journal}
  {Nature Photon.}\ }\textbf {\bibinfo {volume} {7}},\ \bibinfo {pages} {613}
  (\bibinfo {year} {2011}{\natexlab{b}})}\BibitemShut {NoStop}%
\bibitem [{\citenamefont {Bondurant}\ and\ \citenamefont
  {Shapiro}(1984)}]{Bondurant1984}%
  \BibitemOpen
  \bibfield  {author} {\bibinfo {author} {\bibfnamefont {R.~S.}\ \bibnamefont
  {Bondurant}}\ and\ \bibinfo {author} {\bibfnamefont {J.~H.}\ \bibnamefont
  {Shapiro}},\ }\href@noop {} {\bibfield  {journal} {\bibinfo  {journal} {Phys.
  Rev. D}\ }\textbf {\bibinfo {volume} {30}},\ \bibinfo {pages} {2548}
  (\bibinfo {year} {1984})}\BibitemShut {NoStop}%
\bibitem [{\citenamefont {Joo}, \citenamefont {Munro},\ and\ \citenamefont
  {Spiller}(2011)}]{Joo2011}%
  \BibitemOpen
  \bibfield  {author} {\bibinfo {author} {\bibfnamefont {J.}~\bibnamefont
  {Joo}}, \bibinfo {author} {\bibfnamefont {W.~J.}\ \bibnamefont {Munro}}, \
  and\ \bibinfo {author} {\bibfnamefont {T.~P.}\ \bibnamefont {Spiller}},\
  }\href@noop {} {\bibfield  {journal} {\bibinfo  {journal} {Phys. Rev. Lett.}\
  }\textbf {\bibinfo {volume} {107}},\ \bibinfo {pages} {083601} (\bibinfo
  {year} {2011})}\BibitemShut {NoStop}%
\bibitem [{\citenamefont {Joo}\ \emph {et~al.}(2012)\citenamefont {Joo},
  \citenamefont {Park}, \citenamefont {Jeong}, \citenamefont {Munro},
  \citenamefont {Nemoto},\ and\ \citenamefont {Spiller}}]{Joo2012}%
  \BibitemOpen
  \bibfield  {author} {\bibinfo {author} {\bibfnamefont {J.}~\bibnamefont
  {Joo}}, \bibinfo {author} {\bibfnamefont {K.}~\bibnamefont {Park}}, \bibinfo
  {author} {\bibfnamefont {H.}~\bibnamefont {Jeong}}, \bibinfo {author}
  {\bibfnamefont {W.~J.}\ \bibnamefont {Munro}}, \bibinfo {author}
  {\bibfnamefont {K.}~\bibnamefont {Nemoto}}, \ and\ \bibinfo {author}
  {\bibfnamefont {T.~P.}\ \bibnamefont {Spiller}},\ }\href@noop {} {\bibfield
  {journal} {\bibinfo  {journal} {Phys. Rev. A}\ }\textbf {\bibinfo {volume}
  {86}},\ \bibinfo {pages} {043828} (\bibinfo {year} {2012})}\BibitemShut
  {NoStop}%
\bibitem [{\citenamefont {Holland}\ and\ \citenamefont
  {Burnett}(1993)}]{Holland1993}%
  \BibitemOpen
  \bibfield  {author} {\bibinfo {author} {\bibfnamefont {M.~J.}\ \bibnamefont
  {Holland}}\ and\ \bibinfo {author} {\bibfnamefont {K.}~\bibnamefont
  {Burnett}},\ }\href@noop {} {\bibfield  {journal} {\bibinfo  {journal} {Phys.
  Rev. Lett.}\ }\textbf {\bibinfo {volume} {71}},\ \bibinfo {pages} {1355}
  (\bibinfo {year} {1993})}\BibitemShut {NoStop}%
\bibitem [{\citenamefont {Sanders}\ and\ \citenamefont
  {Milburn}(1995)}]{Sanders1995}%
  \BibitemOpen
  \bibfield  {author} {\bibinfo {author} {\bibfnamefont {B.~C.}\ \bibnamefont
  {Sanders}}\ and\ \bibinfo {author} {\bibfnamefont {G.~J.}\ \bibnamefont
  {Milburn}},\ }\href@noop {} {\bibfield  {journal} {\bibinfo  {journal} {Phys.
  Rev. Lett.}\ }\textbf {\bibinfo {volume} {75}},\ \bibinfo {pages} {2944}
  (\bibinfo {year} {1995})}\BibitemShut {NoStop}%
\bibitem [{\citenamefont {Munro}\ \emph {et~al.}(2002)\citenamefont {Munro},
  \citenamefont {Nemoto}, \citenamefont {Milburn},\ and\ \citenamefont
  {Braunstein}}]{Munro2002}%
  \BibitemOpen
  \bibfield  {author} {\bibinfo {author} {\bibfnamefont {W.~J.}\ \bibnamefont
  {Munro}}, \bibinfo {author} {\bibfnamefont {K.}~\bibnamefont {Nemoto}},
  \bibinfo {author} {\bibfnamefont {G.~J.}\ \bibnamefont {Milburn}}, \ and\
  \bibinfo {author} {\bibfnamefont {S.~L.}\ \bibnamefont {Braunstein}},\
  }\href@noop {} {\bibfield  {journal} {\bibinfo  {journal} {Phys. Rev. A}\
  }\textbf {\bibinfo {volume} {66}},\ \bibinfo {pages} {023819} (\bibinfo
  {year} {2002})}\BibitemShut {NoStop}%
\bibitem [{\citenamefont {Lloyd}(2008)}]{Lloyd2008}%
  \BibitemOpen
  \bibfield  {author} {\bibinfo {author} {\bibfnamefont {S.}~\bibnamefont
  {Lloyd}},\ }\href@noop {} {\bibfield  {journal} {\bibinfo  {journal}
  {Science}\ }\textbf {\bibinfo {volume} {321}},\ \bibinfo {pages} {1463}
  (\bibinfo {year} {2008})}\BibitemShut {NoStop}%
\bibitem [{\citenamefont {Tan}\ \emph {et~al.}(2008)\citenamefont {Tan},
  \citenamefont {Erkmen}, \citenamefont {Giovannetti}, \citenamefont {Guha},
  \citenamefont {Lloyd}, \citenamefont {Maccone}, \citenamefont {Pirandola},\
  and\ \citenamefont {Shapiro}}]{Tan2008}%
  \BibitemOpen
  \bibfield  {author} {\bibinfo {author} {\bibfnamefont {S.-H.}\ \bibnamefont
  {Tan}}, \bibinfo {author} {\bibfnamefont {B.~I.}\ \bibnamefont {Erkmen}},
  \bibinfo {author} {\bibfnamefont {V.}~\bibnamefont {Giovannetti}}, \bibinfo
  {author} {\bibfnamefont {S.}~\bibnamefont {Guha}}, \bibinfo {author}
  {\bibfnamefont {S.}~\bibnamefont {Lloyd}}, \bibinfo {author} {\bibfnamefont
  {L.}~\bibnamefont {Maccone}}, \bibinfo {author} {\bibfnamefont
  {S.}~\bibnamefont {Pirandola}}, \ and\ \bibinfo {author} {\bibfnamefont
  {J.~H.}\ \bibnamefont {Shapiro}},\ }\href@noop {} {\bibfield  {journal}
  {\bibinfo  {journal} {Phys. Rev. Lett.}\ }\textbf {\bibinfo {volume} {101}},\
  \bibinfo {pages} {253601} (\bibinfo {year} {2008})}\BibitemShut {NoStop}%
\bibitem [{\citenamefont {Shapiro}\ and\ \citenamefont
  {Lloyd}(2009)}]{Shapiro2009}%
  \BibitemOpen
  \bibfield  {author} {\bibinfo {author} {\bibfnamefont {J.~H.}\ \bibnamefont
  {Shapiro}}\ and\ \bibinfo {author} {\bibfnamefont {S.}~\bibnamefont
  {Lloyd}},\ }\href@noop {} {\bibfield  {journal} {\bibinfo  {journal} {New J.
  Phys.}\ }\textbf {\bibinfo {volume} {11}},\ \bibinfo {pages} {063045}
  (\bibinfo {year} {2009})}\BibitemShut {NoStop}%
\bibitem [{\citenamefont {Lopaeva}\ \emph {et~al.}(2013)\citenamefont
  {Lopaeva}, \citenamefont {Ruo~Berchera}, \citenamefont {Degiovanni},
  \citenamefont {Olivares}, \citenamefont {Brida},\ and\ \citenamefont
  {Genovese}}]{Lopaeva2013}%
  \BibitemOpen
  \bibfield  {author} {\bibinfo {author} {\bibfnamefont {E.~D.}\ \bibnamefont
  {Lopaeva}}, \bibinfo {author} {\bibfnamefont {I.}~\bibnamefont
  {Ruo~Berchera}}, \bibinfo {author} {\bibfnamefont {I.~P.}\ \bibnamefont
  {Degiovanni}}, \bibinfo {author} {\bibfnamefont {S.}~\bibnamefont
  {Olivares}}, \bibinfo {author} {\bibfnamefont {G.}~\bibnamefont {Brida}}, \
  and\ \bibinfo {author} {\bibfnamefont {M.}~\bibnamefont {Genovese}},\
  }\href@noop {} {\bibfield  {journal} {\bibinfo  {journal} {Phys. Rev. Lett.}\
  }\textbf {\bibinfo {volume} {110}},\ \bibinfo {pages} {153603} (\bibinfo
  {year} {2013})}\BibitemShut {NoStop}%
\bibitem [{\citenamefont {Lopaeva}\ \emph {et~al.}(2014)\citenamefont
  {Lopaeva}, \citenamefont {Berchera}, \citenamefont {Olivares}, \citenamefont
  {Brida}, \citenamefont {Degiovanni},\ and\ \citenamefont
  {Genovese}}]{Lopaeva2014}%
  \BibitemOpen
  \bibfield  {author} {\bibinfo {author} {\bibfnamefont {E.~D.}\ \bibnamefont
  {Lopaeva}}, \bibinfo {author} {\bibfnamefont {I.~R.}\ \bibnamefont
  {Berchera}}, \bibinfo {author} {\bibfnamefont {S.}~\bibnamefont {Olivares}},
  \bibinfo {author} {\bibfnamefont {G.}~\bibnamefont {Brida}}, \bibinfo
  {author} {\bibfnamefont {I.~P.}\ \bibnamefont {Degiovanni}}, \ and\ \bibinfo
  {author} {\bibfnamefont {M.}~\bibnamefont {Genovese}},\ }\href@noop {}
  {\bibfield  {journal} {\bibinfo  {journal} {Phys. Scr.}\ }\textbf {\bibinfo
  {volume} {T160}},\ \bibinfo {pages} {014026} (\bibinfo {year}
  {2014})}\BibitemShut {NoStop}%
\bibitem [{\citenamefont {Zhang}\ \emph {et~al.}(2015)\citenamefont {Zhang},
  \citenamefont {Mouradian}, \citenamefont {Wong},\ and\ \citenamefont
  {Shapiro}}]{Zhang2015}%
  \BibitemOpen
  \bibfield  {author} {\bibinfo {author} {\bibfnamefont {Z.}~\bibnamefont
  {Zhang}}, \bibinfo {author} {\bibfnamefont {S.}~\bibnamefont {Mouradian}},
  \bibinfo {author} {\bibfnamefont {F.~N.~C.}\ \bibnamefont {Wong}}, \ and\
  \bibinfo {author} {\bibfnamefont {J.~H.}\ \bibnamefont {Shapiro}},\
  }\href@noop {} {\bibfield  {journal} {\bibinfo  {journal} {Phys. Rev. Lett.}\
  }\textbf {\bibinfo {volume} {114}},\ \bibinfo {pages} {110506} (\bibinfo
  {year} {2015})}\BibitemShut {NoStop}%
\bibitem [{\citenamefont {Barzanjeh}\ \emph {et~al.}(2015)\citenamefont
  {Barzanjeh}, \citenamefont {Guha}, \citenamefont {Weedbrook}, \citenamefont
  {Vitali}, \citenamefont {Shapiro},\ and\ \citenamefont
  {Pirandola}}]{Barzanjeh2015}%
  \BibitemOpen
  \bibfield  {author} {\bibinfo {author} {\bibfnamefont {S.}~\bibnamefont
  {Barzanjeh}}, \bibinfo {author} {\bibfnamefont {S.}~\bibnamefont {Guha}},
  \bibinfo {author} {\bibfnamefont {C.}~\bibnamefont {Weedbrook}}, \bibinfo
  {author} {\bibfnamefont {D.}~\bibnamefont {Vitali}}, \bibinfo {author}
  {\bibfnamefont {J.~H.}\ \bibnamefont {Shapiro}}, \ and\ \bibinfo {author}
  {\bibfnamefont {S.}~\bibnamefont {Pirandola}},\ }\href@noop {} {\bibfield
  {journal} {\bibinfo  {journal} {Phys. Rev. Lett.}\ }\textbf {\bibinfo
  {volume} {114}},\ \bibinfo {pages} {080503} (\bibinfo {year}
  {2015})}\BibitemShut {NoStop}%
\bibitem [{\citenamefont {Shapiro}(2009)}]{Shapiro2009a}%
  \BibitemOpen
  \bibfield  {author} {\bibinfo {author} {\bibfnamefont {J.~H.}\ \bibnamefont
  {Shapiro}},\ }\href@noop {} {\bibfield  {journal} {\bibinfo  {journal} {Phys.
  Rev. A}\ }\textbf {\bibinfo {volume} {80}},\ \bibinfo {pages} {022320}
  (\bibinfo {year} {2009})}\BibitemShut {NoStop}%
\bibitem [{\citenamefont {Zhang}\ \emph {et~al.}(2013)\citenamefont {Zhang},
  \citenamefont {Tengner}, \citenamefont {Zhong}, \citenamefont {Wong},\ and\
  \citenamefont {Shapiro}}]{Zhang2013}%
  \BibitemOpen
  \bibfield  {author} {\bibinfo {author} {\bibfnamefont {Z.}~\bibnamefont
  {Zhang}}, \bibinfo {author} {\bibfnamefont {M.}~\bibnamefont {Tengner}},
  \bibinfo {author} {\bibfnamefont {T.}~\bibnamefont {Zhong}}, \bibinfo
  {author} {\bibfnamefont {F.~N.~C.}\ \bibnamefont {Wong}}, \ and\ \bibinfo
  {author} {\bibfnamefont {J.~H.}\ \bibnamefont {Shapiro}},\ }\href@noop {}
  {\bibfield  {journal} {\bibinfo  {journal} {Phys. Rev. Lett.}\ }\textbf
  {\bibinfo {volume} {111}},\ \bibinfo {pages} {010501} (\bibinfo {year}
  {2013})}\BibitemShut {NoStop}%
\bibitem [{\citenamefont {Sanz}\ \emph {et~al.}(2017)\citenamefont {Sanz},
  \citenamefont {Las~Heras}, \citenamefont {Garc\'{\i}a-Ripoll}, \citenamefont
  {Solano},\ and\ \citenamefont {Di~Candia}}]{Sanz2017}%
  \BibitemOpen
  \bibfield  {author} {\bibinfo {author} {\bibfnamefont {M.}~\bibnamefont
  {Sanz}}, \bibinfo {author} {\bibfnamefont {U.}~\bibnamefont {Las~Heras}},
  \bibinfo {author} {\bibfnamefont {J.~J.}\ \bibnamefont {Garc\'{\i}a-Ripoll}},
  \bibinfo {author} {\bibfnamefont {E.}~\bibnamefont {Solano}}, \ and\ \bibinfo
  {author} {\bibfnamefont {R.}~\bibnamefont {Di~Candia}},\ }\href@noop {}
  {\bibfield  {journal} {\bibinfo  {journal} {Phys. Rev. Lett.}\ }\textbf
  {\bibinfo {volume} {118}},\ \bibinfo {pages} {070803} (\bibinfo {year}
  {2017})}\BibitemShut {NoStop}%
\bibitem [{\citenamefont {Jakeman}\ and\ \citenamefont
  {Rarity}(1986)}]{Jakeman1986}%
  \BibitemOpen
  \bibfield  {author} {\bibinfo {author} {\bibfnamefont {E.}~\bibnamefont
  {Jakeman}}\ and\ \bibinfo {author} {\bibfnamefont {J.}~\bibnamefont
  {Rarity}},\ }\href@noop {} {\bibfield  {journal} {\bibinfo  {journal} {Opt.
  Comm.}\ }\textbf {\bibinfo {volume} {59}},\ \bibinfo {pages} {219} (\bibinfo
  {year} {1986})}\BibitemShut {NoStop}%
\bibitem [{\citenamefont {Sabines-Chesterking}\ \emph
  {et~al.}(2017)\citenamefont {Sabines-Chesterking}, \citenamefont {Whittaker},
  \citenamefont {Joshi}, \citenamefont {Birchall}, \citenamefont {Moreau},
  \citenamefont {McMillan}, \citenamefont {Cable}, \citenamefont {O'Brien},
  \citenamefont {Rarity},\ and\ \citenamefont
  {Matthews}}]{SabinesChesterking2019}%
  \BibitemOpen
  \bibfield  {author} {\bibinfo {author} {\bibfnamefont {J.}~\bibnamefont
  {Sabines-Chesterking}}, \bibinfo {author} {\bibfnamefont {R.}~\bibnamefont
  {Whittaker}}, \bibinfo {author} {\bibfnamefont {S.~K.}\ \bibnamefont
  {Joshi}}, \bibinfo {author} {\bibfnamefont {P.~M.}\ \bibnamefont {Birchall}},
  \bibinfo {author} {\bibfnamefont {P.~A.}\ \bibnamefont {Moreau}}, \bibinfo
  {author} {\bibfnamefont {A.}~\bibnamefont {McMillan}}, \bibinfo {author}
  {\bibfnamefont {H.~V.}\ \bibnamefont {Cable}}, \bibinfo {author}
  {\bibfnamefont {J.~L.}\ \bibnamefont {O'Brien}}, \bibinfo {author}
  {\bibfnamefont {J.~G.}\ \bibnamefont {Rarity}}, \ and\ \bibinfo {author}
  {\bibfnamefont {J.~C.~F.}\ \bibnamefont {Matthews}},\ }\href@noop {}
  {\bibfield  {journal} {\bibinfo  {journal} {Phys. Rev. Applied}\ }\textbf
  {\bibinfo {volume} {8}},\ \bibinfo {pages} {014016} (\bibinfo {year}
  {2017})}\BibitemShut {NoStop}%
\bibitem [{\citenamefont {Burnham}\ and\ \citenamefont
  {Weinberg}(1970)}]{Burnham1970}%
  \BibitemOpen
  \bibfield  {author} {\bibinfo {author} {\bibfnamefont {D.~C.}\ \bibnamefont
  {Burnham}}\ and\ \bibinfo {author} {\bibfnamefont {D.~L.}\ \bibnamefont
  {Weinberg}},\ }\href@noop {} {\bibfield  {journal} {\bibinfo  {journal}
  {Phys. Rev. Lett.}\ }\textbf {\bibinfo {volume} {25}},\ \bibinfo {pages} {84}
  (\bibinfo {year} {1970})}\BibitemShut {NoStop}%
\bibitem [{\citenamefont {Shen}(1984)}]{Shen1984}%
  \BibitemOpen
  \bibfield  {author} {\bibinfo {author} {\bibfnamefont {Y.~R.}\ \bibnamefont
  {Shen}},\ }\href@noop {} {\emph {\bibinfo {title} {The principles of
  nonlinear optics}}}\ (\bibinfo  {publisher} {Wiley},\ \bibinfo {address} {New
  York},\ \bibinfo {year} {1984})\BibitemShut {NoStop}%
\bibitem [{\citenamefont {Hong}, \citenamefont {Ou},\ and\ \citenamefont
  {Mandel}(1987)}]{Hong1987}%
  \BibitemOpen
  \bibfield  {author} {\bibinfo {author} {\bibfnamefont {C.~K.}\ \bibnamefont
  {Hong}}, \bibinfo {author} {\bibfnamefont {Z.~Y.}\ \bibnamefont {Ou}}, \ and\
  \bibinfo {author} {\bibfnamefont {L.}~\bibnamefont {Mandel}},\ }\href@noop {}
  {\bibfield  {journal} {\bibinfo  {journal} {Phys. Rev. Lett.}\ }\textbf
  {\bibinfo {volume} {59}},\ \bibinfo {pages} {2044} (\bibinfo {year}
  {1987})}\BibitemShut {NoStop}%
\bibitem [{\citenamefont {Shih}\ and\ \citenamefont {Alley}(1988)}]{Shih1988}%
  \BibitemOpen
  \bibfield  {author} {\bibinfo {author} {\bibfnamefont {Y.~H.}\ \bibnamefont
  {Shih}}\ and\ \bibinfo {author} {\bibfnamefont {C.~O.}\ \bibnamefont
  {Alley}},\ }\href@noop {} {\bibfield  {journal} {\bibinfo  {journal} {Phys.
  Rev. Lett.}\ }\textbf {\bibinfo {volume} {61}},\ \bibinfo {pages}
  {2921--2924} (\bibinfo {year} {1988})}\BibitemShut {NoStop}%
\bibitem [{\citenamefont {Shih}\ and\ \citenamefont
  {Sergienko}(1994)}]{Shih1994}%
  \BibitemOpen
  \bibfield  {author} {\bibinfo {author} {\bibfnamefont {Y.~H.}\ \bibnamefont
  {Shih}}\ and\ \bibinfo {author} {\bibfnamefont {A.~V.}\ \bibnamefont
  {Sergienko}},\ }\href@noop {} {\bibfield  {journal} {\bibinfo  {journal}
  {Phys. Rev. A}\ }\textbf {\bibinfo {volume} {50}},\ \bibinfo {pages} {2564}
  (\bibinfo {year} {1994})}\BibitemShut {NoStop}%
\bibitem [{\citenamefont {Kwiat}\ \emph {et~al.}(1995)\citenamefont {Kwiat},
  \citenamefont {Mattle}, \citenamefont {Weinfurter}, \citenamefont
  {Zeilinger}, \citenamefont {Sergienko},\ and\ \citenamefont
  {Shih}}]{Kwiat1995}%
  \BibitemOpen
  \bibfield  {author} {\bibinfo {author} {\bibfnamefont {P.~G.}\ \bibnamefont
  {Kwiat}}, \bibinfo {author} {\bibfnamefont {K.}~\bibnamefont {Mattle}},
  \bibinfo {author} {\bibfnamefont {H.}~\bibnamefont {Weinfurter}}, \bibinfo
  {author} {\bibfnamefont {A.}~\bibnamefont {Zeilinger}}, \bibinfo {author}
  {\bibfnamefont {A.~V.}\ \bibnamefont {Sergienko}}, \ and\ \bibinfo {author}
  {\bibfnamefont {Y.}~\bibnamefont {Shih}},\ }\href@noop {} {\bibfield
  {journal} {\bibinfo  {journal} {Phys. Rev. Lett.}\ }\textbf {\bibinfo
  {volume} {75}},\ \bibinfo {pages} {4337} (\bibinfo {year}
  {1995})}\BibitemShut {NoStop}%
\bibitem [{\citenamefont {Cooper}\ \emph {et~al.}(2013)\citenamefont {Cooper},
  \citenamefont {Wright}, \citenamefont {S\"{o}ller},\ and\ \citenamefont
  {Smith}}]{Cooper2013}%
  \BibitemOpen
  \bibfield  {author} {\bibinfo {author} {\bibfnamefont {M.}~\bibnamefont
  {Cooper}}, \bibinfo {author} {\bibfnamefont {L.~J.}\ \bibnamefont {Wright}},
  \bibinfo {author} {\bibfnamefont {C.}~\bibnamefont {S\"{o}ller}}, \ and\
  \bibinfo {author} {\bibfnamefont {B.~J.}\ \bibnamefont {Smith}},\ }\href@noop
  {} {\bibfield  {journal} {\bibinfo  {journal} {Opt. Express}\ }\textbf
  {\bibinfo {volume} {21}},\ \bibinfo {pages} {5309} (\bibinfo {year}
  {2013})}\BibitemShut {NoStop}%
\bibitem [{\citenamefont {Ohnesorge}\ \emph {et~al.}(1997)\citenamefont
  {Ohnesorge}, \citenamefont {Bayer}, \citenamefont {Forchel}, \citenamefont
  {Reithmaier}, \citenamefont {Gippius},\ and\ \citenamefont
  {Tikhodeev}}]{Ohnesorge1997}%
  \BibitemOpen
  \bibfield  {author} {\bibinfo {author} {\bibfnamefont {B.}~\bibnamefont
  {Ohnesorge}}, \bibinfo {author} {\bibfnamefont {M.}~\bibnamefont {Bayer}},
  \bibinfo {author} {\bibfnamefont {A.}~\bibnamefont {Forchel}}, \bibinfo
  {author} {\bibfnamefont {J.~P.}\ \bibnamefont {Reithmaier}}, \bibinfo
  {author} {\bibfnamefont {N.~A.}\ \bibnamefont {Gippius}}, \ and\ \bibinfo
  {author} {\bibfnamefont {S.~G.}\ \bibnamefont {Tikhodeev}},\ }\href@noop {}
  {\bibfield  {journal} {\bibinfo  {journal} {Phys. Rev. B}\ }\textbf {\bibinfo
  {volume} {56}},\ \bibinfo {pages} {R4367} (\bibinfo {year}
  {1997})}\BibitemShut {NoStop}%
\bibitem [{\citenamefont {G\'erard}\ \emph {et~al.}(1998)\citenamefont
  {G\'erard}, \citenamefont {Sermage}, \citenamefont {Gayral}, \citenamefont
  {Legrand}, \citenamefont {Costard},\ and\ \citenamefont
  {Thierry-Mieg}}]{Gerard1998}%
  \BibitemOpen
  \bibfield  {author} {\bibinfo {author} {\bibfnamefont {J.~M.}\ \bibnamefont
  {G\'erard}}, \bibinfo {author} {\bibfnamefont {B.}~\bibnamefont {Sermage}},
  \bibinfo {author} {\bibfnamefont {B.}~\bibnamefont {Gayral}}, \bibinfo
  {author} {\bibfnamefont {B.}~\bibnamefont {Legrand}}, \bibinfo {author}
  {\bibfnamefont {E.}~\bibnamefont {Costard}}, \ and\ \bibinfo {author}
  {\bibfnamefont {V.}~\bibnamefont {Thierry-Mieg}},\ }\href@noop {} {\bibfield
  {journal} {\bibinfo  {journal} {Phys. Rev. Lett.}\ }\textbf {\bibinfo
  {volume} {81}},\ \bibinfo {pages} {1110--1113} (\bibinfo {year}
  {1998})}\BibitemShut {NoStop}%
\bibitem [{\citenamefont {Kuhn}, \citenamefont {Hennrich},\ and\ \citenamefont
  {Rempe}(2002)}]{Kuhn2002}%
  \BibitemOpen
  \bibfield  {author} {\bibinfo {author} {\bibfnamefont {A.}~\bibnamefont
  {Kuhn}}, \bibinfo {author} {\bibfnamefont {M.}~\bibnamefont {Hennrich}}, \
  and\ \bibinfo {author} {\bibfnamefont {G.}~\bibnamefont {Rempe}},\
  }\href@noop {} {\bibfield  {journal} {\bibinfo  {journal} {Phys. Rev. Lett.}\
  }\textbf {\bibinfo {volume} {89}},\ \bibinfo {pages} {067901} (\bibinfo
  {year} {2002})}\BibitemShut {NoStop}%
\bibitem [{\citenamefont {McKeever}\ \emph {et~al.}(2004)\citenamefont
  {McKeever}, \citenamefont {Boca}, \citenamefont {Boozer}, \citenamefont
  {Miller}, \citenamefont {Buck}, \citenamefont {Kuzmich},\ and\ \citenamefont
  {Kimble}}]{McKeever1992}%
  \BibitemOpen
  \bibfield  {author} {\bibinfo {author} {\bibfnamefont {J.}~\bibnamefont
  {McKeever}}, \bibinfo {author} {\bibfnamefont {A.}~\bibnamefont {Boca}},
  \bibinfo {author} {\bibfnamefont {A.~D.}\ \bibnamefont {Boozer}}, \bibinfo
  {author} {\bibfnamefont {R.}~\bibnamefont {Miller}}, \bibinfo {author}
  {\bibfnamefont {J.~R.}\ \bibnamefont {Buck}}, \bibinfo {author}
  {\bibfnamefont {A.}~\bibnamefont {Kuzmich}}, \ and\ \bibinfo {author}
  {\bibfnamefont {H.~J.}\ \bibnamefont {Kimble}},\ }\href@noop {} {\bibfield
  {journal} {\bibinfo  {journal} {Science}\ }\textbf {\bibinfo {volume}
  {303}},\ \bibinfo {pages} {1992} (\bibinfo {year} {2004})}\BibitemShut
  {NoStop}%
\bibitem [{\citenamefont {Mitchell}, \citenamefont {Lundeen},\ and\
  \citenamefont {Steinberg}(2004)}]{Mitchell2004}%
  \BibitemOpen
  \bibfield  {author} {\bibinfo {author} {\bibfnamefont {M.~W.}\ \bibnamefont
  {Mitchell}}, \bibinfo {author} {\bibfnamefont {J.~S.}\ \bibnamefont
  {Lundeen}}, \ and\ \bibinfo {author} {\bibfnamefont {A.~M.}\ \bibnamefont
  {Steinberg}},\ }\href@noop {} {\bibfield  {journal} {\bibinfo  {journal}
  {Nature}\ }\textbf {\bibinfo {volume} {429}},\ \bibinfo {pages} {161}
  (\bibinfo {year} {2004})}\BibitemShut {NoStop}%
\bibitem [{\citenamefont {Walther}\ \emph {et~al.}(2004)\citenamefont
  {Walther}, \citenamefont {Pan}, \citenamefont {Aspelmeyer}, \citenamefont
  {Ursin}, \citenamefont {Gasparoni},\ and\ \citenamefont
  {Zeilinger}}]{Walther2004}%
  \BibitemOpen
  \bibfield  {author} {\bibinfo {author} {\bibfnamefont {P.}~\bibnamefont
  {Walther}}, \bibinfo {author} {\bibfnamefont {J.-W.}\ \bibnamefont {Pan}},
  \bibinfo {author} {\bibfnamefont {M.}~\bibnamefont {Aspelmeyer}}, \bibinfo
  {author} {\bibfnamefont {R.}~\bibnamefont {Ursin}}, \bibinfo {author}
  {\bibfnamefont {S.}~\bibnamefont {Gasparoni}}, \ and\ \bibinfo {author}
  {\bibfnamefont {A.}~\bibnamefont {Zeilinger}},\ }\href@noop {} {\bibfield
  {journal} {\bibinfo  {journal} {Nature}\ }\textbf {\bibinfo {volume} {429}},\
  \bibinfo {pages} {158} (\bibinfo {year} {2004})}\BibitemShut {NoStop}%
\bibitem [{\citenamefont {Afek}, \citenamefont {Ambar},\ and\ \citenamefont
  {Silberberg}(2010)}]{Afek2010}%
  \BibitemOpen
  \bibfield  {author} {\bibinfo {author} {\bibfnamefont {I.}~\bibnamefont
  {Afek}}, \bibinfo {author} {\bibfnamefont {O.}~\bibnamefont {Ambar}}, \ and\
  \bibinfo {author} {\bibfnamefont {Y.}~\bibnamefont {Silberberg}},\
  }\href@noop {} {\bibfield  {journal} {\bibinfo  {journal} {Science}\ }\textbf
  {\bibinfo {volume} {328}},\ \bibinfo {pages} {879} (\bibinfo {year}
  {2010})}\BibitemShut {NoStop}%
\bibitem [{\citenamefont {Wang}\ \emph {et~al.}(2011)\citenamefont {Wang},
  \citenamefont {Mariantoni}, \citenamefont {Bialczak}, \citenamefont
  {Lenander}, \citenamefont {Lucero}, \citenamefont {Neeley}, \citenamefont
  {O'Connell}, \citenamefont {Sank}, \citenamefont {Weides}, \citenamefont
  {Wenner}, \citenamefont {Yamamoto}, \citenamefont {Yin}, \citenamefont
  {Zhao}, \citenamefont {Martinis},\ and\ \citenamefont {Cleland}}]{Wang2011}%
  \BibitemOpen
  \bibfield  {author} {\bibinfo {author} {\bibfnamefont {H.}~\bibnamefont
  {Wang}}, \bibinfo {author} {\bibfnamefont {M.}~\bibnamefont {Mariantoni}},
  \bibinfo {author} {\bibfnamefont {R.~C.}\ \bibnamefont {Bialczak}}, \bibinfo
  {author} {\bibfnamefont {M.}~\bibnamefont {Lenander}}, \bibinfo {author}
  {\bibfnamefont {E.}~\bibnamefont {Lucero}}, \bibinfo {author} {\bibfnamefont
  {M.}~\bibnamefont {Neeley}}, \bibinfo {author} {\bibfnamefont {A.~D.}\
  \bibnamefont {O'Connell}}, \bibinfo {author} {\bibfnamefont {D.}~\bibnamefont
  {Sank}}, \bibinfo {author} {\bibfnamefont {M.}~\bibnamefont {Weides}},
  \bibinfo {author} {\bibfnamefont {J.}~\bibnamefont {Wenner}}, \bibinfo
  {author} {\bibfnamefont {T.}~\bibnamefont {Yamamoto}}, \bibinfo {author}
  {\bibfnamefont {Y.}~\bibnamefont {Yin}}, \bibinfo {author} {\bibfnamefont
  {J.}~\bibnamefont {Zhao}}, \bibinfo {author} {\bibfnamefont {J.~M.}\
  \bibnamefont {Martinis}}, \ and\ \bibinfo {author} {\bibfnamefont {A.~N.}\
  \bibnamefont {Cleland}},\ }\href@noop {} {\bibfield  {journal} {\bibinfo
  {journal} {Phys. Rev. Lett.}\ }\textbf {\bibinfo {volume} {106}},\ \bibinfo
  {pages} {060401} (\bibinfo {year} {2011})}\BibitemShut {NoStop}%
\bibitem [{\citenamefont {Sanders}(1992)}]{Sanders1992}%
  \BibitemOpen
  \bibfield  {author} {\bibinfo {author} {\bibfnamefont {B.~C.}\ \bibnamefont
  {Sanders}},\ }\href@noop {} {\bibfield  {journal} {\bibinfo  {journal} {Phys.
  Rev. A}\ }\textbf {\bibinfo {volume} {45}},\ \bibinfo {pages} {6811}
  (\bibinfo {year} {1992})}\BibitemShut {NoStop}%
\bibitem [{\citenamefont {Mecozzi}\ and\ \citenamefont
  {Tombesi}(1987)}]{Mecozzi1987}%
  \BibitemOpen
  \bibfield  {author} {\bibinfo {author} {\bibfnamefont {A.}~\bibnamefont
  {Mecozzi}}\ and\ \bibinfo {author} {\bibfnamefont {P.}~\bibnamefont
  {Tombesi}},\ }\href@noop {} {\bibfield  {journal} {\bibinfo  {journal} {Phys.
  Rev. Lett.}\ }\textbf {\bibinfo {volume} {58}},\ \bibinfo {pages} {1055}
  (\bibinfo {year} {1987})}\BibitemShut {NoStop}%
\bibitem [{\citenamefont {Gerry}(1999)}]{Gerry1999}%
  \BibitemOpen
  \bibfield  {author} {\bibinfo {author} {\bibfnamefont {C.~C.}\ \bibnamefont
  {Gerry}},\ }\href@noop {} {\bibfield  {journal} {\bibinfo  {journal} {Phys.
  Rev. A}\ }\textbf {\bibinfo {volume} {59}},\ \bibinfo {pages} {4095}
  (\bibinfo {year} {1999})}\BibitemShut {NoStop}%
\bibitem [{\citenamefont {Jeong}\ \emph {et~al.}(2004)\citenamefont {Jeong},
  \citenamefont {Kim}, \citenamefont {Ralph},\ and\ \citenamefont
  {Ham}}]{Jeong2004}%
  \BibitemOpen
  \bibfield  {author} {\bibinfo {author} {\bibfnamefont {H.}~\bibnamefont
  {Jeong}}, \bibinfo {author} {\bibfnamefont {M.~S.}\ \bibnamefont {Kim}},
  \bibinfo {author} {\bibfnamefont {T.~C.}\ \bibnamefont {Ralph}}, \ and\
  \bibinfo {author} {\bibfnamefont {B.~S.}\ \bibnamefont {Ham}},\ }\href@noop
  {} {\bibfield  {journal} {\bibinfo  {journal} {Phys. Rev. A}\ }\textbf
  {\bibinfo {volume} {70}},\ \bibinfo {pages} {061801(R)} (\bibinfo {year}
  {2004})}\BibitemShut {NoStop}%
\bibitem [{\citenamefont {Wenger}, \citenamefont {Tualle-Brouri},\ and\
  \citenamefont {Grangier}(2004)}]{Wenger2004}%
  \BibitemOpen
  \bibfield  {author} {\bibinfo {author} {\bibfnamefont {J.}~\bibnamefont
  {Wenger}}, \bibinfo {author} {\bibfnamefont {R.}~\bibnamefont
  {Tualle-Brouri}}, \ and\ \bibinfo {author} {\bibfnamefont {P.}~\bibnamefont
  {Grangier}},\ }\href@noop {} {\bibfield  {journal} {\bibinfo  {journal}
  {Phys. Rev. Lett.}\ }\textbf {\bibinfo {volume} {92}},\ \bibinfo {pages}
  {153601} (\bibinfo {year} {2004})}\BibitemShut {NoStop}%
\bibitem [{\citenamefont {Ourjoumtsev}\ \emph {et~al.}(2006)\citenamefont
  {Ourjoumtsev}, \citenamefont {Tualle-Brouri}, \citenamefont {Laurat},\ and\
  \citenamefont {Grangier}}]{Ourjoumtsev2006}%
  \BibitemOpen
  \bibfield  {author} {\bibinfo {author} {\bibfnamefont {A.}~\bibnamefont
  {Ourjoumtsev}}, \bibinfo {author} {\bibfnamefont {R.}~\bibnamefont
  {Tualle-Brouri}}, \bibinfo {author} {\bibfnamefont {J.}~\bibnamefont
  {Laurat}}, \ and\ \bibinfo {author} {\bibfnamefont {P.}~\bibnamefont
  {Grangier}},\ }\href@noop {} {\bibfield  {journal} {\bibinfo  {journal}
  {Science}\ }\textbf {\bibinfo {volume} {312}},\ \bibinfo {pages} {83}
  (\bibinfo {year} {2006})}\BibitemShut {NoStop}%
\bibitem [{\citenamefont {Neergaard-Nielsen}\ \emph {et~al.}(2006)\citenamefont
  {Neergaard-Nielsen}, \citenamefont {Nielsen}, \citenamefont {Hettich},
  \citenamefont {M\o{}lmer},\ and\ \citenamefont
  {Polzik}}]{NeergaardNielsen2006}%
  \BibitemOpen
  \bibfield  {author} {\bibinfo {author} {\bibfnamefont {J.~S.}\ \bibnamefont
  {Neergaard-Nielsen}}, \bibinfo {author} {\bibfnamefont {B.~M.}\ \bibnamefont
  {Nielsen}}, \bibinfo {author} {\bibfnamefont {C.}~\bibnamefont {Hettich}},
  \bibinfo {author} {\bibfnamefont {K.}~\bibnamefont {M\o{}lmer}}, \ and\
  \bibinfo {author} {\bibfnamefont {E.~S.}\ \bibnamefont {Polzik}},\
  }\href@noop {} {\bibfield  {journal} {\bibinfo  {journal} {Phys. Rev. Lett.}\
  }\textbf {\bibinfo {volume} {97}},\ \bibinfo {pages} {083604} (\bibinfo
  {year} {2006})}\BibitemShut {NoStop}%
\bibitem [{\citenamefont {Wakui}\ \emph {et~al.}(2007)\citenamefont {Wakui},
  \citenamefont {Takahashi}, \citenamefont {Furusawa},\ and\ \citenamefont
  {Sasaki}}]{Wakui2007}%
  \BibitemOpen
  \bibfield  {author} {\bibinfo {author} {\bibfnamefont {K.}~\bibnamefont
  {Wakui}}, \bibinfo {author} {\bibfnamefont {H.}~\bibnamefont {Takahashi}},
  \bibinfo {author} {\bibfnamefont {A.}~\bibnamefont {Furusawa}}, \ and\
  \bibinfo {author} {\bibfnamefont {M.}~\bibnamefont {Sasaki}},\ }\href@noop {}
  {\bibfield  {journal} {\bibinfo  {journal} {Opt. Express}\ }\textbf {\bibinfo
  {volume} {15}},\ \bibinfo {pages} {3568} (\bibinfo {year}
  {2007})}\BibitemShut {NoStop}%
\bibitem [{\citenamefont {Lund}\ \emph {et~al.}(2004)\citenamefont {Lund},
  \citenamefont {Jeong}, \citenamefont {Ralph},\ and\ \citenamefont
  {Kim}}]{Lund2004}%
  \BibitemOpen
  \bibfield  {author} {\bibinfo {author} {\bibfnamefont {A.~P.}\ \bibnamefont
  {Lund}}, \bibinfo {author} {\bibfnamefont {H.}~\bibnamefont {Jeong}},
  \bibinfo {author} {\bibfnamefont {T.~C.}\ \bibnamefont {Ralph}}, \ and\
  \bibinfo {author} {\bibfnamefont {M.~S.}\ \bibnamefont {Kim}},\ }\href@noop
  {} {\bibfield  {journal} {\bibinfo  {journal} {Phys. Rev. A}\ }\textbf
  {\bibinfo {volume} {70}},\ \bibinfo {pages} {020101(R)} (\bibinfo {year}
  {2004})}\BibitemShut {NoStop}%
\bibitem [{\citenamefont {Ourjoumtsev}\ \emph {et~al.}(2007)\citenamefont
  {Ourjoumtsev}, \citenamefont {Jeong}, \citenamefont {Tualle-Brouri},\ and\
  \citenamefont {Grangier}}]{Ourjoumtsev2007}%
  \BibitemOpen
  \bibfield  {author} {\bibinfo {author} {\bibfnamefont {A.}~\bibnamefont
  {Ourjoumtsev}}, \bibinfo {author} {\bibfnamefont {H.}~\bibnamefont {Jeong}},
  \bibinfo {author} {\bibfnamefont {R.}~\bibnamefont {Tualle-Brouri}}, \ and\
  \bibinfo {author} {\bibfnamefont {P.}~\bibnamefont {Grangier}},\ }\href@noop
  {} {\bibfield  {journal} {\bibinfo  {journal} {Nature}\ }\textbf {\bibinfo
  {volume} {448}},\ \bibinfo {pages} {784} (\bibinfo {year}
  {2007})}\BibitemShut {NoStop}%
\bibitem [{\citenamefont {Ourjoumtsev}\ \emph {et~al.}(2009)\citenamefont
  {Ourjoumtsev}, \citenamefont {Ferreyrol}, \citenamefont {Tualle-Brouri},\
  and\ \citenamefont {Grangier}}]{Ourjoumtsev2009}%
  \BibitemOpen
  \bibfield  {author} {\bibinfo {author} {\bibfnamefont {A.}~\bibnamefont
  {Ourjoumtsev}}, \bibinfo {author} {\bibfnamefont {F.}~\bibnamefont
  {Ferreyrol}}, \bibinfo {author} {\bibfnamefont {R.}~\bibnamefont
  {Tualle-Brouri}}, \ and\ \bibinfo {author} {\bibfnamefont {P.}~\bibnamefont
  {Grangier}},\ }\href@noop {} {\bibfield  {journal} {\bibinfo  {journal} {Nat.
  Phys.}\ }\textbf {\bibinfo {volume} {5}},\ \bibinfo {pages} {189} (\bibinfo
  {year} {2009})}\BibitemShut {NoStop}%
\bibitem [{\citenamefont {Vogel}, \citenamefont {Akulin},\ and\ \citenamefont
  {Schleich}(1993)}]{Vogel1993}%
  \BibitemOpen
  \bibfield  {author} {\bibinfo {author} {\bibfnamefont {K.}~\bibnamefont
  {Vogel}}, \bibinfo {author} {\bibfnamefont {V.~M.}\ \bibnamefont {Akulin}}, \
  and\ \bibinfo {author} {\bibfnamefont {W.~P.}\ \bibnamefont {Schleich}},\
  }\href@noop {} {\bibfield  {journal} {\bibinfo  {journal} {Phys. Rev. Lett.}\
  }\textbf {\bibinfo {volume} {71}},\ \bibinfo {pages} {1816} (\bibinfo {year}
  {1993})}\BibitemShut {NoStop}%
\bibitem [{\citenamefont {Hofheinz}\ \emph {et~al.}(2009)\citenamefont
  {Hofheinz}, \citenamefont {Wang}, \citenamefont {Ansmann}, \citenamefont
  {Bialczak}, \citenamefont {Lucero}, \citenamefont {Neeley}, \citenamefont
  {O'Connell}, \citenamefont {Sank}, \citenamefont {Wenner}, \citenamefont
  {Martinis},\ and\ \citenamefont {Cleland}}]{Hofheinz2009}%
  \BibitemOpen
  \bibfield  {author} {\bibinfo {author} {\bibfnamefont {M.}~\bibnamefont
  {Hofheinz}}, \bibinfo {author} {\bibfnamefont {H.}~\bibnamefont {Wang}},
  \bibinfo {author} {\bibfnamefont {M.}~\bibnamefont {Ansmann}}, \bibinfo
  {author} {\bibfnamefont {R.~C.}\ \bibnamefont {Bialczak}}, \bibinfo {author}
  {\bibfnamefont {E.}~\bibnamefont {Lucero}}, \bibinfo {author} {\bibfnamefont
  {M.}~\bibnamefont {Neeley}}, \bibinfo {author} {\bibfnamefont {A.~D.}\
  \bibnamefont {O'Connell}}, \bibinfo {author} {\bibfnamefont {D.}~\bibnamefont
  {Sank}}, \bibinfo {author} {\bibfnamefont {J.}~\bibnamefont {Wenner}},
  \bibinfo {author} {\bibfnamefont {J.~M.}\ \bibnamefont {Martinis}}, \ and\
  \bibinfo {author} {\bibfnamefont {A.~N.}\ \bibnamefont {Cleland}},\
  }\href@noop {} {\bibfield  {journal} {\bibinfo  {journal} {Nature}\ }\textbf
  {\bibinfo {volume} {459}},\ \bibinfo {pages} {546} (\bibinfo {year}
  {2009})}\BibitemShut {NoStop}%
\bibitem [{\citenamefont {Jeong}\ \emph {et~al.}(2014)\citenamefont {Jeong},
  \citenamefont {Zavatta}, \citenamefont {Kang}, \citenamefont {Lee},
  \citenamefont {Costanzo}, \citenamefont {Grandi}, \citenamefont {Ralph},\
  and\ \citenamefont {Bellini}}]{Jeong2014}%
  \BibitemOpen
  \bibfield  {author} {\bibinfo {author} {\bibfnamefont {H.}~\bibnamefont
  {Jeong}}, \bibinfo {author} {\bibfnamefont {A.}~\bibnamefont {Zavatta}},
  \bibinfo {author} {\bibfnamefont {M.}~\bibnamefont {Kang}}, \bibinfo {author}
  {\bibfnamefont {S.-W.}\ \bibnamefont {Lee}}, \bibinfo {author} {\bibfnamefont
  {L.~S.}\ \bibnamefont {Costanzo}}, \bibinfo {author} {\bibfnamefont
  {S.}~\bibnamefont {Grandi}}, \bibinfo {author} {\bibfnamefont {T.~C.}\
  \bibnamefont {Ralph}}, \ and\ \bibinfo {author} {\bibfnamefont
  {M.}~\bibnamefont {Bellini}},\ }\href@noop {} {\bibfield  {journal} {\bibinfo
   {journal} {Nat. Photonics}\ }\textbf {\bibinfo {volume} {8}},\ \bibinfo
  {pages} {564} (\bibinfo {year} {2014})}\BibitemShut {NoStop}%
\bibitem [{\citenamefont {Morin}\ \emph {et~al.}(2014)\citenamefont {Morin},
  \citenamefont {Huang}, \citenamefont {Liu}, \citenamefont {Le~Jeannic},
  \citenamefont {Fabre},\ and\ \citenamefont {Laurat}}]{Morin2014}%
  \BibitemOpen
  \bibfield  {author} {\bibinfo {author} {\bibfnamefont {O.}~\bibnamefont
  {Morin}}, \bibinfo {author} {\bibfnamefont {K.}~\bibnamefont {Huang}},
  \bibinfo {author} {\bibfnamefont {J.}~\bibnamefont {Liu}}, \bibinfo {author}
  {\bibfnamefont {H.}~\bibnamefont {Le~Jeannic}}, \bibinfo {author}
  {\bibfnamefont {C.}~\bibnamefont {Fabre}}, \ and\ \bibinfo {author}
  {\bibfnamefont {J.}~\bibnamefont {Laurat}},\ }\href@noop {} {\bibfield
  {journal} {\bibinfo  {journal} {Nat. Photonics}\ }\textbf {\bibinfo {volume}
  {8}},\ \bibinfo {pages} {570} (\bibinfo {year} {2014})}\BibitemShut {NoStop}%
\bibitem [{\citenamefont {Dodonov}(2002)}]{Dodonov2002}%
  \BibitemOpen
  \bibfield  {author} {\bibinfo {author} {\bibfnamefont {V.~V.}\ \bibnamefont
  {Dodonov}},\ }\href@noop {} {\bibfield  {journal} {\bibinfo  {journal} {J.
  Opt. B: Quantum S. O.}\ }\textbf {\bibinfo {volume} {4}},\ \bibinfo {pages}
  {R1} (\bibinfo {year} {2002})}\BibitemShut {NoStop}%
\bibitem [{\citenamefont {Davidovich}(1996)}]{Davidovich1996}%
  \BibitemOpen
  \bibfield  {author} {\bibinfo {author} {\bibfnamefont {L.}~\bibnamefont
  {Davidovich}},\ }\href@noop {} {\bibfield  {journal} {\bibinfo  {journal}
  {Rev. Mod. Phys.}\ }\textbf {\bibinfo {volume} {68}},\ \bibinfo {pages} {127}
  (\bibinfo {year} {1996})}\BibitemShut {NoStop}%
\bibitem [{\citenamefont {Sanders}(2012)}]{Sanders2012}%
  \BibitemOpen
  \bibfield  {author} {\bibinfo {author} {\bibfnamefont {B.~C.}\ \bibnamefont
  {Sanders}},\ }\href@noop {} {\bibfield  {journal} {\bibinfo  {journal} {J.
  Phys. A: Math. Theor.}\ }\textbf {\bibinfo {volume} {45}},\ \bibinfo {pages}
  {244002} (\bibinfo {year} {2012})}\BibitemShut {NoStop}%
\bibitem [{\citenamefont {Dell’Anno}, \citenamefont {Siena},\ and\
  \citenamefont {Illuminati}(2006)}]{DellAnno2006}%
  \BibitemOpen
  \bibfield  {author} {\bibinfo {author} {\bibfnamefont {F.}~\bibnamefont
  {Dell’Anno}}, \bibinfo {author} {\bibfnamefont {S.~D.}\ \bibnamefont
  {Siena}}, \ and\ \bibinfo {author} {\bibfnamefont {F.}~\bibnamefont
  {Illuminati}},\ }\href@noop {} {\bibfield  {journal} {\bibinfo  {journal}
  {Phys. Rep.}\ }\textbf {\bibinfo {volume} {428}},\ \bibinfo {pages} {53}
  (\bibinfo {year} {2006})}\BibitemShut {NoStop}%
\bibitem [{\citenamefont {Boyd}, \citenamefont {Lukishova},\ and\ \citenamefont
  {Zadkov}(2019)}]{Boyd2019}%
  \BibitemOpen
  \bibinfo {editor} {\bibfnamefont {R.~W.}\ \bibnamefont {Boyd}}, \bibinfo
  {editor} {\bibfnamefont {S.~G.}\ \bibnamefont {Lukishova}}, \ and\ \bibinfo
  {editor} {\bibfnamefont {V.~N.}\ \bibnamefont {Zadkov}},\ eds.,\ \href@noop
  {} {\emph {\bibinfo {title} {Quantum Photonics: Pioneering Advances and
  Emerging Applications}}}\ (\bibinfo  {publisher} {Springer},\ \bibinfo
  {address} {Cham},\ \bibinfo {year} {2019})\ Chap.~\bibinfo {chapter}
  {3}\BibitemShut {NoStop}%
\bibitem [{\citenamefont {Escher}, \citenamefont {de~Matos~Filho},\ and\
  \citenamefont {Davidovich}(2011)}]{Escher2011}%
  \BibitemOpen
  \bibfield  {author} {\bibinfo {author} {\bibfnamefont {B.~M.}\ \bibnamefont
  {Escher}}, \bibinfo {author} {\bibfnamefont {R.~L.}\ \bibnamefont
  {de~Matos~Filho}}, \ and\ \bibinfo {author} {\bibfnamefont {L.}~\bibnamefont
  {Davidovich}},\ }\href@noop {} {\bibfield  {journal} {\bibinfo  {journal}
  {Nat. Phys.}\ }\textbf {\bibinfo {volume} {7}},\ \bibinfo {pages} {406}
  (\bibinfo {year} {2011})}\BibitemShut {NoStop}%
\bibitem [{\citenamefont {Demkowicz-Dobrzański}, \citenamefont
  {Kołodyński},\ and\ \citenamefont {Gu{\c t}{\u a}}(2012)}]{Dobrzanski2012}%
  \BibitemOpen
  \bibfield  {author} {\bibinfo {author} {\bibfnamefont {R.}~\bibnamefont
  {Demkowicz-Dobrzański}}, \bibinfo {author} {\bibfnamefont {J.}~\bibnamefont
  {Kołodyński}}, \ and\ \bibinfo {author} {\bibfnamefont {M.}~\bibnamefont
  {Gu{\c t}{\u a}}},\ }\href@noop {} {\bibfield  {journal} {\bibinfo  {journal}
  {Nat. Commun.}\ }\textbf {\bibinfo {volume} {3}},\ \bibinfo {pages} {1063}
  (\bibinfo {year} {2012})}\BibitemShut {NoStop}%
\bibitem [{\citenamefont {Demkowicz-Dobrzański}, \citenamefont {Jarzyna},\
  and\ \citenamefont {Kołodyński}(2015)}]{Dobrzanski2015}%
  \BibitemOpen
  \bibfield  {author} {\bibinfo {author} {\bibfnamefont {R.}~\bibnamefont
  {Demkowicz-Dobrzański}}, \bibinfo {author} {\bibfnamefont {M.}~\bibnamefont
  {Jarzyna}}, \ and\ \bibinfo {author} {\bibfnamefont {J.}~\bibnamefont
  {Kołodyński}},\ }\href@noop {} {\ \bibinfo {series} {Prog. Opt.},\ \textbf
  {\bibinfo {volume} {60}},\ \bibinfo {pages} {345} (\bibinfo {year}
  {2015})}\BibitemShut {NoStop}%
\bibitem [{\citenamefont {Sidhu}\ and\ \citenamefont {Kok}(2019)}]{Sidhu2019}%
  \BibitemOpen
  \bibfield  {author} {\bibinfo {author} {\bibfnamefont {J.~S.}\ \bibnamefont
  {Sidhu}}\ and\ \bibinfo {author} {\bibfnamefont {P.}~\bibnamefont {Kok}},\
  }\href@noop {} {} (\bibinfo {year} {2019}),\ \Eprint
  {http://arxiv.org/abs/arXiv:1907.06628} {arXiv:1907.06628} \BibitemShut
  {NoStop}%
\bibitem [{\citenamefont {Jeong}, \citenamefont {Kang},\ and\ \citenamefont
  {Kwon}(2015)}]{Jeong2015}%
  \BibitemOpen
  \bibfield  {author} {\bibinfo {author} {\bibfnamefont {H.}~\bibnamefont
  {Jeong}}, \bibinfo {author} {\bibfnamefont {M.}~\bibnamefont {Kang}}, \ and\
  \bibinfo {author} {\bibfnamefont {H.}~\bibnamefont {Kwon}},\ }\href@noop {}
  {\bibfield  {journal} {\bibinfo  {journal} {Opt. Comm.}\ }\textbf {\bibinfo
  {volume} {337}},\ \bibinfo {pages} {12} (\bibinfo {year} {2015})}\BibitemShut
  {NoStop}%
\bibitem [{\citenamefont {Fr\"owis}\ \emph {et~al.}(2018)\citenamefont
  {Fr\"owis}, \citenamefont {Sekatski}, \citenamefont {D\"ur}, \citenamefont
  {Gisin},\ and\ \citenamefont {Sangouard}}]{Frowis2018}%
  \BibitemOpen
  \bibfield  {author} {\bibinfo {author} {\bibfnamefont {F.}~\bibnamefont
  {Fr\"owis}}, \bibinfo {author} {\bibfnamefont {P.}~\bibnamefont {Sekatski}},
  \bibinfo {author} {\bibfnamefont {W.}~\bibnamefont {D\"ur}}, \bibinfo
  {author} {\bibfnamefont {N.}~\bibnamefont {Gisin}}, \ and\ \bibinfo {author}
  {\bibfnamefont {N.}~\bibnamefont {Sangouard}},\ }\href@noop {} {\bibfield
  {journal} {\bibinfo  {journal} {Rev. Mod. Phys.}\ }\textbf {\bibinfo {volume}
  {90}},\ \bibinfo {pages} {025004} (\bibinfo {year} {2018})}\BibitemShut
  {NoStop}%
\end{thebibliography}%

\end{document}